\documentclass[twocolumn,hidelinks]{svjour3}
\pdfoutput=1
\usepackage[T1]{fontenc}

\usepackage{lipsum}
\usepackage{authblk}
\usepackage{orcidlink}
\usepackage{comment}

\usepackage{amssymb,amsmath,amsfonts}
\usepackage{stmaryrd}
\usepackage{mathtools}
\usepackage{galois}
\usepackage{xspace}
\usepackage{algorithmic, listings}
\usepackage{array}
\usepackage{enumitem}
\usepackage{longtable, multirow, multicol, makecell}
\usepackage{floatrow}

\usepackage{graphicx}
\usepackage{tikz,tikz-cd}
\definecolor{ForestGreen}{RGB}{34,139,34}
\usetikzlibrary{positioning, quotes, shapes.geometric, fit, calc, arrows}
\usetikzlibrary{decorations.pathmorphing}
\usetikzlibrary{arrows.meta,shapes,shadows}
\usetikzlibrary{backgrounds, patterns}
\usepackage{pgfplots}

\usepackage{subcaption}

\usepackage[misc,geometry]{ifsym}

\usepackage{pifont}
\usepackage[makeroom]{cancel}

\usepackage{dashbox}

\newcommand{\bibpath}{./bibtex/clip}

\newcommand{\rr}{\mathbb{R}}

\newcommand{\nn}{\mathbb{N}}
\newcommand{\dd}{\mathcal{D}}
\newcommand{\pp}{\mathcal{P}}
\newcommand{\ii}{\mathcal{I}}
\newcommand{\bb}{\mathbb{B}} 

\newcommand{\rri}{{\overline{\mathbb{R}}}}

\DeclareRobustCommand{\stirling}{\genfrac\{\}{0pt}{}}

\newcommand{\vn}{\vec{n}}

\newcommand{\lfp}{\mathrm{lfp}\,}
\newcommand{\gfp}{\mathrm{gfp}\,}
\newcommand{\fsol}{f_{sol}}
\newcommand{\candf}{\hat{f}}
\newcommand{\Fix}{\mathrm{Fixp}}
\newcommand{\PreFix}{\mathrm{Prefp}}
\newcommand{\PostFix}{\mathrm{Postfp}}

\NewCommandCopy{\rawPhi}{\Phi}
\renewcommand{\Phi}{\mathrm{\rawPhi}}
\NewCommandCopy{\rawPsi}{\Psi}
\renewcommand{\Psi}{\mathrm{\rawPsi}}
\NewCommandCopy{\rawDelta}{\Delta}
\renewcommand{\Delta}{\mathrm{\rawDelta}}

\newcommand{\End}{\mathrm{End}}
\newcommand{\Vect}{\mathrm{Vec}}

\newcommand{\ite}{\mathbf{ite}}

\newcommand{\inputset}{\mathcal{X}} 

\newcommand{\Data}{\mathtt{Data}}

\newcommand{\sem}[1]{\llbracket#1\rrbracket}
\newcommand{\sema}[1]{\llbracket#1\rrbracket^\sharp}
\newcommand{\boundaries}{\mathfrak{B}}
\newcommand{\boundariesbound}{$\boundaries$-bound\xspace}
\newcommand{\boundariesbounds}{$\boundaries$-bounds\xspace} 
\newcommand{\boundariesubs}{$\boundaries$-ubs\xspace}
\newcommand{\casesdomain}{\mathcal{C}}

\newcommand{\transfc}{T}
\newcommand{\transfa}{T^\sharp}

\newcommand{\Seq}{\texttt{Seq}}
\newcommand{\alphaB}{\alpha_\boundaries}
\newcommand{\gammaB}{\gamma_\boundaries}

\newcommand{\pleq}{\mathop{\dot{\leq}}}

\newcommand{\ciaopp}{CiaoPP\xspace}

\tikzset{dot/.style = {circle, fill, minimum size=#1, inner sep=0pt, outer sep=0pt},
         dot/.default = 6pt
}

\pgfmathsetseed{1532168}
\newcommand\irregularcircle[2]{
  \pgfextra {\pgfmathsetmacro\len{(#1)+rand*(#2)}}
  +(0:\len pt)
  \foreach \a in {10,20,...,350}{
    \pgfextra {\pgfmathsetmacro\len{(#1)+rand*(#2)}}
    -- +(\a:\len pt)
  } -- cycle
}

\newsavebox{\codebox}

\usepackage[colorinlistoftodos,textwidth=2.2cm]{todonotes}
\newcommand{\lrnote}[1]{\todo[backgroundcolor=cyan!20]{{\bf LR:}~#1}}
\newcommand{\plgnote}[1]{\todo[backgroundcolor=orange!20]{{\bf PLG:}~#1}}
\newcommand{\mhnote}[1]{\todo[backgroundcolor=blue!20]{{\bf MH:}~#1}}
\newcommand{\lrnoteoffline}[1]{\todo[backgroundcolor=cyan!20]{{\bf LR:}~#1}}
\newcommand{\plgnoteoffline}[1]{\todo[backgroundcolor=orange!20]{{\bf PLG:}~#1}}
\newcommand{\mhnoteoffline}[1]{\todo[backgroundcolor=blue!20]{{\bf MH:}~#1}}
\newcommand{\lrnoteinline}[1]{\todo[inline,backgroundcolor=cyan!20]{{\bf LR:}~#1}}
\newcommand{\plgnoteinline}[1]{\todo[inline,backgroundcolor=orange!20]{{\bf PLG:}~#1}}
\newcommand{\mhnoteinline}[1]{\todo[inline,backgroundcolor=blue!20]{{\bf MH:}~#1}}
\renewcommand{\lrnote}[1]{}
\renewcommand{\plgnote}[1]{}
\renewcommand{\mhnote}[1]{}
\renewcommand{\lrnoteinline}[1]{}
\renewcommand{\plgnoteinline}[1]{}
\renewcommand{\mhnoteinline}[1]{}
\renewcommand{\lrnoteoffline}[1]{}
\renewcommand{\plgnoteoffline}[1]{}
\renewcommand{\mhnoteoffline}[1]{}

\begin{document}

\journalname{STTT}

\title{Abstractions of Sequences, Functions and Operators}

\author{
Louis Rustenholz$^{1,3,(\textrm{\Letter})}$\orcidlink{0000-0002-1599-2431},
Pedro Lopez-Garcia$^{2,3}$\orcidlink{0000-0002-1092-2071}, and
Manuel V. Hermenegildo$^{1,3}$\orcidlink{0000-0002-7583-323X}}

\institute{
$^{1}$ Universidad Polit\'{e}cnica de Madrid (UPM), Madrid, Spain \\
$^{2}$ Spanish Council for Scientific Research (CSIC), Madrid, Spain\\
$^{3}$ IMDEA Software Institute, Pozuelo de Alarc\'{o}n, Spain}

\authorrunning{L. Rustenholz et al.}

\thanks{The authors acknowledge funding XXX}

\maketitle

\begin{abstract}
%
We present theoretical and practical results on the order
theory of lattices of functions, focusing on Galois connections that
abstract (sets of) functions -- a topic known as
\emph{higher-order abstract interpretation}.

We are motivated by the challenge of inferring closed-form bounds on
functions which are defined recursively, i.e. as the fixed point of an
operator or, equivalently, as the solution to a functional equation.
This has multiple applications in program analysis (e.g. cost analysis,
loop acceleration, declarative language analysis) and in hybrid
systems governed by differential equations.

%
Our main contribution is a new family of constraint-based abstract
domains for abstracting numerical functions, 
\emph{\boundariesbound domains}, which abstract a function $f$ by a
conjunction of bounds from a preselected set of boundary functions.
They allow inferring
\emph{highly non-linear numerical invariants}, which
classical numerical abstract domains struggle with.
We uncover a \emph{convexity property} in the constraint space that
simplifies, and, in some cases, fully \emph{automates}, transfer
function design.

%

We also introduce \emph{domain abstraction}, a
functor
that lifts
arbitrary mappings in value space to Galois connections in function
space.
This supports abstraction
from symbolic to numerical
functions (i.e. \emph{size abstraction}), and enables
dimensionality reduction of equations.

We base our constructions of transfer functions
on a simple \emph{operator language}, starting with \emph{sequences},
and extending to more general \emph{functions}, including multivariate,
piecewise, and non-discrete domains.

\keywords{
  Static Analysis
  \and
  Abstract Interpretation
  \and
  Functional Equations
  \and
  Higher-order Analysis
  \and
  Non-linear Invariants
  \and
  Cost Analysis
}

\end{abstract}


\setcounter{tocdepth}{2}


\section{Introduction}
\label{sec:intro}

Functions are fundamental  objects in program semantics and formal
descriptions of dynamical systems, and are thus of prime interest to
static analysis. However, classical abstract interpretation has often
focused on state-oriented perspectives, emphasising ``zeroth-order''
abstract domains, which describe sets of values.
As a result, properties and abstractions of lattices of
\emph{functions} -- a topic sometimes called \emph{higher-order
abstract interpretation}~\cite{cousot1994higher-short} -- have received
comparatively little attention.
%
In this paper, we explore a perspective treating functions as
first-class objects, and we build abstract domains of abstract functions,
using functions as the basic building block.
They allow us to use abstract interpretation as a solving strategy to
infer closed-form bounds on the solutions to functional equations,
viewed as abstractions of programs and systems.

Our original motivation stems from \emph{static cost analysis}: the
automated inference of bounds on a program's resource consumption
(e.g. execution time, memory, energy consumption).
In static cost analysis, unlike classical worst-case execution time
(WCET) analysis, we wish to infer \emph{cost functions} that are
\emph{parametric} in some abstraction of program inputs (e.g.
numerical \emph{sizes}), and are correct both asymptotically and for
small inputs.
%
%
One established approach to static cost analysis, used in the static analyser
\ciaopp~\cite{granularity-shortest,caslog-shortest,low-bounds-ilps97-short,ciaopp-sas03-journal-scp-shortest,resource-iclp07-extrashortest,plai-resources-iclp14-shortest,gen-staticprofiling-iclp16-extrashortest,resource-verification-tplp18-extrashortest,mlrec-tplp2024-nourl,order-recsolv-sas24-nourl}
and other related tools,~\cite{Wegbreit75-short-plus,Le88-shortest,Rosendahl89-shortest,AlbertAGP11a-short,montoya-phdthesis-short,kincaid2018-shortest,LommenGiesl23-shortest}
is based on
automated extraction of recurrence equations, which may be viewed as
an \emph{abstraction} of the program itself.
Solutions or bounds on the solutions to these equations then
provide information on the cost of the program.

Of course, functional equations are ubiquitous in static analysis and
computer science in general, from symbolic semantic equations, to
numerical recurrence equations for cost analysis, and even
differential equations in dynamical systems.
Thus, the work presented in this paper is applicable
in many contexts beyond cost analysis.

In particular, higher-order abstract interpretation has recently
received a surge of attention in the static analysis of higher-order
declarative
languages~\cite{montagu-icfp2020-short,montagu-salto-prelim-ml23-short,valnet-ECOOP25-short},
where information on functions defined recursively must be inferred
and represented with abstract values.
%
%
However, most recent works on this question build their abstract
domains of functions by falling back on classical \emph{relational}
numerical domains to encode abstract functions, such as the relational
liftings of~\cite{Bautista-FMSD2024-short} or the disjunctive relational
summaries of~\cite{valnet-ECOOP25-short}.
We argue that it is worthwhile to explore the design of
\emph{domain-specific} abstract domains of functions.
Our intuition is that functions $X\to Y$ may be simpler objects than
arbitrary 
sets
$\pp(X \times Y)$, at least for the families of
numerical functions arising from the programs we wish to approximate.
In particular, we show that it is possible to exploit the \emph{local
regularity} of such functions, and more
importantly the \emph{simple recursive structure} used to define them
via \emph{operators} $\Phi\in\big((X\to Y)\to(X\to Y)\big)$.
In fact, while classical numerical domains struggle to express
more than affine properties (e.g., those of polyhedra),
our investigation leads
to \emph{non-linear} abstract domains of functions, called
\boundariesbounds.

The ideas presented in this paper build upon investigations
performed in the context of the cost analysis pipeline of the
\ciaopp system.
Most
explicitly, we build on the framework of~\cite{order-recsolv-sas24-nourl}, which uses
an equation-as-operator viewpoint (see Section~\ref{subsec:eq-as-op}).
The motivation of~\cite{rustenholz-aars-msc,mlrec-tplp2024-nourl,order-recsolv-sas24-nourl}
originates from the limitations of state-of-the-art recurrence
solvers and cost analysis approaches.
First, general purpose abstract domains are insufficient, as the
numerical properties relevant to static cost analysis are highly
non-linear.
In contrast, symbolic computations techniques
implemented in Computer Algebra Systems (CAS)
can
infer extremely non-linear solutions, but only for very restricted
classes of equations.
%
Unfortunately, such equations
have very different features from those that naturally arise in
abstractions of semantic equations.  Despite significant progress in
specialised solvers, they still
do not support the breadth of behaviours encountered in real programs
(see~\cite{mlrec-tplp2024-nourl,order-recsolv-sas24-nourl} for a
detailed discussion of these limitations).

The clearest mismatch between (classical) symbolic computation and
static analysis is perhaps the former's focus on \emph{exact}
resolution of \emph{well-behaved} equations, instead of sound
\emph{approximate} resolution of \emph{general} equations.
The work of~\cite{order-recsolv-sas24-nourl} proposes to tackle
the latter problem with the tools of order theory.
It reframes
the bound inference problem into pre/postfixpoint search, suggests
various strategies (see Section~\ref{subsec:eq-as-op}), and sketches early
ideas on abstract interpretation and \boundariesbounds domains.
However, the experimental evaluation
in~\cite{order-recsolv-sas24-nourl} focused on a \emph{dynamic}
instantiation of the proposed approach, based on constrained
optimisation, which requires input/output samples.
A key motivation for this choice was practical.
Optimisation-based methods are easier to prototype than abstract
interpretation techniques, which demand
more design effort: for each template family $\boundaries$, one must 
design a fully-fledged abstract domain, including abstract operators,
transfer functions, and widening strategies.
%
For these reasons, \boundariesbound domains were left as future work
in~\cite{order-recsolv-sas24-nourl}.
Nevertheless, abstract interpretation
has several advantages, including greater generalisability (e.g. to
continuous systems for which input/output samples may not be
available),
and the ability to provide
\emph{guarantees}, \emph{by
construction}, of always obtaining a postfixpoint (though possibly
an imprecise one).

%
%
The present paper reports new results in this direction, serving as a
stepping stone towards further
applications and generalisations.

\vspace{2pt}
Motivated by these challenges, we provide the following contributions.
\vspace{-2pt}
\begin{itemize}[leftmargin=*] 
  \item In Section~\ref{sec:domain-abstraction}, we
    present
    \emph{theoretical results} on lattices of functions:
    %
    technical tools applicable to general higher-order
    abstract interpretation.
    We use them to provide formal foundations to constructions
    known in cost analysis, laying the groundwork for practical
    construction of sound and accurate 
    analysers.
  \item In particular, we introduce \textbf{domain abstraction}
    (Theorem~\ref{th:domain-abstraction}), a functor that lifts
    mappings in value space to Galois connections in function space.
    %
    We exemplify it for \emph{size abstraction} (i.e. symbolic to
    numerical abstraction) and \emph{dimensionality reduction}.
  \item In Section~\ref{sec:boundariesbound-domains}, we introduce
    \emph{\boundariesbound domains}, the main contribution of this paper,
    which are a class of \textbf{non-linear constraint-based abstract domains}
    parameterised by a set $\boundaries$.
  \item  \boundariesbounds are put in practice in
    Section~\ref{sec:abstr-dom-sequences}, where we build
    abstract domains and their transfer functions, using polynomial bounds
    (Section~\ref{subsec:poly-bounds}) and product of polynomials with
    exponentials (Section~\ref{subsec:exp-poly-bounds}).
    These examples are given in the simple case of \emph{sequences},
    using a minimalistic operator language. 
  \item We uncover a \textbf{convexity property in constraint space}
    (Theorem~\ref{th:convexity-constraints}), that facilitates the
    design of abstract transfer functions, in some cases enabling
    automated \textbf{transfer function synthesis} (see
    Section~\ref{sec:convexity-and-synthesis}).
  \item Finally, we discuss how \boundariesbound domains can be
    extended
    to multivariate, piecewise functions on less
    minimalistic languages in Section~\ref{sec:abstr-dom-functions}.
\end{itemize}
\lrnoteoffline{Should be p.2}

\section{Preliminaries}
\label{sec:preliminaries}

\subsection{Notation}

$\nn$, $\rr_+$ and $\bb=\{\bot,\top\}$ represent the sets
of non-negative integers, non-negative real numbers, and Booleans,
respectively.
For convenience, we use a function $\ite$ defined by
$\ite(\top,x,y)=x$ and $\ite(\bot,x,y)=y$.
We extend the real numbers equipped with their usual order structure
$(\rr,\leq)$ into a complete lattice $(\rri,\leq_{\rri})$, where
$\rri:=\{-\infty\}\,\cup\,\rr\,\cup\,\{+\infty\}$.
$(\ii(\rr),\sqsubseteq)$ is the lattice of intervals of real numbers.
$\End(A):= A \to A$ is the set of functions from $A$ to itself, and
$\End_\leq(L)$ is the set of monotone (i.e. order-preserving)
functions from $L$ to itself.
Given a monotone operator $\Phi\in\End_\leq(L)$ on a lattice,
$\lfp \Phi$ and $\gfp \Phi$ denote the least and greatest fixpoints of
$\Phi$, respectively.
$\Fix(\Phi)$, $\PreFix(\Phi)$, and $\PostFix(\Phi)$
respectively
represent its set of fixpoints, 
prefixpoints
($f \leq \Phi f$), and postfixpoints ($\Phi f \leq f$). 

\subsection{Equations as Operators Viewpoint}
\label{subsec:eq-as-op}

The work presented in this paper builds upon the equation as operator
viewpoint developed in~\cite{order-recsolv-sas24-nourl}, as well as
further insights developed in the study of this order theory framework
of functional equations.


The basic idea of this approach is to recall that a functional
equation with unknown $f:\dd\to L$ (for arbitrary sets $\dd$ and $L$),
may be equivalently described as a fixpoint problem
$$\forall \vn\in\dd,\,f(\vn)=\Phi(f)(\vn),$$
for some operator in function space $\Phi\in\End(\dd\to L)$,
so that $\Phi f = f$ if and only if $f$ is a solution to the equation.
The same observation applies to recursive definitions of a function of
type $\dd\to L$, where the operator $\Phi$ is used to define $f(\vn)$
in terms of other values of $f$ (the recursive calls).

\begin{example}
  \label{ex:eq-as-op}

  Consider the functional equation~(\ref{eq:eq-as-op-ex1}) below, on
  the unknown $f:\nn^2\to\rr$.
  \begin{equation}\footnotesize
    \label{eq:eq-as-op-ex1}
    f(n, c) = 
    \begin{cases}
        f(n-1, 0)   + n + 300 \!& \!\!\text{if } n>0 \text{ and } c\geq100,\\
        f(n-1, c+1) + n       \!& \!\!\text{if } n>0 \text{ and } c<100,\\
        c                     \!& \!\!\text{if } n=0,
    \end{cases}
  \end{equation}

  For any function $f:\nn^2\to\rr$ (possibly not a solution), the right
  hand side of~(\ref{eq:eq-as-op-ex1}) is also a function in
  $\nn^2\to\rr$: this defines an operator
  \begin{equation*}
  \begin{aligned}\footnotesize
    \Phi &\,: (\nn^2 \to \rr) \to (\nn^2 \to \rr)\\
           &\,f \mapsto \!\left(\!(n,c) \mapsto \!{\scriptsize\begin{cases}
         f(n-1, 0)   + n + 300 \!& \!\!\text{if } n>0 \text{ and } c\geq100\\
         f(n-1, c+1) + n       \!& \!\!\text{if } n>0 \text{ and } c<100\\
         c                     \!& \!\!\text{if } n=0
        \end{cases}}\!\right)\!.
  \end{aligned}
  \end{equation*}
  Searching for a solution $f$ of~(\ref{eq:eq-as-op-ex1}) is now
  equivalent to searching for an $f$ such that $f=\Phi f$, i.e.
  $f\in\Fix(\Phi)$.
\end{example}

\plgnote{Give a concrete example, e.g. given $g(n,c) = n^2 + c$,
  $\Phi(g)(n, c) = \ldots$ ? (if time/space allows).}


The next step in the philosophy of this approach is to consider that
the operator $\Phi$ can provide valuable information on the solutions
$\fsol\in\Fix(\Phi)$ even when it is applied \emph{outside} of these
fixpoints.
A fundamental instance of this idea
arises
when we employ order theory to define pre- and postfixpoints of
$\Phi$, which then yield
lower and upper bounds on its solutions.


\begin{definition}[Monotone Equation~\cite{order-recsolv-sas24-nourl}]\-
  Let $(\dd\to L,\leq)$ be a complete lattice of functions.
  A \emph{monotone equation} on $(\dd\to L,\leq)$ is simply a monotone
  operator $\Phi\in\End_\leq(\dd\to L)$.
  We respectively call $\lfp \Phi$ and $\gfp \Phi$
  the \emph{least} and \emph{greatest solutions} to the equation.
  If $\lfp \Phi = \gfp \Phi$,
  the unique fixpoint of $\Phi$ is simply called \emph{the} solution
  to the equation.
\end{definition}

For such equations, we can apply the Knaster-Tarski theorem and
obtain the following corollaries~\cite{order-recsolv-sas24-nourl}.

\begin{proposition} 
   Monotone equations admit a solution.
\end{proposition}

\begin{theorem}
  \label{th:proof-principle}
  Let $\Phi\in\End_\leq(\dd\to L)$ be a monotone equation.
  For any $f:\dd\to L$,
  \begin{itemize}
    \item if $\Phi f \leq f$, i.e. if $f\in\PostFix(\Phi)$,
          then $\lfp \Phi \leq f$, i.e. $f$ is an upper bound on the least
          solution of $\Phi$,
    \item if $f \leq \Phi f$, i.e. if $f\in\PreFix(\Phi)$,
          then $f \leq \gfp \Phi$, i.e. $f$ is a lower bound on the greatest
          solution of $\Phi$.
  \end{itemize}
\end{theorem}

\begin{remark}
  In particular, notice that if $\Phi$ admits a unique solution
  $\fsol$, then any postfixpoint is an upper bound on $\fsol$, and any
  prefixpoint is a lower bound on $\fsol$.
  This is always the case for \emph{terminating equations} (see~\cite{order-recsolv-preprint-extended}).
\end{remark}

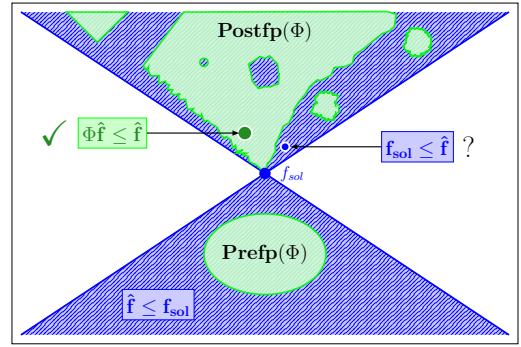
\begin{figure}\centering

  \begin{center}\footnotesize
    Equations $\leftrightarrow$ Operators
    $\quad\!\!$
    Sol. $\leftrightarrow$ Fixp.
    $\quad\!\!$
    Bounds $\leftarrow$ Pre/Postfp
  \end{center}
  \vspace{-.5em}
  \resizebox{.8\linewidth}{!}{%
    \tikzstyle{arrow} = [draw,thick,->]
    \begin{tikzpicture}[framed, every node/.style={align=center}, >=latex, scale=1,
                        x=1.5cm]

      \node (fsol) at (0,0) {};

      \begin{scope}[transparency group]
        \begin{scope}[blend mode=multiply]
            \draw[-, very thick, blue] (fsol) -- ++(4,4)   node (ub-r) {};
            \draw[-, very thick, blue] (fsol) -- ++(-4,4)  node (ub-l) {};
            \draw[-, very thick, blue] (fsol) -- ++(4,-4)  node (lb-r) {};
            \draw[-, very thick, blue] (fsol) -- ++(-4,-4) node (lb-l) {};
            \fill[blue!20] (fsol.center)--(ub-r.center)--(ub-l.center);
            \fill[pattern=north east lines, pattern color=blue] (fsol.center)--(ub-r.center)--(ub-l.center);
            \fill[blue!20] (fsol.center)--(lb-r.center)--(lb-l.center);
            \fill[pattern=north east lines, pattern color=blue] (fsol.center)--(lb-r.center)--(lb-l.center);
        \end{scope}
        \begin{scope}[blend mode=normal,opacity=.9]

            \node (pofp-r) at ($(fsol)+(2,4)$) {};
            \node (pofp-midr) at ($(fsol)!.66!(pofp-r)$) {};
            \node (pofp-l) at ($(fsol)+(-3,4)$) {};
            \node (pofp-midl) at ($(fsol)!0.666!(pofp-l)$) {}; 
            \node (pofp-midl2) at ($(fsol)!.75!(pofp-l)$) {}; 
            \node (up-midl) at (-1.333,4) {};
            \node (up-midl2) at (-1.75,4) {};

            \draw [green,rounded corners=1mm, fill=green!20] (1,1.66) \irregularcircle{.33cm}{1mm};
            \draw [green,rounded corners=1mm, fill=green!20] (2.5,3.25) \irregularcircle{.33cm}{.5mm};
            \draw[-, very thick, green, fill=green!20, even odd rule]
               (fsol.center)
               decorate [ decoration={random steps, segment length=2mm} ] {
                 to[out=70,in=-160] (pofp-midr.center)
                 to[out=50,in=-40] (pofp-r.center)
               }
               -- (up-midl.center)
               -- (pofp-midl.center)
               decorate [ decoration={random steps, segment length=.75mm} ] {
                 -- (fsol.center)
               }
               decorate [ decoration={random steps, segment length=2mm} ] {
                 (0,2.5) ellipse (.33cm and .33cm)
               }
               (-1,2.75) ellipse (.1cm and .1cm)
               ($(pofp-midl2.center)+(-.5,.25)$) -- ($(up-midl2.center)+(-.5,0)$) -- ($(pofp-l.center)+(-.25,0)$) -- ($(pofp-midl2.center)+(-.5,.25)$)
               ;

            \draw[-, very thick, green, fill=green!20, even odd rule]
               (0,-2) ellipse (1.5cm and 1cm);


        \end{scope}
      \end{scope}

      \node[shape=circle, fill=blue, label={[text=blue]right:{\large $\;\fsol$}}] (fsol-circ) at (fsol) {};
      \node[draw=blue, fill=blue!20, text=blue] at (-1.75,-3.25) {\Large $\mathbf{\candf \leq \fsol}$};
      \node[text=black] at (0,-2) {\Large $\mathbf{Prefp(\Phi)}$};

      \node[draw=blue, fill=blue!20, text=blue] (ubs) at (2.5,.66) {\Large $\mathbf{\fsol \leq \candf}$};
      \node[right= 1mm of ubs] {\huge ?};
      \draw[draw=white, fill=blue, very thick] (.33,.66) circle (3pt) node (ub) {};
      \draw[arrow] (ubs) -- (ub);

      \node[draw=green, fill=green!20, text=ForestGreen] (postfixps) at (-2.5,1) {\Large $\mathbf{\Phi \candf \leq \candf}$};
      \node[left= 1mm of postfixps] {\textcolor{ForestGreen}{\Huge \checkmark}};
      \draw[draw=white, fill=ForestGreen, very thick] (-.33,1) circle (5pt) node (postfixp) {};
      \draw[arrow] (postfixps) -- (postfixp);

      \node[text=black] at (0,3.5) {\Large $\mathbf{Postfp(\Phi)}$};

    \end{tikzpicture}%
  }

  \small
  \caption{Illustration of the function space explored under the
    equation-as-operator viewpoint.
    While exploring the space, it is not known
    whether a point lies within the set of bounds on $\fsol$ (blue),
    but it is possible to check whether it is a pre- or postfixpoint
    (green).
%
    %
%
    }
  \label{fig:eq-as-op}
\end{figure}

A priori, given an equation $\Phi$, automatically checking whether an
arbitrary \emph{candidate function} $\candf$ is a bound on the unknown
solution is a difficult problem.
When Theorem~\ref{th:proof-principle} applies, it provides a
significantly simpler sufficient condition for
verifying whether a guessed bound is correct:
we only\footnote{Function comparison is
still a non-trivial problem, but complete methods exist for several
classes of functions, and efficient incomplete methods are available
for many others.} need to perform function comparison
$\Phi\candf \leq^{(?)} \candf$.

To \emph{discover} bounds on the solution to the equation, it is
therefore
sufficient to perform a pre-/post-fixpoint search -- a task for which
a wide variety of powerful methods can be designed.
New opportunities arise when dealing with \emph{numerical} functional
equations, 
some of which are discussed in~\cite{order-recsolv-sas24-nourl}. That
work explores several pre-/post-fixpoint search strategies, discusses
trade-offs, and proposes viewing
the search \emph{geometrically},
as an exploration
in function space (see Fig.~\ref{fig:eq-as-op}).
For example, it is possible to exploit the total order structure of
$\rr$, the vector space structure of $\dd\to\rr$, and the availability
of procedures for quantifier elimination or numerical optimisation.

For its experimental evaluation,~\cite{order-recsolv-sas24-nourl}
focused on an instantiation of the proposed order-theoretical
framework
%
%
based on constrained optimisation.
%
This
method is only applicable when we can produce \emph{input/output
samples} of the form $\{(\vn,\fsol(\vn))\,|\,\vn\in\inputset\}$, for
some finite set of inputs $\inputset\subseteq\dd$.  This is feasible,
for instance, for \emph{recurrence} equations where $\dd=\nn^k$ is
discrete and the equation can be executed like
a program.



From the point of view of abstract interpretation, the most
natural pre/postfixpoint search methodology is to use instead
\emph{abstract Kleene iteration}, with \emph{abstract functions},
\emph{abstract equations} and \emph{abstract operators}.
Tools enabling this alternative strategy are discussed in the current
paper, which complements~\cite{order-recsolv-sas24-nourl}.

\subsection{Orders in function space}

Given an equation $\Phi$, what order can we impose on our set of
functions to make Theorem~\ref{th:proof-principle} usable and useful?

In this paper, we always use the most natural order structure on
sets of functions: the \emph{pointwise order}. It gives us a
notion of upper and lower bounds that is useful for program analysis,
as well as an intuitive notion of function comparison.

\begin{definition}[Pointwise lattice structure]
   For $X$ a set and $(L,\leq,\vee,\wedge)$ a complete lattice, the
   set of functions $X\to L$
   can be equipped with a
   natural lattice structure $(X\to L, \pleq, \dot{\vee}, \dot{\wedge})$, defined
   by

\begingroup
  \abovedisplayskip=0pt
  \belowdisplayskip=2pt
   \begin{center}
   \begin{minipage}[c]{.8\linewidth}
   \begin{equation*}\centering
     \begin{gathered}
     f \pleq g \overset{\Delta}{\iff} \forall x\in X,\, f(x) \leq g(x)\\
     f \,\dot{\vee}\,   g \overset{\Delta}{=}  \big(x \mapsto f(x) \vee g(x)\big),
     \quad\,\dot{\bot}\,     \overset{\Delta}{=}  \big(x \mapsto \bot\big)\\
     f \,\dot{\wedge}\,   g \overset{\Delta}{=}  \big(x \mapsto f(x) \wedge g(x)\big),
     \quad\,\dot{\top}\,     \overset{\Delta}{=}  \big(x \mapsto \top\big)\\
     \end{gathered}
   \end{equation*}
   \end{minipage}
   \end{center}
\endgroup
\end{definition}

This order structure, however, appears in several
different flavours, depending on the order used in
the codomain
$(L,\leq)$.
We are interested in functions arising in numerical contexts
where three main types
of order appear
in the literature and this work:
the usual order on the real numbers, set inclusion orders
(e.g. interval inclusion), and the
flat order.

\begin{example}[Flat order, denotational semantics]
  It is common to see pointwise orders in classical denotational
  semantics, e.g. to assign meaning to recursively defined functions in
  the Programming language for Computable Functions (PCF)~\cite{Plotkin77-LCF-short,Sco82,curien-pcf-chapter-98}.
  In such cases, to give a semantic interpretation to
  a function recursively defined by an operator
  $\Phi\in\End(\dd\to X)$, we start from the discrete order $(X,=)$
  (where all distinct elements are incomparable), which is then extended to a
  Scott domain $(X_\bot,\leq_{\text{flat}}$) by adjoining a new bottom
  element below all others, $X_\bot=\{\bot\}\cup X$.
  Elements of $\dd\to X_\bot$ can then be thought of
  as \emph{partial
  functions}, where the value $\bot$ represents
  \texttt{undefined}. $\Phi$ is then extended into a monotone
  operator $\widetilde{\Phi}$ on the pointwise-ordered
  $(\dd\to X_\bot,\dot{\leq}_{\text{flat}})$, and the recursive
  definition is interpreted via the function $\lfp \widetilde{\Phi}$,
  which may e.g. be computed by iterating from $\dot{\bot}$, the
  nowhere-defined function.
  If a complete lattice is desired,
  a top element can be added to obtain the \emph{flat lattice}
   $\{\bot\}\cup X\cup\{\top\}$, e.g.  to represent conflicting
   definitions. The flat lattice is also used in the context of
   constant propagation~\cite{RedDragon-short}.
\end{example}

\begin{remark}[Termination]
  We use an abstraction of the pointwise flat lattice to the pointwise
  Boolean lattice (on $\bb=\{\bot,\top\}$) to define termination of
  equations, via
  $\alpha:\{\bot\}\cup X\cup\{\top\}\to \bb$, where $\alpha(x)=\top$ for
  all $x\in X$. In this abstraction, $\bot$ represents
  \texttt{undefined}, and $\top$ represents \texttt{defined}.
  If $\lfp \Phi^\sharp=\dot{\top}$ for the abstract $\Phi^\sharp\in\End_{\dot{\leq}_{\bb}}(\dd\to\bb)$,
  this guarantees that there is a unique
  solution $\fsol = \lfp \Phi = \gfp \Phi$
  for the concrete equation $\Phi\in\End_{\dot{\leq}_L}(\dd\to L)$ (see~\cite{order-recsolv-preprint-extended}).
\end{remark}

\begin{example}[Pointwise inclusion, non-relational abstract domains]
  Another classical flavour of pointwise orders comes from
  non-relational domains in abstract
  interpretation~\cite{MineTutorialBook2017-short}.
  %
  In the analysis of an imperative program, with a finite set of
  real-valued program variables $\mathcal{V}$,
  \emph{value} abstract domains are defined via a Galois connection
  $(D^\sharp,\sqsubseteq^\sharp)\galois{\alpha}{\gamma}(\pp(\rr),\subseteq)$,
  where $\sqsubseteq^\sharp$ is thought of
  as (an abstraction of) the
  \emph{inclusion order}.
  This abstract domain is then lifted \emph{pointwise} to a non-relational
  \emph{memory state} abstract domain $(\mathcal{V}\to D^\sharp,\dot{\sqsubseteq}^\sharp)$.
\end{example}

A defining feature of our approach is the use of pointwise orders on
\emph{numerical functions}, i.e. constructed from numerical orders
$(\rr,\leq)$ and on an infinite domain $\dd$ thought of as classical
input values (not merely program points, labels, or variable names in $\mathcal{V}$).

Some simple examples are given below. In each case,
we also provide the value of $\dot{\bot}$, which will serve as the
starting point for abstract iteration.

\begin{example}[Numerical functions]
      In $(\dd\to\rri, \pleq)$,
      the join and meet operations are respectively the pointwise
      $\max$ and pointwise $\min$. 
      We have $\dot{\bot} = x\mapsto -\infty$.
\end{example}

\begin{example}[Numerical functions (2)]
For simplicity, we sometimes use $(\dd\to\rri_+, \pleq)$,
   where $\rri_+:=\rr_+\cup\{+\infty\}$ and $\dot{\bot} = x\mapsto 0$.
\end{example}

Equations $\Phi$ are not always monotone with respect to this
pointwise order on the real numbers, so Theorem~\ref{th:proof-principle} may
not be directly applicable. However, as we will see in the next
section, it is possible to recover monotonicity by moving to
\emph{interval-valued} functions, or more generally \emph{set-valued}
functions.

\begin{example}[Set-valued functions]
    In $(\dd\to\ii(\rr), \dot{\sqsubseteq}_\ii)$,
    join/meet are pointwise union/intersection,
    and $\dot{\bot} = x\mapsto \varnothing$.
    Similar examples can be given for $\dd\to\pp(\rr)$,
    or intermediate cases, using lifted classical numerical abstract
    domains (zones, polyhedra, etc.) to abstract $\dd\to\pp(\rr^k)$.
\end{example}

\begin{example}[With Kaucher intervals]
  \label{ex:kaucher}
  In practice, to improve reasoning over interval-valued functions and
  simplify their representation as pairs of real-valued functions, it
  can also be useful to use \emph{Kaucher intervals} $\ii_K(\rr)$. These form
  an ``algebraic completion'' of classical intervals
  $\ii(\rr)$, extended with \emph{improper intervals} $[a,b]$ where
  $a>b$~\cite{Kaucher80,ModalItvBook2014-short,Shary-nontraditv-2023-short}.
  Kaucher intervals replace the empty interval $\varnothing$ and improve algebraic
  properties with respect to addition and multiplication.
  More precisely, $\ii_K(\rr):=\,\rri\times\rri$, and
  $[a,b]\sqsubseteq_{\ii_K(\rr)} [c,d]$
  whenever
  $(a\geq_\rri c) \wedge (b \leq_\rri d)$, i.e.
  $\sqsubseteq_{\ii_K(\rr)} \,=\, \geq_\rri\times\leq_\rri$.
  The pointwise extension becomes
  $(\dd\to\ii_K(\rr), \dot{\sqsubseteq}_{\ii_K})$,
  with $\dot{\bot} = x\mapsto [+\infty,-\infty]$.
\end{example}

\section{Some Galois Connections in Function Space: the Case of Domain Abstraction}
\label{sec:domain-abstraction}

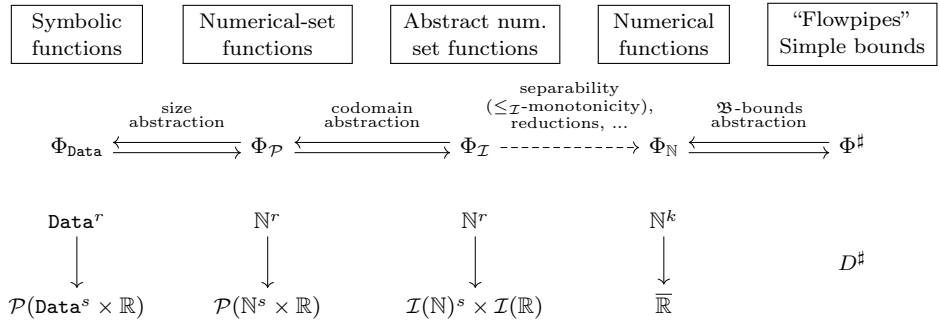
\begin{figure*}[ht]\centering
    \begin{tikzcd}[scale=1,ampersand replacement=\&,row sep=3mm,column sep=2.5mm]
           \fbox{\parbox{1.5cm}{\centering \footnotesize Symbolic functions}}
        \& \fbox{\parbox{2cm}{\centering \footnotesize Numerical-set functions}}
        \& \fbox{\parbox{2cm}{\centering \footnotesize Abstract num. set functions}}
        \& \color{black}{\fbox{\parbox{1.5cm}{\centering \footnotesize Numerical functions}}}
        \& \fbox{\parbox{2cm}{\centering \footnotesize ``Flowpipes''\\Simple bounds}}\\
        \&\&\&\&\\[1mm]
           \Phi_{\Data} \ar[r,shift right] %
        \& \Phi_{\pp}   \ar[r,shift right]
           \ar[l,shift right, "\substack{\text{size}\\\text{abstraction}\\[-.2em]\text{}}" above] 
        \& \Phi_{\ii}   \ar[r,black,dashed, "\substack{\text{separability}\\\text{($\leq_\ii$-monotonicity),}\\\text{reductions, ...}\\\text{\phantom{.}}}" above]
           \ar[l,shift right, "\substack{\text{codomain}\\\text{abstraction}\\[-.2em]\text{}}" above]
        \&  \color{black}{\Phi_{\nn}}  \ar[r,shift right]
        \& \Phi^\sharp
           \ar[l,shift right, "\substack{\text{\boundariesbounds}\\\text{abstraction}\\[-.2em]\text{}}" above]
           \\
        \&\&\&\&\\[-1mm]
           \Data^r \ar[dd] 
        \& \nn^r \ar[dd]
        \& \nn^r \ar[dd]
        \& \color{black}{\nn^k}   \ar[black,dd] 
        \& \\[-2mm]
           %
        \&
        \&
        \&
        \& 
           D^\sharp
        \\[-2mm]
           \pp(\Data^s\times\rr)
        \& \pp(\nn^s\times\rr)
        \& \ii(\nn)^s\times\ii(\rr)
        \& \color{black}{\rri} 
        \& \\
      \end{tikzcd}
    \caption{From symbolic functions to manageable numerical functions.}
    \label{fig:abstr-landscape}
  \end{figure*}

Now that we have given a few flavours of the order structures that can
be imposed on sets of functions, we provide some convenient Galois
constructors and useful
properties on these lattices of functions.

All of these Galois connections have proven useful in our research on
cost analysis. Fig.~\ref{fig:abstr-landscape} provides an overview of
the abstraction landscape used in our cost analysis pipeline,
implemented in the \ciaopp system, as presented in this paper. In this
section, our goal is to provide technical tools that may assist
readers interested in conducting their own investigations in
higher-order abstract interpretation and the order theory of function
spaces. Additionally, we aim to offer a formal foundation for the
practical construction of sound and accurate cost analysers.

For this reason, we include in Section~\ref{subsec:classic-galois-fun}
some classical constructions to build Galois connections between lattices
of functions.
In Section~\ref{subsec:domain-abstr-galois}, we present a Galois
constructor for which we have not been able to find references in the
static analysis literature, and which is thus novel to the best of the
authors' knowledge: \textbf{domain abstraction}, which transforms a
simple arbitrary mapping in value space into a Galois connection in
function space.
We discuss the case of interval-valued functions in
Section~\ref{subsec:interval-fun}, and set the stage
in~\ref{subsec:B-bound-intro} for the main Galois connection of this
paper: the \emph{\boundariesbound abstraction}, which is developed
in
greater
detail in the following sections.







\subsection{Classical Manipulations} 
\label{subsec:classic-galois-fun}

In this section, we assume
a Galois connection
$$(D,\leq)\galois{\alpha}{\gamma}(D^\sharp,\sqsubseteq)$$ and recall
some well-known
constructors/functors that lift this connection to lattices
of functions.
See courses in abstract interpretation
or~\cite{cousot1994higher-short,cousotpopl14galoiscalculus-short} for condensed
presentations and additional results of this kind.


  \begin{proposition}[Codomain abstraction]
    \label{prop:codomain-abstraction}
    Let $X$ be a set. The Galois connection between $D$ and $D^\sharp$ lifts to
    %
    %
    $(X\to D,\pleq) \galois{\dot{\alpha}}{\dot{\gamma}} (X\to D^\sharp,\dot{\sqsubseteq})$, 
    with $\dot{\alpha}(f) = \alpha\circ f$ and $\dot{\gamma}(f^\sharp) = \gamma\circ f^\sharp$.
  \end{proposition}



  \begin{proposition}[$\End$-lifting] \label{prop:end-lifting}
    The Galois connection between $D$ and $D^\sharp$ lifts to
    %
    \begin{equation*}
    \begin{aligned}
      (\End_\leq(D),\pleq) &\galois{\vec{\alpha}}{\vec{\gamma}} (\End_\sqsubseteq(D^\sharp),\dot{\sqsubseteq})\\
      f &\longmapsto \alpha\circ f\circ\gamma,\\
      \gamma\circ f^\sharp\circ\alpha &\longmapsfrom f^\sharp.\\
    \end{aligned}
    \end{equation*}
    For example, the construction of $\vec{\alpha}(f)$ may be visualised in
    the following commutative diagram.
    \begin{center}
      \begin{tikzcd}[ampersand replacement=\&,column sep=large]
        D \ar[r, "f"] \& D \ar[d,"\alpha"] \\
        D^\sharp \ar[u, "\gamma"] \ar[r, dashed, "f^\sharp=\vec{\alpha}(f)"] \& D^\sharp
      \end{tikzcd}
    \end{center}
  \end{proposition}

  \begin{corollary}
    The construction of Proposition~\ref{prop:end-lifting} can be
    iterated: from values $D\galois{}{} D^\sharp$, to monotone
    functions $\End_\leq(D)\galois{}{} \End_{\sqsubseteq}(D^\sharp)$,
    to \textbf{monotone operators}
    (built on $\End^2\approx\big((\cdot\to\cdot)\to(\cdot\to\cdot)\big)$),
    and beyond to monotone higher-order operators.
  \end{corollary}

  Thanks to this result, the notion of \emph{abstraction of an
  equation} is automatically defined (both for equations on monotone
  functions and on arbitrary functions, if we start from an
  appropriate notion of function abstraction), and no
  additional checks are needed to ensure that the resulting structure
  forms a Galois connection.
%

\subsection{Domain Abstraction}
\label{subsec:domain-abstr-galois}

More surprisingly, it is also possible to take a simple, unstructured
function on values, and lift it into a full Galois connection on
functions, if we instead transform the \emph{domain} of the functions.

  \begin{theorem}[Domain abstraction]
    \label{th:domain-abstraction}
    Let $m:X\to A$ be an arbitrary mapping, and let $(L,\sqsubseteq)$ be a complete lattice.
    Then, there is a Galois connection 
    {
    \begin{align*}
      (X\to L,\dot{\sqsubseteq}) &\galois{\alpha}{\gamma} (A\to L,\dot{\sqsubseteq})\\
      f &\longmapsto \Big( a \mapsto \bigsqcup_{x\in m^{-1}(a)} f(x)\Big)\\[-5pt]
      \Big(x \mapsto f^\sharp\big( m(x)\big) \Big) &\longmapsfrom f^\sharp.
    \end{align*}}
    It is a Galois insertion whenever $m$ is surjective.

    In other words, precomposition by $m$ admits a left adjoint, which
    may be interpreted as a constrained maximisation operation.
    Indeed, in the sense of the order on $L$,
    $\alpha(f)(a) = \sup\,\{ f(x)\,|\,m(x)=a)\}$.
  \end{theorem}

  We opt for an abstract proof, beginning by considering the operation
  of precomposition
  $\gamma:=(f^\sharp \mapsto f^\sharp\circ m)$,
  checking the existence of its unique possible left adjoint $\alpha$,
  before observing its properties (see~\cite{Cousot21-book}). A more
  explicit proof can also be performed by the interested reader, by
  checking the classical inequations.

  \begin{proof}
    The fact that $\gamma$ preserves meets is a direct consequence of
    the pointwise structure: for all $F^\sharp\subseteq (A\to L)$ and
    $x\in X$,
    $ \gamma\big(\dot{\bigsqcap}_{f^\sharp\in F^\sharp} f^\sharp\big)(x)
      = \big(\dot{\bigsqcap}_{f^\sharp\in F^\sharp} f^\sharp\big)(m(x))
      = \big(\bigsqcap_{f^\sharp\in F^\sharp} f^\sharp(m(x))\big)
      = \big(\dot{\bigsqcap}_{f^\sharp\in F^\sharp} \gamma(f^\sharp) \big)(x).$
   This gives us the existence of $\alpha$ forming a Galois connection
   with $\gamma$, and its formula given above by a join.
    Finally, for the characterisation of the Galois insertion case,
    we simply compute
    $\alpha\circ\gamma(f^\sharp)(a)=\sqcup\{f^\sharp(m(x))\,|\,m(x)=a\}$
    \begingroup
    \abovedisplayskip=2pt
    \belowdisplayskip=6pt
    \begin{equation*}
      = \begin{cases}
        f^\sharp(a) & \text{if } \exists x\in X,\,m(x)=a,\\
        \bot       & \text{otherwise}.
      \end{cases}
    \end{equation*}
    \endgroup
  \end{proof}

  The construction of the abstraction in
  Theorem~\ref{th:domain-abstraction} is illustrated in
  Fig.~\ref{fig:dom-abstr}.
  To give some intuition, we present a first example coming from the
  context of cost analysis.

  \begin{figure}[t]
    \pgfdeclarelayer{back}
    \pgfdeclarelayer{middle}
    \pgfdeclarelayer{front}
    \pgfsetlayers{back,middle,main,front}
  \begin{tikzpicture}[scale=.5, every node/.style={font=\small},
       ->, >=Stealth, thick, node distance=1cm and 2cm]

    \begin{pgfonlayer}{back}
    \node[draw, ellipse, minimum width=2.8cm, minimum height=4.5cm, fill=red!10] (X) at (0,0) {};
    \node[draw, ellipse, minimum width=1.2cm, minimum height=3cm, fill=blue!10] (A) at (7,0) {};
    \end{pgfonlayer}
    \node[align=center,anchor=south east,font=\normalsize] (Xlab) at ($(X.north west)+(-.2,-.2)$) {$X$};
    \node[align=center,anchor=south west,font=\normalsize] (Alab) at (A.north east) {$A$};
    \node[align=left,font=\large] (Xlab2) at ($(Xlab.north)+(-.5,0.5)$) {\textcolor{orange}{$f:X\to L$}};
    \node[align=right,font=\large] (Alab2) at ($(Alab.north)+(0,1)$) {\textcolor{ForestGreen}{$f^\sharp:A\to L$}};

    \node at ($(X.center)+(0,3.75)$) {\(\vdots\)};
    \node (x1) at ($(X.center)+(0,2.5)+(0,-0.6)$) {$\bullet$};
    \node (x2) at ($(X.center)+(-0.5,0.5)+(0,-0.4)$) {$\bullet$};
    \node (x3) at ($(X.center)+(0.3,-0.3)+(0,-0.4)$) {$\bullet$};
    \node (x4) at ($(X.center)+(-0.6,-1.1)+(0,-0.4)$) {$\bullet$};
    \node (x5) at ($(X.center)+(-0.7,-2.7)+(0,-0.6)$) {$\bullet$};
    \node (x6) at ($(X.center)+(0.4,-2.7)+(0,-0.6)$) {$\bullet$};
    \node (xhid5) at ($(X.center)+(0,-2.4)+(0,-0.6)$) {};
    \node (xhid6) at ($(X.center)+(0,-3)+(0,-0.6)$) {};

    \node at ($(A.center)+(0,2.25)$) {\(\vdots\)};
    \node (a1) at ($(A.center)+(0,1)$) {$\bullet$};
    \node (a2) at ($(A.center)+(0,0)$) {$\bullet$};
    \node (a3) at ($(A.center)+(0,-1)$) {$\bullet$};
    \node (a4) at ($(A.center)+(0,-2)$) {$\bullet$};

    \draw[->] (x1) -- (a1);
    \draw[->] (x2) -- (a2);
    \draw[->] (x3) -- (a2);
    \draw[->] (x4) -- (a2);
    \draw[->] (x5) -- (a4);
    \draw[->] (x6) -- (a4);

    \begin{pgfonlayer}{middle}
    \node[draw=red, fill=red!20, dashed, ellipse, inner sep=2pt,
          fit=(x1)] (preimage1) {};
    \node[draw=red, fill=red!20, dashed, ellipse, inner sep=-2pt,
          fit=(x2)(x3)(x4)] (preimage2) {};
    \node[draw=red, fill=red!20, dashed, ellipse, inner sep=-1pt,
          fit=(x5)(x6)(xhid5)(xhid6)] (preimage4) {};
    \end{pgfonlayer}

    \draw[->, red!70!black, thick]
      (a4.center) to[out=-150, in=-10] node[below=0.05cm, near start] {\(m^{-1}\)} (preimage4);

    \node[above=0.1cm of $(x1)!0.5!(a1)$] {\(m\)};

    \node[draw=black,rectangle] (f5) at ($(x5.center)+(3,-3.5)$) {\textcolor{orange}{$u$}};
    \draw[->, orange, thick] (x6.center) to[out=-90,in=160] node[right,midway] {} (f5);
    \node[draw=black,rectangle] (f6) at ($(x6.center)+(3.5,-3.5)$) {\textcolor{orange}{$v$}};
    \draw[->, orange, thick] (x5.center) to[out=-100,in=-150] node[below left,midway] {\textcolor{orange}{$f$}} (f6);

    \node (f56mid) at ($(f5)!0.5!(f6)$) {};
    \node[draw=black,rectangle] (fuv) at ($(f56mid)+(0,1.5)$) {\textcolor{ForestGreen}{$u \sqcup v$}};

    \draw[-,black,dashed] (f5) -- (fuv);
    \draw[-,black,dashed] (f6) -- (fuv);
    \draw[->, ForestGreen, very thick] (a4.center) to[out=-90,in=0] node[below right,midway] {\textcolor{ForestGreen}{$f^\sharp=\alpha(f)$}} (fuv);

    \draw[->, orange, dashed] (x1.center) to[out=-180,in=90] ($(x1.center)+(-3.5,-.3)$);
    \draw[->, orange, dashed] (x2.center) to[out=-180,in=90] ($(x2.center)+(-3.5,-.3)$);
    \draw[->, orange, dashed] (x3.center) to[out=-180,in=90] ($(x3.center)+(-3.5,-.3)$);
    \draw[->, orange, dashed] (x4.center) to[out=-180,in=90] ($(x4.center)+(-3.5,-.3)$);

    \node[draw=black,rectangle] (ima3) at ($(a3.center)+(2.5,-1)$) {\textcolor{ForestGreen}{$\bot$}};
    \draw[->, ForestGreen, dashed] (a1.center) to[out=0,in=90] ($(a1.center)+(2.5,-.3)$);
    \draw[->, ForestGreen, dashed] (a2.center) to[out=0,in=90] ($(a2.center)+(2.5,-.3)$);
    \draw[->, ForestGreen, dashed] (a3.center) to[out=0,in=90] (ima3);



    \end{tikzpicture}
  \vspace{-.5em}
  \caption{Illustration of domain abstraction.}
  \label{fig:dom-abstr}
  \end{figure}
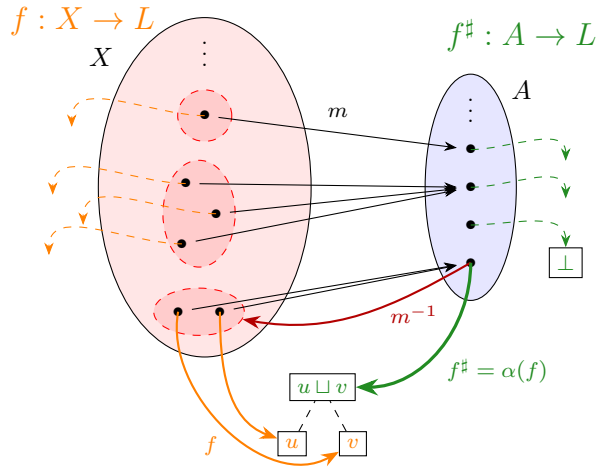

  \begin{example}[Size abstraction]
    Suppose we are interested in the cost (a real number, e.g. the
    execution time)
    of a deterministic procedure as a function of
    its inputs, which are elements of a large set $\Data$ of symbolic
    data structures (e.g. lists of numbers, trees with nodes
    containing integer values, or integer numbers themselves).

    The cost function we are interested in is an element
    $f:\Data\to\rri$, but this problem is often too difficult to be
    solved exactly.
    Thus, we perform an abstraction step, by introducing a \emph{size
    metric} $m:\Data\to\texttt{Sizes}$ from the large set of data
    structures to a simpler set of \emph{data sizes}, e.g.
    $\texttt{Sizes}=\nn$. For example, $m$ might compute the length of
    lists, the number of nodes in a tree, or the absolute value of
    integer numbers.

    We are thus left with a problem about an abstract cost function
    $f^\sharp:\texttt{Sizes}\to\rri$, which gives the cost of the
    procedure as a function of data \emph{sizes}.
    If we can infer the best possible $f^\sharp$, we do not have
    access to the exact cost of the procedure called on an arbitrary
    $x\in\Data$: only to the worst case on all other inputs of equal
    size, $f^\sharp\circ m (x)$.
    This explains both the intuition behind the precomposition step in
    $\gamma$, and the usage of join in $\alpha$.
  \end{example}

  Before discussing additional intuition coming from cost analysis,
  we must mention a couple of constructions of Galois connections
  derived from Theorem~\ref{th:domain-abstraction}.

  \begin{corollary}
    \label{cor:powerset-lifting}
    \textbf{Powerset lifting}, the classical trick to turn a relation
    $m:X\to A$ that fails to be a Galois connection into a Galois
    connection $(\pp(X),\subseteq)\galois{\alpha}{\gamma}(\pp(A),\subseteq)$
    by adding a few powersets functors $\pp$, is an instance of domain
    abstraction, with $L = \bb$.
  \end{corollary}

  \begin{corollary}[For non-deterministic functions]
    \label{cor:nondet-fun}
    By combining the connections of
    Theorem~\ref{th:domain-abstraction} and
    Corollary~\ref{cor:powerset-lifting}, we obtain a Galois
    connection
    $(X\to\pp(X),\dot{\subseteq})$\\$\galois{\alpha}{\gamma}(A\to\pp(A),\dot{\subseteq})$,
    given by $\alpha:f\mapsto m\circ f \circ m^{-1}$ and
    $\gamma:f^\sharp\mapsto m^{-1}\circ f^\sharp \circ m$ (using monadic composition as a shorthand).

    This can be interpreted as a Galois connection between the sets of
    non-deterministic (endo)functions on the ``big'' set $X$ and of
    non-deterministic (endo)functions on the ``small'' set $A$.
  \end{corollary}

  \begin{example}[Size abstraction, continued]
    Corollary~\ref{cor:nondet-fun} formalises the size
    abstraction step of \ciaopp's cost analysis pipeline,
    in which the semantics of predicates operating on concrete values
    (i.e. data structures in general), are viewed as non-deterministic
    functions for given calling patterns (or modes). These concrete
    predicates are then replaced by purely numerical predicates on
    data \emph{sizes}.

    More exactly, if we additionally apply
    Proposition~\ref{prop:end-lifting} to raise the discussion at the
    level of operators instead of functions, we obtain the full notion
    of size abstraction: \emph{programs} (i.e. their source code,
    recursively defining functions that operate on data structures),
    \emph{abstracted} into purely numerical programs, \emph{i.e. into
    numerical recurrence equations, viewed as operators}.
  \end{example}

  \begin{remark}
    In the interest of future discovery of connections between topics,
    it may be interesting to observe that, like many generic
    constructions in abstract interpretation,
    Theorem~\ref{th:domain-abstraction} may be viewed as an instance
    of categorical ideas.
    From this viewpoint, Galois connections are simply
    \emph{adjunctions} in $\{0,1\}$-enriched categories, the latter
    being preorders. Monotone functions are functors, and sets of
    functions ordered pointwise are exponential objects (hence the
    claim of pointwise ordering as ``the most natural order structure on sets
    of functions'').

    With this vocabulary, Theorem~\ref{th:domain-abstraction} is a
    \emph{left Kan extension}, being the left adjoint of a pullback,
    and is structurally similar to the definition of existential
    quantifiers in categorical logic~\cite{lawvere1970quantifiers,EnrichedCateogories-Kelly1982}. The
    case of right Kan extensions / universal quantifiers is just the
    dual of Theorem~\ref{th:domain-abstraction}: a Galois connection
    from $(A\to L,\dot{\sqsubseteq})$ to $(X\to L,\dot{\sqsubseteq})$
    (or, equivalently, from $(X\to L,\dot{\sqsupseteq})$ to $(A\to
    L,\dot{\sqsupseteq})$, featuring the solution of a
    \emph{minimisation} instead of maximisation.
  \end{remark}

  \begin{example}[Dimensionality reduction]
    Another application of domain abstraction is the interpretation of
    \emph{dimensionality reduction} (for functional equations) as a
    Galois connection.

    For example, consider the (monotone) recurrence equation
    $\Phi\in\End_{\pleq_\rri}(\nn^2\to\rri)$ below, which may be viewed
    as the (monotone) size equation for a \texttt{merge} predicate (as
    part of a \texttt{merge sort} algorithm).
    %
    \begin{equation*}\footnotesize
    \begin{aligned}\footnotesize
      \Phi(f)(x,y) = {\scriptsize\begin{cases*}
      \max(f(x-1,y), f(x,y-1)) + 1 & if $x > 0 \land y > 0$ \\
      x                            & if $x > 0 \land y = 0$ \\
      y                            & if $x = 0$
      \end{cases*}}
    \end{aligned}
    \end{equation*}
    Some solvers do not support such multivariate
    equations, so we may try to simplify by introducing
    $m:\nn^2\to\nn$, $m(x,y)=x+y$. 
    Using domain abstraction and $\End$-lifting, after some computation,
    we obtain the abstracted single-variable equation
    $\Phi^\sharp\in\End_{\pleq_\rri}(\nn\to\rri)$
    defined by $\Phi^\sharp(f^\sharp)(0)=0$, $\Phi^\sharp(f^\sharp)(1)=1$,
    and $\Phi^\sharp(f^\sharp)(n) = \max(n,1+f^\sharp(n-1))$
    for $n\geq 2$,
    which has solution $(\lfp\Phi^\sharp)(n) = n$.

    Concretising back to $\nn^2\to\rri$, this gives us the upper bound
    $f_{ub}(x,y)=x+y$ on the solution of the original equation, which
    is actually the exact solution in this case.
  \end{example}

  \begin{remark}[The problems of inference]
    Of course, in this section, we have only stated that a Galois
    connection parameterised by $m$ \emph{can be defined}, but we have
    not discussed how to \emph{choose a good $m$}, nor
    how to \emph{compute the transfer functions/abstractions
    $\Phi^\sharp$ in practice}, given such an $m$.

    In the case of size abstraction, these problems have been studied
    by related
    work~\cite{granularity-shortest,caslog-shortest,low-bounds-ilps97-short,sized-types-iclp2013-shortest,plai-resources-iclp14-shortest}. The first
    question concerns \emph{automated size metric inference}
    given a program to be analysed, and the second concerns
    \emph{size abstraction} (recurrence extraction). While techniques
    exist, and powerful results have been known for decades for some
    classes of metrics $m$, many questions remain open when considering
    extensions to larger classes of metrics. The reader is invited to
    consult the related work for further insights.

    However, very little is known about the case of dimensionality
    reduction, viewed from this domain abstraction perspective, to the
    best of the authors' knowledge: these remain fresh and open
    questions.

    Further discussion on the practical considerations of
    implementing domain abstractions is beyond  the scope of this
    paper; for now, we simply observe that these phenomena
    can be interpreted as Galois connections.
  \end{remark}

\subsection{Interval functions}
\label{subsec:interval-fun}

Domain abstraction and its liftings describe the first step of
Fig.~\ref{fig:abstr-landscape}, while classical codomain abstraction
directly provides the second step, leading to interval functional
equations $\Phi_\ii\in\End(\nn^r\to\ii(\nn)^s\times\ii(\rr))$.

We now discuss in more detail these interval-valued functions, the
nuances of their different representations, the monotonicity
properties of such equations, and how they can (sometimes) be
transformed into classical numerical functional equations (the third
step of Fig.~\ref{fig:abstr-landscape}).

\phantom{.}

We will always consider \emph{monotone} interval equation.
This is justified by the paragraph and proposition below, but may be
skipped in a first read.

This comes from the fact that, for systems/programs/equations that are
interesting to us, we use the powerset construct $\pp$ to interpret
non-determinism, and \emph{assume} that the corresponding $\Phi_\Data$
is $\dot{\subseteq}$-monotone.
This is a very natural property: \emph{more possible behaviours can
only enable more behaviours, not forbid some}.
By unfolding the definition, and considering the resolution of a
recursive call within an expression, one may see that this corresponds
to the idea of \emph{monadic} resolution of non-determinism.




\begin{proposition}
  Consider a monotone equation
  $\Phi_\Data\in\End_{\dot{\subseteq}}(\Data^r\to\pp(\Data^s\times \rr))$
  (corresponding to the interpretation of a
  non-deterministic procedure with $r$ inputs, $s$ outputs, with cost
  semantics included), 
  and select a size metric $m:\Data\to\nn$.

  The best abstraction
  $\Phi_\ii\in\End(\nn^r\to\ii(\nn)^s\times \ii(\rr))$ constructed by
  using the Galois connections in Fig.~\ref{fig:abstr-landscape} is
  then
  a
  \emph{monotone} interval functional equation.
\end{proposition}
\begin{proof}
  We first need to explicit how $\Phi_\ii$ is constructed.

  First, the ideal abstraction $\Phi_\pp$ is obtained as the
  $\End$-lifting of the size abstraction for $\Phi_\Data$ (via
  Corollary~\ref{cor:nondet-fun}).
  Explicitly, if $(\alpha_{ndet},\gamma_{ndet})$ is the connection of
  Corollary~\ref{cor:nondet-fun}, then
  $\Phi_\pp:=\alpha_{ndet}\circ\Phi_\Data\circ\gamma_{ndet}$ by
  $\End$-lifting. Since $\gamma_{ndet}$, $\Phi_\Data$ and
  $\alpha_{ndet}$ are all monotone, this makes $\Phi_\pp$ a monotone
  equation.

  Now, the best abstraction $\Phi_\ii$ is obtained as the
  $\End$-lifting of the codomain abstraction for $\Phi_\pp$, via the
  (non-relational)
  powerset-to-hyperboxes abstraction, i.e. the product
  $\pp(\nn^s\times\rr)\galois{\alpha}{\gamma}\ii(\nn)^s\times\ii(\rr)$
  of interval abstractions.
  Hence, if $(\dot{\alpha},\dot{\gamma})$ is the corresponding
  connection in Proposition~\ref{prop:codomain-abstraction},
  $\Phi_\ii:=\dot{\alpha}\circ\Phi_\pp\circ\dot{\gamma}$. Like above,
  $\Phi_\ii$ is thus monotone as a composition of monotones.
\end{proof}

We now explain how to treat an interval valued function
$\dd\to\ii(\rr)$ as a pair of lower and upper bound functions in
$\dd\to\rri$.

\begin{definition}[Bounds of interval-valued function]
  Consider an interval-valued function $f:\dd\to\ii(\rr)$, which never
  outputs the empty interval. Then, we can uniquely define the lower
  bound and upper bound of each $f(x)$, which can be gathered into two
  functions $f_{lb},f_{ub}:\dd\to\rri$, so that
  $\forall x\in\dd,\,f(x)=[f_{lb}(x),f_{ub}(x)]$.

  To simplify notations, we may use the notion of \emph{an interval of
  functions}, and suggestively write this relationship as
  $f=[f_{lb},\,f_{ub}]$.

  If we allow $f$ to output the empty interval, a convention must be
  chosen, e.g. that $f_{lb}(x)=+\infty$ and $f_{ub}(x)=-\infty$
  whenever $f(x)=\varnothing$.
\end{definition}

\begin{remark}
  If we use classical intervals $\ii(\rr)$, the notions of
  interval-valued functions, pairs of functions, and interval of
  functions are not exactly equivalent.

  This comes from the issue of the empty interval, and its multiple
  representations via a pair of bounds. $f:\dd\to\ii(\rr)$ admits
  several valid representations as $(f_{lb},f_{ub})\in(\dd\to\rri)^2$,
  while the notion of interval of functions $\ii(\dd\to\rr)$ is weaker
  than that of interval-valued function, because the notion of
  $\varnothing$/$\bot$ is more granular in the latter case.

  These uninteresting complications disappear if we accept to work
  with Kaucher intervals (see Ex.~\ref{ex:kaucher}),
  where $(\dd\to\ii_K(\rr),\dot{\sqsubseteq}_{\ii_K})
   \cong \big(\dd\to(\rri\times\rri),\overset{.}{\overline{\geq_\rri\times\leq_\rri}}\big)
   \cong (\dd\to\rri,\dot{\geq}_\rri) \times (\dd\to\rri,\dot{\leq}_\rri)$,
  which may also be suggestively written $\ii_K(\dd\to\rri)$.
\end{remark}

\begin{remark}[From interval equations to classical equations]
  Thanks to these isomorphisms, an interval functional equation on
  $f:\dd\to\ii_K(\rr)$ can always be viewed as a \emph{system} of
  classical numerical function equations on the pair of unknowns
  $(f_{lb},f_{ub})\in(\dd\to\rri)^2$.

  Importantly, this transformation \emph{preserves monotonicity}, from
  $\dot{\sqsubseteq}_{\ii_K}$ to
  $\dot{\geq}_\rri\times\dot{\leq}_\rri$.

  Because of this, techniques to bound the solutions of numerical
  systems of equations (as in~\cite{order-recsolv-sas24-nourl} or this
  paper) can also be used to bound the solutions of interval
  functional equations.
\end{remark}

Note that with this construction, the equation $\Phi_\ii$ is
transformed into a \emph{coupled system of equations}.

In some favourable cases this system can be \emph{separated}, i.e.
decoupled, into two independent equations $\Phi_{\nn,lb}$ and
$\Phi_{\nn,ub}$.
It can be shown that this occurs when, in addition to being
$\dot{\sqsubseteq}_\ii$-monotone (where
$\sqsubseteq_\ii=\geq_\rri\times\leq_\rri$), the operator
$\Phi_\ii$ is also \emph{$\dot{\leq}_\ii$-monotone}, for the new order
$\leq_\ii\,:=\,\leq_\rri\times\leq_\rri$, defined as in
\cite{ModalItvBook2014-short}.

Further favourable properties can be exploited to transform
numerical-set equations $\Phi_\pp$ or $\Phi_\ii$ into purely numerical
functions equations, trying to make these equations as precise as
possible. However, these subtle analyses are outside the scope of this
paper, and will be described in future work.

\subsection{Towards \boundariesbounds domains: reminders on \emph{constraint domains}, and the case of functions}
\label{subsec:B-bound-intro}

We are now ready to set the stage for \emph{\boundariesbounds}, the
key family of abstract domains discussed in this paper, and the last
Galois connection of Fig.~\ref{fig:abstr-landscape}.

Our \boundariesbounds are an instance of the idea of
\emph{constraint domains}: complex concrete objects are abstracted by
conjunctions of constraints taken in some preselected set
$\mathfrak{C}$ of simpler, manageable constraints.
This is an ubiquitous concept in abstract interpretation, and we apply
it to the case of functions, leading to abstractions of the shape
$(\dd\to\rri,\pleq) \galois{\alphaB}{\gammaB} (\pp(\boundaries),\supseteq),$
where $\boundaries$ is a preselected set of \emph{boundary functions},
as will be detailed in the next section (abstracting pointwise-ordered
functions, or interval-valued functions, or sets of functions).
The main appeal of these \boundariesbound domains is that,
for functions, we will be able to work relatively easily with
\emph{classes} of highly non-linear numerical invariants, which is a
challenge for classical numerical abstract domains.

In this section, we simply wish to draw the reader's attention to this
reversed inclusion order $\supseteq$, to avoid confusion, and to
quickly recall a few classical examples of constraint domains to build
intuition.

The reversed inclusion is simply a consequence of the duality between
objects and constraints (or between generators and constraints), which
may be schematised as
\begin{center}
  Bigger object ($\subseteq$) $\longrightarrow$ less constraints  ($\supseteq$),

  Smaller object ($\supseteq$) $\longrightarrow$ more constraints  ($\subseteq$).
\end{center}

\begin{example}[Classical numerical abstract domains]
  Most classical numerical abstract domains follow this pattern:
  the domain of polyhedra~\cite{CousotH78-short} abstracts sets of
  values by conjunctions of \emph{affine inequalities}, so that the
  inclusion of a set $E$ in the concretisation of a polyhedron can be
  written as
  $\forall x\in E,\,\bigwedge_{i} \sum_j a_{i,j}x_j \leq b_i$.

  Similarly, abstract values in the domains of hyperboxes (intervals),
  zones and octagons may be seen as conjunctions of inequalities taken
  in classes of constraints that are smaller than the full
  class of affine inequalities~\cite{CousotCousot76-short,Mine2001-zones-short,Mine2001-octagons-short}.
\end{example}

\begin{example}[Template domains]
  Instead of considering an infinite family of functions to generate
  the set of constraints $\mathfrak{C}$, it is also possible to start
  from finitely many functions, for efficiency or expressivity
  reasons.
  Domains constructed in this way are often called \emph{template
  domains}, e.g.
  %
  the \emph{$P$-level sets}
  of~\cite{AdjeGoubaultPlevels-LMCS11-short}, which generalise the affine
  templates of~\cite{SriramAffineTemplatesVMCAI05-short}, and are an example
  of non-linear numerical abstract domain.

  For $P$-level sets, a finite family $P\subseteq(\rr^d\to\rr)$ of
  possibly non-linear functions is selected, and $\pp(\rr^d)$ is then
  abstracted by conjunctions on inequalities of the form
  $\bigwedge_{p\in P} p(x)\leq v(p)$, where $v:P\to\rri$.
\end{example}

\begin{example}[Hilbert's Nullstellensatz~\cite{Hilbert1893-Nullstellensatz}]
  \label{ex:hilbert-nullstellensatz}
  The reversion of the inclusion order is very visible in the case of
  Hilbert's Nullstellensatz, one version of which may be stated by
  giving properties of a Galois connection
  $(\pp(\mathbb{C}^k),\subseteq)\galois{\alpha}{\gamma}(\pp(\mathbb{C}[x_1,\ldots,x_k]),\supseteq)$,
  with $\mathbb{C}$ the complex numbers and $\mathbb{C}[x_1,\ldots,x_k]$
  polynomials in $k$ variables and complex coefficients.
  %
  %
  The connection is given by
  $\alpha(E) = \{p\in\mathbb{C}[x_1,\ldots,x_k]\,|\,\forall \vec{x}\in E,p(\vec{x})=0\}$
  and
  $\gamma(P) = \{\vec{x} \in \mathbb{C}^k\,|\, \bigwedge_{p\in P}p(x)=0\}$.
  The images of the closure operators
  $\gamma\circ\alpha$ and $\alpha\circ\gamma$ are, respectively,
  the set of \emph{algebraic varieties} and the set of \emph{ideals of
  polynomials}.

  This fact has been used to construct numerical abstract domains of
  \emph{polynomial equalities}~\cite{RodriguezKapurSas04-short,MullerOlm04-short}
  Once again, sets of values are abstracted by conjunctions of
  constraints taken in a preselected set (the full class of
  polynomials), which leads to a reversed inclusion order. 

  The reader may note that, for any $E\in\pp(\mathbb{C}^k)$, the set
  of polynomials $\alpha(E)$ is an \emph{infinite} set of satisfied
  constraints. However, despite being infinite, it has a lot of
  redundancies: it is an \emph{ideal}, and by Hilbert's finite basis
  theorem~\cite{Hilbert1890-Basistheorem}, it even admits a \emph{finite
  representation}.

  We will be looking for similar phenomena with \boundariesbound
  domains, to recover finite representability. In this prospect, it
  may be enlightening to note that one proof of Hilbert's basis
  theorem~\cite{gordan1899-hilbertbasisbydickson} uses Dickson's
  lemma~\cite{AichingerDicksonSurvey2020-short}, which states that
  up-closed subsets of $\pp(\nn^k)$ are finitely generated (where
  $\nn^k$ is given the product order).
\end{example}

Many other constructions in abstract interpretation, computer science
and logic follow this pattern: e.g. types viewed as constraints on
sets of values~\cite{Cousot97-types-as-abstract-interpretation-short} (weaker types
allow more values), or in a sense the duality between syntax and
semantics, via the relation between models and theories~\cite{SmithNotesGaloisSyntaxSemantics}
(more axioms are satisfied by fewer models).
We now specialise these ideas to create \boundariesbound domains.

\section{The \boundariesbound Domains: Abstracting Functions with Simpler Functions}
\label{sec:boundariesbound-domains}

We are now ready to present our \boundariesbound domains.
In this section, we define them abstractly via the corresponding
Galois connections. In the next sections, we will implement them as
abstract domains, using specific choices of $\boundaries$ (with
computable abstract operators, transfer functions, etc.), such as
polynomials and products of polynomials with exponentials.

Throughout this section, we select a subset
$\boundaries\subseteq(\dd\to\rri)$ of functions, which we call
\emph{boundary functions}. The domains we define are parametric in
this choice of $\boundaries$.
We abstract concrete 
functions $f:\dd\to\rri$ by
conjunctions of bounds $b\in\boundaries$.

For simplicity, we begin by considering only upper bounds: this yields the
domain of \emph{\boundariesubs}
$D^\sharp :=
(\pp(\boundaries),\supseteq)$, related to concrete functions by a
Galois connection $(\gammaB,\alphaB)$ defined below, such that we have
$f \pleq \gammaB(F^\sharp_{ub})\!\!\iff\!\!\bigwedge_{f^\sharp_i \in F^\sharp_{ub}} f\pleq f^\sharp_i,$
as illustrated in Fig.~\ref{fig:B-bounds-gamma-plot-1}.
\begin{definition}[\boundariesubs]
  The domain of \boundariesubs is the complete lattice
  $(\pp(\boundaries),\supseteq)$. Here,
  join is given by
  \emph{intersection} and meet by \emph{union}.
\end{definition}
To define
our desired Galois connection, we must consider \emph{all}
upper bounds $b\in\boundaries$ of $f$, since
there may be no best one ($\boundaries$ is not necessarily
closed under infinite meets, i.e. pointwise minima).
\begin{proposition}
  The following is a Galois connection.
  \begin{align*}
        (\dd\to\rri,\pleq)
        &\galois{\alphaB}{\gammaB}
         (\pp(\boundaries),\supseteq)
        \\
        f
        &\longmapsto
        \big\{f_{ub}^\sharp\in\boundaries\,\big|\,f\pleq f_{ub}^\sharp\big\}
        \\
        \!\!\Big(\vn\mapsto
             \min_{f_{ub}^\sharp\in F_{ub}^\sharp}f_{ub}^\sharp(\vn)\Big)
        &\longmapsfrom
        F_{ub}^\sharp
      \end{align*}
\end{proposition}

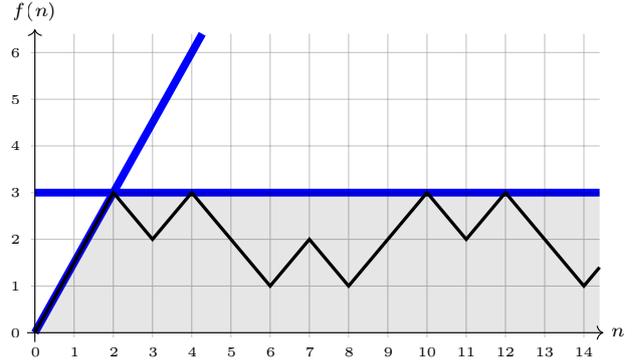
\begin{figure}
  \resizebox{\linewidth}{!}{%
      {\begin{tikzpicture}[scale=0.5,yscale=1.2]
          \begin{scope}[transparency group]
              \begin{scope}[blend mode=multiply]
        \draw[very thin,color=gray!50!white] (-0.1,-0.1) grid (14.4,6.4);
        \draw[->] (-0.2,0) -- (14.5,0) node[right] {\scriptsize $n$};
        \draw[->] (0,-0.2) -- (0,6.5) node[above] {\scriptsize $f(n)$};
        \foreach \x in {0,...,14} {
              \node [anchor=north] at (\x,-0.5-1ex) {\tiny $\x$};
          }
          \foreach \y in {0,...,6} {
              \node [anchor=east] at (-0.5-0.5em,\y) {\tiny $\y$};
          }

        \draw[-, very thick, black]
           (0,0)--(2,3)--(3,2)--(4,3)--(6,1)--(7,2)--(8,1)--(10,3)--(11,2)--(12,3)--(14,1)--(14.4,1.4);

        \draw[-, line width=1mm, blue] (0, 0) -- (2,3) -- (4.267,6.4);
        \draw[-, line width=1mm, blue] (0, 3) -- (14.4,3);
        \fill[fill=black!10] (0,0)--(2,3)--(14.4,3)--(14.4,0);

            \end{scope}
        \end{scope}
      \end{tikzpicture}}
      }
  \caption{\small Concrete functions with image in the grey fog (like
    the function in black) are bounded by the conjunction of the two
\boundariesubs
    (blue). Here, $\boundaries$ is simply the set of affine bounds.}
  \label{fig:B-bounds-gamma-plot-1}
\end{figure}

\noindent This can be generalised to support both upper and lower bounds:
elements of this new domain may then be pictured as
``flowpipes'', with relatively simple boundaries (pointwise
max/min of elements of $\boundaries$),

\begin{definition}[\boundariesbounds]
  Our full domain of \boundariesbounds is the lattice
  $(\pp(\boundaries),\supseteq)\times(\pp(\boundaries),\supseteq)$.
  It is related to concrete interval-valued functions by the following
  Galois connection (or similarly for Kaucher intervals).\\
  \begin{minipage}{\linewidth}
  \begin{equation*}\small\centering
    \begin{gathered}
        (\dd\to\ii(\rr),\dot{\sqsubseteq}_\ii)
        \galois{\alphaB}{\gammaB}
         (\pp(\boundaries),\supseteq)\times(\pp(\boundaries),\supseteq)
        \\
        \left(\vn\mapsto\left[\max_{f_{lb}^\sharp\in F_{lb}^\sharp}f_{lb}^\sharp(\vn),\;
             \min_{f_{ub}^\sharp\in F_{ub}^\sharp}f_{ub}^\sharp(\vn)\right]\right)
        \longmapsfrom
        (F_{lb}^\sharp,\,F_{ub}^\sharp)
        \\
        (f_{lb},\,f_{ub})
        \longmapsto
        \left(
          \begin{array}{l}
            \big\{f_{lb}^\sharp\in\boundaries\,\big|\,
                    \forall\vn\in\dd,\,f_{lb}^\sharp(\vn)\leq f_{lb}(\vn)\big\},\\[2pt]
            \big\{f_{ub}^\sharp\in\boundaries\,\big|\,
                     \forall\vn\in\dd,\,f_{ub}^\sharp(\vn)\geq f_{ub}(\vn)\}.
          \end{array}
        \right)
  \end{gathered}
  \end{equation*}
  \end{minipage}
\end{definition}

The full domain of \boundariesbounds is useful even when we are
only interested in obtaining upper bound information (e.g. a cost upper
bound), such as when abstracting operations like subtraction or
multiplication with negatives.
However, in this paper, for the sake of presentation, we often
simplify and focus on \boundariesubs,
with positive-valued functions $\dd\to\rri_+$.

\begin{remark}[Abstracting functions or sets of functions?]
  For a function $f:\dd\to\rri$, there is no difference between
  providing an upper bound $f_{ub}$ in $(\dd\to\rri,\pleq_\rri)$
  (like the upper boundary $f_{ub}$ of the grey area in
  Fig.~\ref{fig:B-bounds-gamma-plot-1}),
  or saying that it belongs to a set $F\in(\pp(\dd\to\rri),\subseteq)$
  that it is upper bounded (for inclusion) by a \emph{down-closed}
%
%
  set $F_{ub}\in\pp(\dd\to\rri)$ (like the set of functions whose
  graph lies within the grey area in
  Fig.~\ref{fig:B-bounds-gamma-plot-1}).

  Hence, the only difference lies in the intuitive interpretation if we
  choose to switch of vocabulary and say that \boundariesbounds
  are particular abstractions of \emph{sets of functions} (with
  concretisation in $(\pp(\dd\to\rri),\subseteq)$), rather than abstractions, i.e.
  bounds, of \emph{functions} (with concretisation in
  $(\dd\to\rri,\pleq)$ for \boundariesubs,
  or in $(\dd\to\ii(\rr),\dot{\sqsubseteq}_\ii)$ for full
  \boundariesbounds). These nuances are only a matter of
  perspective.
  In this paper, we predominantly adopt the latter,
  as it allows us to directly consider iterating an operator-equation
  in function space (we can visualise a rising function).
\end{remark}

We observe that, a priori, $\alphaB$ may introduce infinitely many
elements. This might seem problematic, as we are replacing a single
non-representable function
with an infinite collection of representable ones, potentially resulting
in a collection that is itself non-representable.

However, there is in fact significant redundancy in $\alphaB(f)$, as
illustrated in Fig.~\ref{fig:B-bounds-alpha-plot-2} and formalised in
the propositions below.  Thanks to this redundancy, finite
representations are often still possible -- much like in the case of
algebraic varieties in Example~\ref{ex:hilbert-nullstellensatz}.

Nonetheless, we are not fortunate enough to have a general version of
Hilbert's basis theorem at our disposal: in some cases, a best
\emph{finite} representation does not exist. In such scenarios, we
must discard some generators, sacrificing precision (an issue familiar
from, e.g. polyhedral abstractions of disks). This can be addressed
heuristically via a \texttt{normalise}/\texttt{widen} operator.

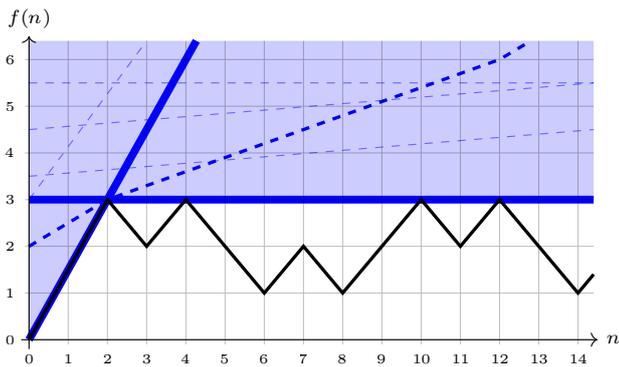
\begin{figure}
  \resizebox{\linewidth}{!}{%
      {\begin{tikzpicture}[scale=0.5,yscale=1.2]
          \begin{scope}[transparency group]
              \begin{scope}[blend mode=multiply]
        \draw[very thin,color=gray!50!white] (-0.1,-0.1) grid (14.4,6.4);
        \draw[->] (-0.2,0) -- (14.5,0) node[right] {\scriptsize $n$};
        \draw[->] (0,-0.2) -- (0,6.5) node[above] {\scriptsize $f(n)$};
        \foreach \x in {0,...,14} {
              \node [anchor=north] at (\x,-0.5-1ex) {\tiny $\x$};
          }
          \foreach \y in {0,...,6} {
              \node [anchor=east] at (-0.5-0.5em,\y) {\tiny $\y$};
          }

        \draw[-, very thick, black]
           (0,0)--(2,3)--(3,2)--(4,3)--(6,1)--(7,2)--(8,1)--(10,3)--(11,2)--(12,3)--(14,1)--(14.4,1.4);

        \draw[-, line width=1mm, blue] (0, 0) -- (2,3) -- (4.267,6.4);
        \draw[-, line width=1mm, blue] (0, 3) -- (14.4,3);

        \draw[-, very thick, blue, dashed] (0,2)--(2,3)--(12,6)--(12.8,6.4);
        \draw[-, blue!50, dashed] (0, 5.5) -- (14.4,5.5);
        \draw[-, blue!50, dashed] (0, 3) -- (3,6.4);
        \draw[-, blue!50, dashed] (0, 3.5) -- (14.4,4.5);
        \draw[-, blue!50, dashed] (0, 4.5) -- (14.4,5.5);
        \fill[fill=blue!20] (0,0)--(2,3)--(14.4,3)--(14.4,6.4)--(0,6.4);
            \end{scope}
        \end{scope}
      \end{tikzpicture}}
      }
  \caption{\small Abstraction of a concrete function $f$ (black)
    by a set of bounds, where $\boundaries$ is the set of affine functions.
    The upper bounds $b\in\boundaries$ define
    the light blue shaded region. The two solid blue lines $b_1$ and
    $b_2$ are extremal elements of this set of bounds
    (generators). The smaller dashed lines are redundant: they are
    implied by the $b_i$ as they lie above them. The larger dashed
    line is also redundant, even though it is minimal, because it is a
    convex combination of $b_1$ and $b_2$.}
  \label{fig:B-bounds-alpha-plot-2}
\end{figure}

\begin{definition}[Up- and down-closed sets]
  For any ordered set $(L,\leq)$, the set of \emph{up-closed sets} of
  elements of $L$ is denoted by $\pp_\uparrow(L)$ and
  defined as
\begingroup
  \abovedisplayskip=2pt
  \belowdisplayskip=2pt
  $$\{U\in\pp(L)\,|\,\forall x\in U,\,\forall y\in L,\,x\leq
  y\Rightarrow y\in U\}.$$
\endgroup
  The set of \emph{down-closed sets} $\pp_\downarrow(L)$ is defined similarly.
\end{definition}

\begin{definition}[Upward and downward closures]
  For any ordered set $(L,\leq)$, we define the \emph{upward closure} operator
  $\mathop{\uparrow}:\pp(L)\to\pp_\uparrow(L)$ by
$$\mathop{\uparrow}{S} = \{y\in L\,|\,\exists x \in S,\,x\leq y\}.$$
  %
  The \emph{downward closure} operator $\mathop{\downarrow}$ is defined analogously.
\end{definition}

\begin{proposition}
  For the \boundariesubs
  abstraction, the image of $\alpha_\boundaries$ is contained in the set of up-closed sets
$\pp_\uparrow(\boundaries)$.
Moreover, $\alpha_\boundaries\circ\gamma_\boundaries=\mathop{\uparrow}$.
  %
  %
The same result holds for lower bounds, using $\pp_\downarrow(\boundaries)$.
\end{proposition}
\begin{proof}
  Immediate from the definition of $\alpha_\boundaries(f)$, i.e. from $\alpha_\boundaries(f)=\{f^\sharp\in\boundaries\,|\,f\pleq f^\sharp\}$.
\end{proof}

Beyond this basic property, which allows some of the
$\alpha_\boundaries(f)$ to be represented as $\mathop{\uparrow}{S}$ for some finite set of
generators $S$, reasoning
about the abstraction
%
%
becomes even simpler when $\boundaries$ is \emph{convex}, in the
functional sense, i.e. convex within the natural vector space
structure on
$\dd \to \rr$.
As we will see in Section~\ref{sec:convexity-and-synthesis}, this
fundamental property of \emph{convexity in constraint space}
greatly simplifies the design of
transfer functions and can even allow for fully automatic transfer
function synthesis via reductions
to convex problems.

\begin{theorem}[Convexity of the set of constraints]
  \label{th:convexity-constraints}
  Suppose that the class of boundary functions $\boundaries$ is
  convex, in the sense that for all
  $b_1,b_2\in\boundaries\cap(\dd\to\rr)$ and all $\lambda\in[0,1]$, we have
  $\lambda\cdot b_1 + (1-\lambda)\cdot b_2 \in \boundaries$.

  Then, for all $f:\dd\to\rri$, the set $\alpha_\boundaries(f)$ is also convex (i.e.
  $\alpha_\boundaries(f)\cap(\dd\to\rr)$ is convex in function space).
\end{theorem}
\begin{proof}
  Let $f:\dd\to\rri$, and let $b_1,b_2\in\alpha_\boundaries(f)\cap(\dd\to\rr)$.
  Take any $\lambda\in[0,1]$.
  For all $\vn\in\dd$, we have $f(\vn)\leq b_1(\vn)$ and
  $f(\vn)\leq b_2(\vn)$, so
  $f(\vn)\leq \lambda\cdot b_1(\vn) + (1-\lambda)\cdot b_2(\vn)$.
  Thus, $\big(\lambda\cdot b_1 + (1-\lambda)\cdot b_2\big)\in\alpha_\boundaries(f)$.
\end{proof}
\begin{remark}
  This theorem
  can be refined with
  special handling of infinities, but we omit such technicalities here for simplicity.
  \end{remark}

\begin{example}
  Most of the templates $\boundaries$ discussed in this paper are convex,
  e.g.
  $\lambda\cdot\big(n\mapsto\sum a_k n^k\big) + (1-\lambda)\cdot\big(n\mapsto\sum b_k n^k\big)
   = \big(n\mapsto \sum (\lambda a_k + (1-\lambda) b_k) n^k\big).$
\end{example}

\begin{remark}
  In practice, to work with \boundariesbounds, we may need to
  replace $(\pp(\boundaries),\supseteq,\cap,\cup)$ with a more
  computable representation
  $(\pp_{\uparrow,\text{fin}}(\boundaries),\sqsubseteq^\sharp,\sqcup^\sharp,\sqcap^\sharp)$,
  where $A \sqsubseteq^\sharp B$ is a sound approximation of
$\mathop{\uparrow}{A} \supseteq\, \mathop{\uparrow}{B}$.

  This typically involves approximating the pointwise order
  $\pleq$ on $\boundaries$ 
  using a coefficient-wise comparison, derived from a parametrisation of
  $\boundaries$ (e.g. comparing polynomials via their coefficient vectors).

  By parametrisation of $\boundaries$ via $\rr^k$, we simply mean a
  function $\rr^k\to\boundaries$, such as the mapping that describes a
  polynomial
  by its vector of coefficients.
\end{remark}

\begin{remark}
  Note that if $\boundaries$ satisfies the descending
  chain condition (i.e. there is no infinite strictly decreasing
  chain of elements of $\boundaries$), then for every
  $U\in\pp_\uparrow(\boundaries)$, there exists a minimal (w.r.t.\ inclusion) set
  $A\in\pp(\boundaries)$
  such that $U=\,\mathop{\uparrow}{A}$.
  Moreover, $A$ forms an antichain.
  If, in addition,
  all antichains in $\boundaries$ are
  finite, it follows
  that all sets $\alpha_\boundaries(f)$ are finitely representable as
  the upper closure of a finite antichain.

  This phenomenon, reminiscent of Dickson's lemma, can occur when
  $\boundaries$ is parameterised by $\nn^k$, such as in the case of
  polynomial or affine bounds with non-negative integer coefficients.
  In this paper, however, we focus instead on $\boundaries$ that can
  be parameterised by vectors of \emph{real numbers}, leveraging
  \emph{convexity properties} and aiming
  to identify convex subsets of $\rr^k$ that admit finite
  representations.
\end{remark}


\section{Abstract Domains of Sequences}
\label{sec:abstr-dom-sequences}

We now present
examples of \boundariesbound domains, starting with the simple case of
the function space $\nn\to\rr$, i.e. 
univariate \emph{sequences}.
We begin
in Section~\ref{subsec:affine-bounds} with simple affine bounds, then
proceed
to \textbf{non-linear abstract domains}. We first
let $\boundaries$ be
the set of polynomials of bounded degree (parameterised initially in
the monomial basis, then in the binomial basis) in
Section~\ref{subsec:poly-bounds}. Next, we move on to exponential
bounds and products of polynomials with exponentials in
Section~\ref{subsec:exp-poly-bounds}.

As a reminder, our ultimate
goal is to compute
an overapproximation of the solution $\lfp\Phi$ of a concrete
numerical equation $\Phi\in\End_{\pleq}(\nn\to\rri)$. In this paper,
we aim to achieve this via
an abstract equation
$\Phi^\sharp\in\End_{\sqsubseteq^\sharp}(\pp_\uparrow(\boundaries))$,
operating on
the abstract domain 
of \boundariesubs.
%
Computing abstract Kleene sequences in this setting yields
an overapproximation of $\lfp\Phi^\sharp$, which itself
overapproximates
the solution $\lfp\Phi$ to our equation.
We exemplify this using
an \emph{ideal, best abstraction} $\Phi^\sharp$ in
Section~\ref{subsec:absiter-demo-1}.

However, in practice, we need 
an automated method
to compute the abstract operator $\Phi^\sharp$ from the definition of
$\Phi$. We follow
the classical compositional approach of abstract interpretation.
We first define the syntax of a simple language for defining
$\Phi$, then compute the abstract semantics (i.e. \emph{transfer
functions}) for its basic constructs, and finally compose them
to obtain a sound abstraction $\Phi^\sharp$ of $\Phi$.
For this purpose, we introduce a simple operator language in
Section~\ref{subsec:seq-language}.









\subsection{A First Example of Ideal Abstract Iteration}
\label{subsec:absiter-demo-1}

Let us begin with a small sanity check to build some intuition about
the abstractions we aim
to design, and to confirm that we can indeed hope to obtain useful
bounds using our abstract domains.

Consider the concrete domain of (non-negative valued)
sequences $D=\nn\to\rri_+$, and the simple case where $\boundaries$ is
the set of affine sequences with non-negative coefficients
(i.e. $n\mapsto an+b$, with $a,b\in\rri_+$). We study the
functional equation
\begingroup
  \abovedisplayskip=2pt
  \belowdisplayskip=2pt
$$\Phi(f)=\ite(n>0, f(f(n-1))+1, 0).$$
\endgroup
This equation is relatively challenging, as it
features a nested call (a construct known to
lead to undecidability in general;
see~\cite{tanny-talk-nestedreceq-short,nested-receq-undecidable}).
Fortunately, it admits a simple solution; $\fsol:n\mapsto n$.
We check whether we can recover this solution via abstract iteration
in Fig.~\ref{fig:absiter-nested-ideal}.

We construct our ``ideal'' $\Phi^\sharp$ in a way that is slightly
less precise than the optimal
$\alpha_\boundaries\circ\Phi\circ\gamma_\boundaries$, but that we can
more realistically
hope to mimic computationally, and which yields
more interesting results.
Indeed, to store our abstract values
$F^\sharp\in\pp_\uparrow(\boundaries)$ in finite memory, we represent
them by a finite set of generators $b_1,\ldots,b_n\in\boundaries$
such that $F^\sharp=\,\uparrow\!\!\{b_1,\ldots,b_n\}$.
Then, at each iteration step, instead of directly concretising the
full abstract value $F^\sharp$ (to compute
$\alpha_\boundaries\circ\Phi\circ\gamma_\boundaries$), we
compute $\Phi$ \emph{separately on each generator}.

\begin{figure}
\begin{center}
\resizebox{.85\linewidth}{!}{%
    \begin{tikzcd}[ampersand replacement = \&, row sep=1.5em]
      D^\sharp = \pp_\uparrow(\boundaries)
        \arrow[rr, "\Phi(\circ\gamma)\text{ on generators}", dotted, bend left=5]
      \&  \& \mathcal{P}(D)
        \arrow[ll, swap, "\text{union of }\alpha\text{'s}"', above, dotted, bend left=5]
      \\
      {\bot^\sharp = \,\uparrow\!\{0n + 0\}}
        {\arrow[rr]}
        {\arrow[d, "\substack{\alpha\circ \Phi\circ\gamma\\\rotatebox[origin=c]{-90}{$\sqsubseteq$}}", dotted]}
      \&  \& 
             {\{01^*\}}
        {\arrow[lld]}
      \\
      {\,\uparrow\!\{{\color{purple}{1n + 0}},\,0n+1\}}
        {\arrow[rr]}
      \&  \& 
             {\{{\color{purple}{1n+0}},\,02^*\}}
        {\arrow[lld]}
      \\[-1.25em]
      {
        {\,\uparrow\!\left\{\begin{array}{l}{\color{purple}{1n + 0}}, \\ {\cancel{\color{gray}{2n+0}}},\,0n+2\end{array}\right\}}
       }
        {\arrow[rr]}
      \&  \& 
             {{\{1n+0,\,03^*\}}}
        {\arrow[lld]}
      \\[-1em]
      {\,\uparrow\!\left\{\begin{array}{l}{\color{purple}{1n + 0}}, \\ {\cancel{\color{gray}{3n+0}}},\,0n+3\end{array}\right\}}
                {\arrow[d, dashed]}
      \&  \&
      \\
              {\,\uparrow\!\{{\color{purple}{1n + 0}},\,\top\} = \,\uparrow\!\{{\color{purple}{1n + 0}}\}}
      \&  \&
      \end{tikzcd}%
  }
\end{center}
%
  \phantom{.}\\[-1em]
  \phantom{.}
  \hrulefill\\[.5em]
%
  \begin{subfigure}{.24\linewidth}
    \resizebox{\linewidth}{!}{%
    \begin{tikzpicture}[]
      \begin{scope}[transparency group]
      \begin{scope}[blend mode=multiply]
        \draw[very thin,color=gray!50] (-0.1,-0.1) grid (4.4,4.4);
        \draw[->] (-0.2,0) -- (4.5,0) node[right] {\scriptsize $n$};
        \draw[->] (0,-0.2) -- (0,4.5) node[above] {\scriptsize $f(n)$};
        \foreach \x in {0,...,4} {
            \node [anchor=north] at (\x,-0.5-1ex) {\tiny $\x$};
        }
        \foreach \y in {0,...,4} {
            \node [anchor=east] at (-0.5-0.5em,\y) {\tiny $\y$};
        }
        \draw[-, very thick, blue] (0, 0) -- (4.4, 0);
        \fill[fill=blue!10] (0,0)--(4.4,0)--(4.4,4.4)--(0,4.4);
      \end{scope}
      \end{scope}
      \draw[-, line width=3pt, black] (0, 0) -- (4.4,0);
      \node at (3.5,.35) {\Large $\bot$};
    \end{tikzpicture}%
    }
  \end{subfigure}
  \hfill
    %
  \begin{subfigure}{.24\linewidth}
    \resizebox{\linewidth}{!}{%
    \begin{tikzpicture}[]
      \begin{scope}[transparency group]
      \begin{scope}[blend mode=multiply]
        \draw[very thin,color=gray!50] (-0.1,-0.1) grid (4.4,4.4);
        \draw[->] (-0.2,0) -- (4.5,0) node[right] {\scriptsize $n$};
        \draw[->] (0,-0.2) -- (0,4.5) node[above] {\scriptsize $f(n)$};
        \foreach \x in {0,...,4} {
            \node [anchor=north] at (\x,-0.5-1ex) {\tiny $\x$};
        }
        \foreach \y in {0,...,4} {
            \node [anchor=east] at (-0.5-0.5em,\y) {\tiny $\y$};
        }
        \draw[-, very thick, blue] (0, 0) -- (4.4, 0);
        \fill[fill=blue!10] (0,0)--(4.4,0)--(4.4,4.4)--(0,4.4);
      \end{scope}
      \end{scope}
      \draw[-, line width=2pt, black, dashed] (0, 0) -- (4.4,0);
      \node at (3.5,.35) {\Large $\bot$};
      \draw[-, line width=3pt, black] (0, 0) -- (1,1) -- (4.4,1);
      \node at (3.5,1.35) {\Large $\Phi(\bot)$};
    \end{tikzpicture}%
    }
  \end{subfigure}
  \hfill
    %
  \begin{subfigure}{.24\linewidth}
    \resizebox{\linewidth}{!}{%
    \begin{tikzpicture}[]
      \begin{scope}[transparency group]
      \begin{scope}[blend mode=multiply]
        \draw[very thin,color=gray!50] (-0.1,-0.1) grid (4.4,4.4);
        \draw[->] (-0.2,0) -- (4.5,0) node[right] {\scriptsize $n$};
        \draw[->] (0,-0.2) -- (0,4.5) node[above] {\scriptsize $f(n)$};
        \foreach \x in {0,...,4} {
            \node [anchor=north] at (\x,-0.5-1ex) {\tiny $\x$};
        }
        \foreach \y in {0,...,4} {
            \node [anchor=east] at (-0.5-0.5em,\y) {\tiny $\y$};
        }
        \draw[-, very thick, blue] (0, 0) -- (4.4, 0);
        \fill[fill=blue!10] (0,0)--(1,1)--(4.4,1)--(4.4,4.4)--(0,4.4);
      \end{scope}
      \end{scope}
      \draw[-, line width=3pt, black] (0, 0) -- (1,1) -- (4.4,1);
      \draw[-, line width=2pt, blue!50] (0, 0) -- (4.4,4.4);
      \draw[-, line width=2pt, blue!50] (0, 1) -- (4.4,1);
      \node at (3.5,1.35) {\color{blue}{\Large $\Phi^\sharp(\bot^\sharp)$}};
    \end{tikzpicture}%
    }
  \end{subfigure}
  \hfill
    %
  \begin{subfigure}{.24\linewidth}
    \resizebox{\linewidth}{!}{%
    \begin{tikzpicture}[]
      \begin{scope}[transparency group]
      \begin{scope}[blend mode=multiply]
        \draw[very thin,color=gray!50] (-0.1,-0.1) grid (4.4,4.4);
        \draw[->] (-0.2,0) -- (4.5,0) node[right] {\scriptsize $n$};
        \draw[->] (0,-0.2) -- (0,4.5) node[above] {\scriptsize $f(n)$};
        \foreach \x in {0,...,4} {
            \node [anchor=north] at (\x,-0.5-1ex) {\tiny $\x$};
        }
        \foreach \y in {0,...,4} {
            \node [anchor=east] at (-0.5-0.5em,\y) {\tiny $\y$};
        }
        \draw[-, very thick, blue] (0, 0) -- (4.4, 0);
        \fill[fill=blue!10] (0,0)--(1,1)--(4.4,1)--(4.4,4.4)--(0,4.4);
      \end{scope}
      \end{scope}
      \draw[-, line width=3pt, blue!50] (0, 0) -- (4.4,4.4);
      \draw[-, line width=3pt, blue!50] (0, 1) -- (4.4,1);
      \draw[->, line width=1pt, blue]   (3.75,1.1) -- (3.75,1.9);
      \draw[->, line width=1pt, purple] (3,3.2) arc (20:220:3mm) -- ++(-20:2mm);
      \node at (2.5,3.7) {\textcolor{purple}{\large \textbf{Fixed!}}};
      \draw[-, line width=2pt, purple] (0, 0) -- (4.4,4.4);
      \draw[-, line width=2pt, black] (0, 0) -- (1,2) -- (4.4,2);
    \end{tikzpicture}%
    }
  \end{subfigure}

    %
  \begin{subfigure}{.24\linewidth}
    \resizebox{\linewidth}{!}{%
    \begin{tikzpicture}[]
      \begin{scope}[transparency group]
      \begin{scope}[blend mode=multiply]
        \draw[very thin,color=gray!50] (-0.1,-0.1) grid (4.4,4.4);
        \draw[->] (-0.2,0) -- (4.5,0) node[right] {\scriptsize $n$};
        \draw[->] (0,-0.2) -- (0,4.5) node[above] {\scriptsize $f(n)$};
        \foreach \x in {0,...,4} {
            \node [anchor=north] at (\x,-0.5-1ex) {\tiny $\x$};
        }
        \foreach \y in {0,...,4} {
            \node [anchor=east] at (-0.5-0.5em,\y) {\tiny $\y$};
        }
        \draw[-, very thick, blue] (0, 0) -- (4.4, 0);
        \fill[fill=blue!10] (0,0)--(2,2)--(4.4,2)--(4.4,4.4)--(0,4.4);
      \end{scope}
      \end{scope}
      \draw[-, line width=2pt, purple] (0, 0) -- (4.4,4.4);
      \draw[-, line width=2pt, black] (0, 0) -- (1,2) -- (4.4,2);
      \draw[-, line width=1pt, dashed, blue!30] (0, 0) -- (2.2,4.4);
      \draw[-, line width=2pt, blue!50] (0, 2) -- (4.4,2);
    \end{tikzpicture}%
    }
  \end{subfigure}
  \hfill
    %
  \begin{subfigure}{.24\linewidth}
    \resizebox{\linewidth}{!}{%
    \begin{tikzpicture}[]
      \begin{scope}[transparency group]
      \begin{scope}[blend mode=multiply]
        \draw[very thin,color=gray!50] (-0.1,-0.1) grid (4.4,4.4);
        \draw[->] (-0.2,0) -- (4.5,0) node[right] {\scriptsize $n$};
        \draw[->] (0,-0.2) -- (0,4.5) node[above] {\scriptsize $f(n)$};
        \foreach \x in {0,...,4} {
            \node [anchor=north] at (\x,-0.5-1ex) {\tiny $\x$};
        }
        \foreach \y in {0,...,4} {
            \node [anchor=east] at (-0.5-0.5em,\y) {\tiny $\y$};
        }
        \draw[-, very thick, blue] (0, 0) -- (4.4, 0);
        \fill[fill=blue!10] (0,0)--(3,3)--(4.4,3)--(4.4,4.4)--(0,4.4);
      \end{scope}
      \end{scope}
      \draw[-, line width=3pt, purple] (0, 0) -- (4.4,4.4);
      \draw[-, line width=2pt, black!50] (0, 0) -- (1,3) -- (4.4,3);
      \draw[-, line width=1pt, dashed, blue!30] (0, 0) -- (1.46667,4.4);
      \draw[-, line width=1pt, blue!50] (0, 3) -- (4.4,3);
    \end{tikzpicture}%
    }
  \end{subfigure}
  \hfill
    %
  \begin{subfigure}{.24\linewidth}
    \resizebox{\linewidth}{!}{%
    \begin{tikzpicture}[]
      \begin{scope}[transparency group]
      \begin{scope}[blend mode=multiply]
        \draw[very thin,color=gray!50] (-0.1,-0.1) grid (4.4,4.4);
        \draw[->] (-0.2,0) -- (4.5,0) node[right] {\scriptsize $n$};
        \draw[->] (0,-0.2) -- (0,4.5) node[above] {\scriptsize $f(n)$};
        \foreach \x in {0,...,4} {
            \node [anchor=north] at (\x,-0.5-1ex) {\tiny $\x$};
        }
        \foreach \y in {0,...,4} {
            \node [anchor=east] at (-0.5-0.5em,\y) {\tiny $\y$};
        }
        \draw[-, very thick, blue] (0, 0) -- (4.4, 0);
        \fill[fill=blue!10] (0,0)--(4,4)--(4.4,4)--(4.4,4.4)--(0,4.4);
      \end{scope}
      \end{scope}
      \draw[-, line width=3pt, purple] (0, 0) -- (4.4,4.4);
      \draw[-, line width=2pt, black!50] (0, 0) -- (1,4) -- (4.4,4);
      \draw[-, line width=1pt, dashed, blue!30] (0, 0) -- (1.1,4.4);
      \draw[-, line width=1pt, blue!50] (0, 4) -- (4.4,4);
    \end{tikzpicture}%
    }
  \end{subfigure}
  \hfill
    %
  \begin{subfigure}{.24\linewidth}
    \resizebox{\linewidth}{!}{%
    \begin{tikzpicture}[]
      \begin{scope}[transparency group]
      \begin{scope}[blend mode=multiply]
        \draw[very thin,color=gray!50] (-0.1,-0.1) grid (4.4,4.4);
        \draw[->] (-0.2,0) -- (4.5,0) node[right] {\scriptsize $n$};
        \draw[->] (0,-0.2) -- (0,4.5) node[above] {\scriptsize $f(n)$};
        \foreach \x in {0,...,4} {
            \node [anchor=north] at (\x,-0.5-1ex) {\tiny $\x$};
        }
        \foreach \y in {0,...,4} {
            \node [anchor=east] at (-0.5-0.5em,\y) {\tiny $\y$};
        }
        \draw[-, very thick, blue] (0, 0) -- (4.4, 0);
        \fill[fill=blue!10] (0,0)--(4.4,4.4)--(0,4.4);
      \end{scope}
      \end{scope}
      \draw[-, line width=3pt, purple] (0, 0) -- (4.4,4.4);
    \end{tikzpicture}%
    }
    \vspace{-.5em}
  \end{subfigure}
  %
  \caption{Example of ideal abstract iteration. $D=\nn\to\rri_+$,
    $\Phi(f) = \ite(n>0, f(f(n-1))+1, 0)$,
    $\boundaries=\texttt{Affines}$.\\[-3em]
    \phantom{.}} 
  \label{fig:absiter-nested-ideal}
\end{figure}

Start from $\bot^\sharp = \,\uparrow\!\!\{0n + 0\}$, generated by the
constant $0$ function. When we apply $\Phi$ to this function,
we obtain the sequence $\Phi(n\mapsto 0)=\ite(n=0,0,1)$, which is represented
in regular-expression
notation by $01^*$ in Fig.~\ref{fig:absiter-nested-ideal}.
This sequence cannot be represented exactly in $\boundaries$, and is
instead represented (exactly) in $\pp_\uparrow(\boundaries)$
by the pair of generators $1n+0$ and $0n+1$
We have thus performed the first step
of (ideal) abstract iteration.

Now, we perform another step, separately on each generator.
We compute $\Phi(01^*)=02^*$, which is abstracted as
$\uparrow\!\!\{0n + 2,2n+0\}$,
while $\Phi(n\mapsto 1n+0)=(n\mapsto n)$ is a \emph{fixed point}.
%
Two interesting phenomena occur here. First, observe that the
generator $2n+0$ becomes redundant in the union
$\uparrow\!\!\{1n+0\}\cup\!\uparrow\!\!\{0n + 2,2n+0\}$, and can be
\emph{discarded}. Second, by treating each generator independently
we have \emph{already discovered an abstract postfixpoint}, namely
$\,\uparrow\!\!\{1n+0\}$.

We could stop abstract iteration here and output the upper bound
$\gamma_\boundaries(\uparrow\!\!\{1n+0\})=(n\mapsto n)$, which is
actually the \emph{optimal} bound in this case.
If we instead expect that further
precision might be gained, we can
continue iterating. This yields an abstract sequence
$\uparrow\!\!\{1n+0,0n+3\}$, $\uparrow\!\!\{1n+0,0n+4\}$,
$\uparrow\!\!\{1n+0,0n+5\}$, etc.
At this point, a widening operator must be applied to accelerate
convergence, eventually yielding
$\uparrow\!\!\{1n+0,0n+\infty\}=\,\uparrow\!\!\{1n+0,\top\}=\,\uparrow\!\!\{1n+0\}$.

This is the least fixed point
of our ideal abstract operator $\Phi^\sharp$, which, in this case,
coincides with the exact solution $\fsol=\lfp \Phi$:
no precision is lost, and abstract iteration succeeds in discovering
the \emph{optimal} bound on the solution.

\subsection{\Seq: A Simple Operator Language for Push/Pop}
\label{subsec:seq-language}

As promised in the introduction to this section, we now introduce a
simple operator language with concrete semantics in
$\End(\nn\to\rri)$. It will be used to define our concrete
operator-equations syntactically, which will allow us
to abstract operators compositionally via transfer functions.

We call this language $\Seq$ and give it a very simple structure: it
includes
three constants (for constant values, constant functions, and the
constant operator), and simple (pointwise) arithmetic operations.  Its
most important constructs are $\texttt{Pop}$ and $\texttt{Push}$,
which capture
the \emph{recursive structure} of our equation language.
Intuitively, they allow us express recursive calls such as
$f(n+1)$, $f(n-1)$, and to define base cases.

\begin{definition}[\Seq~syntax]
  \begingroup
  \abovedisplayskip=1pt
  \belowdisplayskip=0pt
  \begin{equation*}
    \begin{aligned}
          \texttt{<expr>}  ::=
            & \mid~ \texttt{Cst}~\texttt{<num>}
              \mid  \texttt{'n'}
              \mid  \texttt{'f'}\\[-1pt]
            & \mid~ \texttt{<expr>}~\Diamond~\texttt{<expr>}\\[-1pt]
            & \mid~ \texttt{Pop}~\texttt{<expr>}
              \mid  \texttt{Push}~\texttt{<num>}~\texttt{<expr>}\\[-1pt]
          \texttt{<num>} ::= &~c\in\rri
            \qquad\qquad \Diamond\in\{+,-,\times\}
    \end{aligned}
  \end{equation*}
  \endgroup
\end{definition}
\begin{definition}[\Seq~semantics]
  \begingroup
  \abovedisplayskip=1pt
  \belowdisplayskip=0pt
  \begin{equation*}\begin{aligned}
          \sem{\cdot} : \Seq &\to \End(\nn\to\rri)\\[.25em]
          %
          %
          \sem{\texttt{Cst}~c}(f)
          & = n\mapsto c\\
          \sem{\texttt{'n'}}(f)
          & = n\mapsto n\\
          \sem{\texttt{'f'}}(f)
          & = f\\[-2pt]
          \sem{e_1~\Diamond~e_2}(f)
          & = \sem{e_1}(f)~\dot{\Diamond}~\sem{e_2}(f)\\[.25em]
          %
          \sem{\texttt{Pop}}(f)
          & = n \mapsto f(n+1)\\
          \sem{\texttt{Push}~c}(f)
          & = n \mapsto  \ite(n=0,c,f(n-1))
        \end{aligned}
        \end{equation*}
   \endgroup
\end{definition}

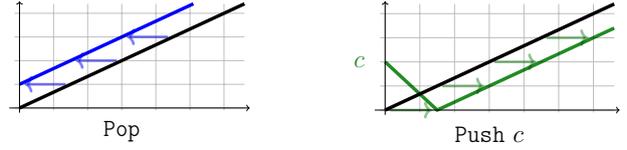
\begin{figure}[h!]
  \vspace{-1em}
  \begin{subfigure}{.45\linewidth}
    \resizebox{!}{2cm}{%
      \begin{tikzpicture}[x=1.5cm,y=.7cm]
        \begin{scope}[transparency group]
        \begin{scope}[blend mode=multiply]
        \draw[very thin,color=gray!50,ystep=1] (-0.1,-0.1) grid (4.4,4.4);
        \draw[->] (-0.2,0) -- (4.5,0) node[right] {};
        \draw[->] (0,-0.2) -- (0,4.5) node[above] {};
        \foreach \x in {0,...,4} {
              \node [anchor=north] at (\x,-0.5-1ex) {};
          }
          \foreach \y in {0,...,4} {
              \node [anchor=east] at (-0.5-0.5em,\y) {};
          }
          \draw[-, line width=1mm, blue] (0, 1)--(3.4, 4.4);
          \draw[<-, line width=0.8mm, blue!50](0.1,1)--(0.9,1);
          \draw[<-, line width=0.8mm, blue!50](1.1,2)--(1.9,2);
          \draw[<-, line width=0.8mm, blue!50](2.1,3)--(2.9,3);
          \draw[-, line width=1mm, black] (0, 0)--(4.4, 4.4);
        \end{scope}
        \end{scope}
        \node at (2,-1) {\huge\texttt{Pop}};
      \end{tikzpicture}%
    }
  \end{subfigure}
  \hfill
  \begin{subfigure}{.45\linewidth}
    \resizebox{!}{2cm}{%
      \begin{tikzpicture}[x=1.5cm,y=.7cm]
        \begin{scope}[transparency group]
        \begin{scope}[blend mode=multiply]
        \draw[very thin,color=gray!50,ystep=1] (-0.1,-0.1) grid (4.4,4.4);
        \draw[->] (-0.2,0) -- (4.5,0) node[right] {};
        \draw[->] (0,-0.2) -- (0,4.5) node[above] {};
        \foreach \x in {0,...,4} {
              \node [anchor=north] at (\x,-0.5-1ex) {};
          }
          \foreach \y in {0,...,4} {
              \node [anchor=east] at (-0.5-0.5em,\y) {};
          }
          \draw[-, line width=1mm, ForestGreen] (0,2) -- (1, 0)--(4.4, 3.4);
          \draw[->, line width=0.8mm, ForestGreen!50](0.1,0)--(0.9,0);
          \draw[->, line width=0.8mm, ForestGreen!50](1.1,1)--(1.9,1);
          \draw[->, line width=0.8mm, ForestGreen!50](2.1,2)--(2.9,2);
          \draw[->, line width=0.8mm, ForestGreen!50](3.1,3)--(3.9,3);
          \draw[-, line width=1mm, black] (0, 0)--(4.4, 4.4);
        \end{scope}
        \end{scope}
        \node at (-.5,2) {\textcolor{ForestGreen}{\huge $c$}};
        \node at (2,-1) {\huge\texttt{Push} $c$};
      \end{tikzpicture}%
    }
  \end{subfigure}
  \vspace{-.5em}
  \caption{Illustration of pop/push concrete semantics.}
  \label{fig:pop-push-sem-plot}
\end{figure}
\vspace{-1em}

\begin{remark}
  For simplicity, we will only consider here this minimal $\Seq$
  language.
  However, if  desired, the  language could  be extended  to represent
  more complex control flows, e.g. with additional constructs such as
  $\texttt{Scale}_a$ and $\texttt{Scale}_{1/a}$, interpreted as
  $f\mapsto n\mapsto f(a\cdot n)$ (interior multiplication), and
  $f\mapsto n\mapsto f\big(\lfloor\frac{n}{a}\rfloor\big)$ (interior
  integer division).
  Like $\texttt{Pop}$ and $\texttt{Push}$, these two constructs are
  restricted forms of precomposition.
  %
  One could even go further and
  allow the use of the composition operation itself ($\Diamond=\circ$
  above), although care is needed to ensure the operation is
  well-defined.
\end{remark}

\begin{example}
  The simple arithmetico-geometric recurrence equation defined by
  $f(0)=4$ and $f(n+1)=\frac{1}{2}\cdot f(n)+3$ for $n\in\nn$
  can be expressed in operator form as
  $\Phi : f \mapsto n \mapsto \ite(n=0,\, 4,\, 1/2 \cdot f(n-1) + 3),$
  which can in turn be written in $\Seq$ as
  $\texttt{Push}~4~\big(\texttt{Cst}(1/2) \times (\texttt{'f'} + \texttt{Cst}(3))\big).$
\end{example}

\subsection{Warm-up: Affine Bounds}
\label{subsec:affine-bounds}

What abstract domain do we obtain in the simple case of affine bounds?
We begin with this case to build intuition,
and move on to nonlinear abstract domains in the following sections.

\begin{definition}
  In this section, we use
  $\boundaries=\{n\mapsto+\infty\}\cup
   \big\{n\mapsto an + b\;\big|\;
         a\in\rr_+\cup\{+\infty\},\,b\in\rr_+\big\}$.
  We represent these functions using
  pairs of numbers, via a simple two-coordinate parameterisation:
  $\boundaries\cong\rri_+\times\rr_+\,\cup\,\{\top_\boundaries\}$.\footnote{The
  treatment of $\infty$ involves some details.
  For functions of the form $f:n\mapsto an+b$, we distinguish the case
  $b=+\infty$ which directly yields
  $\top_\boundaries=n\mapsto+\infty$, while for $a=+\infty$ we
  consider that $f(0)=b$ and that it is only for $n\geq 1$ that
  $f(n)=+\infty$.
  The choice ``$(+\infty)\cdot 0=0$'' gives here a nicer treatment of \texttt{Push}.}
%
  We order $\boundaries$ via the coordinate-wise
  $\sqsubseteq_\boundaries$, defined by
  $(a_1,b_1)\sqsubseteq_\boundaries(a_2,b_2)\iff a_1\leq b_1\text{ and
  }a_2\leq b_2$.
\end{definition}

\begin{remark}
  For simplicity of presentation, we only discuss here the case of upper
  bounds with positive coefficients. If we wish to interpret
  subtraction and multiplication by negative numbers, we would also need to include
  lower bounds and the full \boundariesbound domain.\footnote{This
  additionally forces us to choose conventions for operations like
  $+\infty+(-\infty)$: the safe default choice is $+\infty$ for upper
  bounds and $-\infty$ for lower bounds.}
\end{remark}

Before discussing the abstract domain
$(\pp_\uparrow(\boundaries),\supseteq)$, we first consider the
properties of $(\boundaries,\pleq)$, which will help us simplify
sets of generators.

\begin{proposition}
  The coordinate-wise order is a sound
  approximation of the concrete pointwise order $\pleq$.
  In fact, in this case, $\sqsubseteq_\boundaries$ is \emph{complete},
  i.e.
  $(a_1, b_1)\!\sqsubseteq_\boundaries\!(a_2, b_2)$ if and only if
  $\forall n,\, a_1n+b_1\leq a_2n+b_2$.
\end{proposition}

\begin{proposition}
  $(\boundaries,\sqsubseteq_\boundaries)$ is a complete lattice, where
  $\bot_\boundaries=(0,0)$, and join/meet operations are
  respectively the
  coordinate-wise max/min, i.e.
  ${\bigsqcup}\{(a_i,b_i)\}=(\max_i a_i,$\\$\max_i b_i)$
  and
  ${\bigsqcap}\{(a_i,b_i)\}=(\min_i a_i,\,\min_i b_i)$.
  To facilitate computations, we can write
  $\top_\boundaries$ as $(\_,+\infty)$.
\end{proposition}

As mentioned in the example of Section~\ref{subsec:absiter-demo-1},
for simplicity, we define our abstract transfer functions in
$\boundaries\to\pp_\uparrow(\boundaries)$, i.e. separately on each
generator, before extending by unions to
$\pp_\uparrow(\boundaries)\to\pp_\uparrow(\boundaries)$.
We compute the best possible transfer functions, which is relatively
straightforward
in the case of affine functions, a class that is preserved by most
operations in our language.

\begin{proposition}
  \label{prop:sema-affine-bounds}
  The \emph{optimal} (i.e. most precise) abstract transfer functions
  $\sema{\cdot}:\boundaries\to\pp_\uparrow(\boundaries)$ can be
  computed as follows. Notice that \texttt{Push} is the only operation
  which requires more than one generator: all other constructs can be
  interpreted exactly within affine functions (or rather optimally:
  $\times$ may lead to $\top_\boundaries$).
  \begin{itemize}[leftmargin=*]
      \item $\sema{\mathtt{Cst}~c}(\_) = \,\uparrow\!\!\big\{(0, c)\big\}$
      \item $\sema{\texttt{'n'}}(\_) = \,\uparrow\!\!\big\{(1, 0)\big\}$
      \item $\sema{\texttt{'f'}}\big(f^\sharp\big) = \,\uparrow\!\!\big\{f^\sharp\big\}$
      \item $\sema{+}\big((a_1,b_1),(a_2,b_2)\big) = \,\uparrow\!\!\big\{(a_1+a_2, b_1+b_2)\big\}$
      \item $\sema{\times}\big((a_1,b_1),(a_2,b_2)\big) = \,\uparrow\!\!\big\{(a_1a_2\infty+a_1b_2+a_2b_1, b_1b_2)\big\}$,
    \end{itemize}
  where we take the convention ``$0\times(+\infty)=0$'' in the
  expression for multiplication (the $a_1a_2\infty$ term corresponds
  to the $a_1a_2n^2$ term in the expansion of $(a_1n+b_1)\times(a_2n+b)$.
  Moreover, $\sema{\mathtt{Pop}}\big((a,b)\big) = \,\uparrow\!\!\big\{(a, b+a)\big\}$
  (indeed, $a(n+1)+b = an+(a+b)$), and
  \begingroup
    \abovedisplayskip=3pt
    \belowdisplayskip=8pt
  \begin{equation*}
    \sema{\mathtt{Push}~c}\big(\textcolor{ForestGreen}{(a,b)}\big) =
          {\small\begin{cases}
            \,\uparrow\!\!\textcolor{blue!50!black}{\big\{(a, c)\big\}} & \text{if $c \geq b-a$}\\
            \,\uparrow\!\!\textcolor{blue!50!black}{\big\{(a, b-a),\,(b-c,c)\big\}} & \text{if $c < b-a$.}
          \end{cases}}
  \end{equation*}
  \endgroup
  The two cases for $\sema{\mathtt{Push}~c}$ are illustrated in
  Fig.~\ref{fig:pop-push-asem-plot}.
\end{proposition}

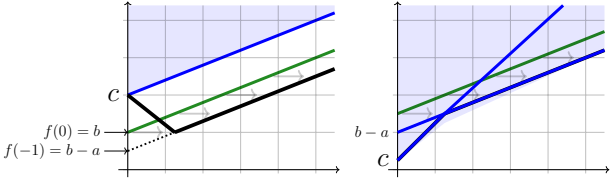
\begin{figure}
  \begin{subfigure}[t]{.45\linewidth}
    \resizebox{!}{2.5cm}{%
      \begin{tikzpicture}[x=1.25cm,y=1cm]
        \begin{scope}[transparency group]
        \begin{scope}[blend mode=multiply]
        \draw[very thin,color=gray!50,ystep=1] (-0.1,-0.1) grid (4.4,4.4);
        \draw[->] (-0.2,0) -- (4.5,0) node[right] {};
        \draw[->] (0,-0.2) -- (0,4.5) node[above] {};
        \foreach \x in {0,...,4} {
              \node [anchor=north] at (\x,-0.5-1ex) {};
          }
          \foreach \y in {0,...,4} {
              \node [anchor=east] at (-0.5-0.5em,\y) {};
          }

          \draw[-, line width=.5mm, black, dotted] (0,0.5) -- (1, 1);
          \draw[-, line width=1mm, black] (1, 1) -- (4.4, 2.7);
          \draw[-, line width=1mm, black] (0, 2) -- (1, 1);
          \draw[->, line width=0.5mm, black!20](0.1,1)--(0.75,1);
          \draw[->, line width=0.5mm, black!20](1.1,1.5)--(1.75,1.5);
          \draw[->, line width=0.5mm, black!20](2.1,2)--(2.75,2);
          \draw[->, line width=0.5mm, black!20](3.1,2.5)--(3.75,2.5);
          \draw[-, line width=.8mm, ForestGreen] (0, 1)--(4.4, 3.2);

            \draw[-, line width=.8mm, blue] (0, 2) -- (4.4, 4.2);
            \fill[fill=blue!10] (0,2)--(4.4,4.2)--(4.4,4.4)--(0,4.4);

        \end{scope}
        \end{scope}
        \node[anchor=east] (labc)  at (-.1,2) {{\huge $c$}};
        \node[anchor=east] (lab0)  at (-.5,1) {{\large $f(0)=b$}};
        \draw[->] (lab0) -- (0,1);
        \node[anchor=east] (labm1) at (-.5,0.5) {{\large $f(-1)=b-a$}};
        \draw[->] (labm1) -- (0,.5);
      \end{tikzpicture}%
    }
  \end{subfigure}
  \hfill
  \begin{subfigure}[t]{.45\linewidth}
    \resizebox{!}{2.5cm}{%
      \begin{tikzpicture}[x=1.25cm,y=1cm]
        \begin{scope}[transparency group]
        \begin{scope}[blend mode=multiply]
        \draw[very thin,color=gray!50,ystep=1] (-0.1,-0.1) grid (4.4,4.4);
        \draw[->] (-0.2,0) -- (4.5,0) node[right] {};
        \draw[->] (0,-0.2) -- (0,4.5) node[above] {};
        \foreach \x in {0,...,4} {
              \node [anchor=north] at (\x,-0.5-1ex) {};
          }
          \foreach \y in {0,...,4} {
              \node [anchor=east] at (-0.5-0.5em,\y) {};
          }

          \draw[-, line width=.5mm, black, dotted] (0,1) -- (1, 1.5);
          \draw[-, line width=1mm, black] (1, 1.5) -- (4.4, 3.2);
          \draw[-, line width=1mm, black] (0, .25) -- (1, 1.5);
          \draw[->, line width=0.5mm, black!20](0.1,1.5)--(0.75,1.5);
          \draw[->, line width=0.5mm, black!20](1.1,2)--(1.75,2);
          \draw[->, line width=0.5mm, black!20](2.1,2.5)--(2.75,2.5);
          \draw[->, line width=0.5mm, black!20](3.1,3)--(3.75,3);
          \draw[-, line width=.8mm, ForestGreen] (0, 1.5)--(4.4, 3.7);
            \fill[fill=blue!10] (0,.25)--(1,1.25)--(4.4,3.2)--(4.4,4.4)--(0,4.4);
        \end{scope}
            \draw[-, line width=.8mm, blue] (0, 1) -- (4.4, 3.2);
            \draw[-, line width=.8mm, blue] (0, .25) -- (1, 1.5) -- (3.52, 4.4);
        \end{scope}
        \node[anchor=east] (labc)  at (-.1,0.25) {{\huge $c$}};
        \node[anchor=east] (labm1) at (-.1,1) {{\large $b-a$}};
      \end{tikzpicture}%
    }
  \end{subfigure}

  \caption{Illustration of pop/push abstract semantics for $\boundaries=$ affine bounds}
  \label{fig:pop-push-asem-plot}
\end{figure}

\begin{remark}
  \label{rem:transfer-synthesis-example-teaser}
  While these transfer functions can be computed and proven optimal
  by hand, it is also possible to \emph{synthesise} them by
  automatically inferring generators of (parameterised families of)
  convex sets. Section~\ref{sec:convexity-and-synthesis} is dedicated
  to these ideas, where they are discussed in more detail.

  As a first example, consider how
  the convex problem for ${\texttt{Push}~c}$ can be formulated.
  \begin{equation*}
  \begin{aligned}
      &\big(n\mapsto un + v\big) \,\dot{\geq}\, \sem{\texttt{Push}~c}\big(n\mapsto an + b\big)\\
      \iff& 
      \begin{cases}
        \phantom{\wedge} & v \geq c\\
                 \wedge  & \sem{\texttt{Pop}}\big(n\mapsto un + v\big) \,\dot{\geq}\, \big(n\mapsto an + b)
      \end{cases}\\
      \iff& 
      \begin{cases}
        \phantom{\wedge} & v \geq c\\
                 \wedge  & \forall n\in\nn,\, u\cdot(n+1) + v \geq an + b
      \end{cases}\\
      \iff& 
      \boxed{
      \begin{cases}
        \phantom{\wedge}\, &\phantom{u\,+\;} v \geq c\\
                 \wedge \, &\phantom{u\,+\;} u \geq a\\
                 \wedge \, &u+v \geq b
      \end{cases}}
    \end{aligned}
  \end{equation*}
\end{remark}

\subsection{Polynomial Bounds}
\label{subsec:poly-bounds}

The situation is very similar in the case of polynomials, which are
also closed under most operations in our language. The main
differences are the increased number of coefficients, and the fact
that the coordinate-wise order is now incomplete, potentially limiting
simplifications in generator sets.

\begin{definition}[Polynomial bounds, bounded degree, monomial basis]
  Let $d\in\nn$. We consider
  $
    \boundaries =
    \big\{n\mapsto \sum_{k\leq d} a_k n^k\,\big|\, a_k\in\rr_+\big\} 
    \cup\big\{n\mapsto\ite(n=0,a_0,+\infty)\\\,\big|\,a_0\in\rr_+\big\}
     \cup\{n\mapsto +\infty\},
  $
  the class of polynomial of degree at most $d$, with non-negative
  real coefficients, extended with appropriate infinities.

  We first consider a parameterisation given by its coordinates (in
  the monomial basis $n^d,...,n^2,n,1$), with which we write
  $\boundaries\cong\rr_+^d\,\cup\,\{+\infty\}\times\rr_+\,\cup\,\{\top_\boundaries\}$.
  We order $\boundaries$ with the corresponding coordinate-wise
  (product) order $\sqsubseteq_{monom}$.
\end{definition}

\begin{proposition}
  $\sqsubseteq_{monom}$ is a sound approximation of $\pleq$, but it is
  not complete.
\end{proposition}
\begin{proof}
  For $f=n\mapsto \sum a_kn^k$ and $g=n\mapsto\sum b_kn^k$,
  if $\forall i,\,a_i\leq b_i$, then $\forall n,\,f(n)\leq g(n)$.

  However, the converse does not hold.
  For example, $n^2 \mathop{\dot{\geq}} n$, i.e.
  $\forall n\in\nn,\,n^2\geq n$, but $(1,0,0)\not\sqsupseteq(0,1,0)$.
  (Note that however, on the real numbers, $(1/2)^2\ngeq(1/2)$.)
\end{proof}
\begin{remark}
  If desired, it is still possible  to compute $\pleq$ for polynomials, but
  we do not need this fact here.
\end{remark}

Like before, there is no difficulty in computing the abstract
semantics of our constants and of the arithmetic operations, using
standard polynomial arithmetic.
The situation is a bit more complex for $\texttt{Pop}$ and
$\texttt{Push}$, although still manageable.
For example, for $\texttt{Pop}$, we obtain a single generator,
computed using the classical Newton binomial expansion:
  \begingroup
    \abovedisplayskip=-3pt
    \belowdisplayskip=3pt
  $$
    \sum_{k\leq d} a_k (n+1)^k = \sum_{k\leq d} \Big(\overbrace{\sum_{j\leq d} a_j \binom{j}{k}}^{p_k}\Big) n^k.
  $$
  \endgroup
%
We can also set up a convex problem to generate the transfer function
for $\texttt{Push}$, but it is quite dense.

To simplify, we can use a
useful trick from the literature on type-based amortised cost
analysis~\cite{jan-hoffman-phd-short}: instead of the classical monomial
basis of polynomials ($1$, $n$, $n^2$, ...), \textbf{use the binomial
  basis}, where shift operations behave more simply. 

\begin{definition}[Binomial basis]
  The binomial basis of the set $\rr_d[X]$ of univariate
  polynomials of degree at most $d$ is the set
  $\big\{n\mapsto\binom{n}{k}\,\big|\,k\in\nn,\,0\leq k\leq d\big\}$.
  For example, $\binom{n}{0} = 1$, $\binom{n}{1}=n$,
      $\binom{n}{2}=\frac{1}{2}n^2-\frac{1}{2}n$, and
      $\binom{n}{3}=\frac{1}{6}n^3-\frac{1}{2}n^2+\frac{1}{3}n$.
\end{definition}

\begin{definition}[Polynomial bounds, bounded degree, binomial basis]
  Let $d\in\nn$. We now consider
  $
    \boundaries =
    \big\{n\mapsto \sum_{k\leq d} a_k \binom{n}{k}\,\big|\, a_k\in\rr_+\big\} 
    \cup\big\{n\mapsto\ite(n=0,a_0,+\infty)\,\big|\,a_0\in\rr_+\big\}
     \cup\{n\mapsto +\infty\},
  $
  which is also a class of polynomials of degree at most $d$ (extended
  with infinities). We consider a parameterisation in the binomial
  basis and order $\boundaries$ with the corresponding coordinate-wise
  (product) order $\sqsubseteq_{binom}$.
  This order is still a sound (but incomplete) approximation of $\pleq$.
\end{definition}

In this representation, $\sema{\texttt{Pop}}$ remains exact, and is even
simpler.

\begin{proposition}
  $\sem{\mathtt{Pop}}\big(\sum_{k\leq d} a_k \binom{n}{k}\big)$ is the polynomial
  $\sum_{k\leq d} p_k \binom{n}{k}$ with coefficients given by $p_k=a_k+a_{k+1}$.
\end{proposition}
\begin{proof}
  By Pascal's identity, we have
  \begingroup
    \abovedisplayskip=3pt
    \belowdisplayskip=5pt
  $$
  \sum_{k\leq d} a_k \binom{n+1}{k}= \sum_{k\leq d} \Big(\underbrace{a_k + a_{k+1}}_{p_k}\Big) \binom{n}{k}.
  $$
  \endgroup
\end{proof}

The case of \texttt{Push} is more involved,
and is most conveniently addressed by synthesis (see
Section~\ref{sec:convexity-and-synthesis}), but the corresponding
convex problem is now significantly sparser
(it can be obtained by using the same push/pop trick as in the
previous section).
We simply give an example here.

\begin{example}
  For $d=2$, we obtain
  $\sema{\texttt{Push}~x}\big((a_2,a_1,a_0)\big)$
  \begingroup
    \abovedisplayskip=3pt
    \belowdisplayskip=8pt
  \begin{equation*}
      =\begin{cases}
        \,\uparrow\!\!\big\{(a_2,a_1-a_2,x),\;(x+a_1-a_0,a_0-x,x)\big\}
          \\\hspace{3.25cm}\text{if }x > a_0-a_1+a_2\\
        \,\uparrow\!\!\big\{(a_2,a_1-a_2,a_0-a_1+a_2)\big\}
          \\\hspace{3.25cm}\text{if }x = a_0-a_1+a_2\\
        \,\uparrow\!\!\big\{(a_2,a_1-a_2,a_0-a_1+a_2),\;(a_2,a_0-x,x)\big\}
          \\\hspace{3.25cm}\text{if }x < a_0-a_1+a_2
      \end{cases}
    \end{equation*}
  \endgroup

To illustrate the result in a simple case,
Fig.~\ref{fig:absiter-simple-quadratic-plots} shows abstract
iteration in the domain \boundariesubs,
using quadratic
polynomials in the binomial basis. This is applied to the equation
$\Phi(f)(n) = \ite(n=0,\,0,\,f(n-1)+n)$, which admits a quadratic
solution.
The reader may notice that the bounded region (shown in white), is no
longer a polyhedron,
but a more complex set, bounded with conjunctions of quadratic
polynomials.
%

\end{example}

\begin{figure}[h]
  \begin{subfigure}{.48\linewidth}
    \resizebox{\linewidth}{!}{%
    \begin{tikzpicture}[x=2cm,y=4.7mm,domain=0:4.5]
        \begin{scope}[transparency group]
          \begin{scope}[blend mode=multiply]
              \draw[very thin,color=gray!50,ystep=1,xstep=1] (-0.1,-0.1) grid (4.4,12.4);
              \draw[->] (-0.1,0) -- (4.5,0) node[right] {\scriptsize $n$};
              \draw[->] (0,-0.2) -- (0,12.5) node[above] {\scriptsize $f(n)$};
              \foreach \x in {0,...,4} {
                  \node [anchor=north] at (\x,-0.5-0ex) {\tiny $\x$};
              }
              \foreach \y in {0,1,...,12} {
                  \node [anchor=east] at (-0.5-0.5em,\y) {\tiny $\y$};
              }
          \begin{scope}
            \clip (-.1,-.1) rectangle (4.4,12.4);
            \fill[blue!10] (0,0) --
              plot[id=plot1] function{0} -- (8.4,40.4) -- (0,40.4) -- cycle;
          \end{scope}
          \end{scope}
        \end{scope}
        \begin{scope}
          \clip (-.1,-.1) rectangle (4.4,12.4);
        \draw[line width=3pt, black] plot[id=plot2] function{0};
        \end{scope}
        \node[anchor=south west] at (4.4,0) {$(0,0,0)$};
      \end{tikzpicture}%
    }
  \end{subfigure}
  \hfill
  \begin{subfigure}{.48\linewidth}
    \resizebox{\linewidth}{!}{%
    \begin{tikzpicture}[x=2cm,y=4.7mm,domain=0:4.5]
        \begin{scope}[transparency group]
          \begin{scope}[blend mode=multiply]
              \draw[very thin,color=gray!50,ystep=1,xstep=1] (-0.1,-0.1) grid (4.4,12.4);
              \draw[->] (-0.1,0) -- (4.5,0) node[right] {\scriptsize $n$};
              \draw[->] (0,-0.2) -- (0,12.5) node[above] {\scriptsize $f(n)$};
              \foreach \x in {0,...,4} {
                  \node [anchor=north] at (\x,-0.5-0ex) {\tiny $\x$};
              }
              \foreach \y in {0,1,...,12} {
                  \node [anchor=east] at (-0.5-0.5em,\y) {\tiny $\y$};
              }
          \begin{scope}
            \clip (-.1,-.1) rectangle (4.4,12.4);
            \fill[blue!10] (0,0) --
              plot[id=plot3] function{0} -- (8.4,40.4) -- (0,40.4) -- cycle;
          \end{scope}
          \end{scope}
        \end{scope}
        \begin{scope}
          \clip (-.1,-.1) rectangle (4.4,12.4);
        \draw[line width=2pt, black, dashed] plot[id=plot4] function{0};
        \draw[line width=2pt, black] plot[id=plot5] function{x};
        \end{scope}
        \node[anchor=south west] at (4.4,0) {$(0,0,0)$};
        \node[anchor=west] at (4.4,4.4) {$(0,1,0)$};
      \end{tikzpicture}%
    }
  \end{subfigure}

  %

  \begin{subfigure}{.48\linewidth}
    \resizebox{\linewidth}{!}{%
    \begin{tikzpicture}[x=2cm,y=4.7mm,domain=0:4.5]
        \begin{scope}[transparency group]
          \begin{scope}[blend mode=multiply]
              \draw[very thin,color=gray!50,ystep=1,xstep=1] (-0.1,-0.1) grid (4.4,12.4);
              \draw[->] (-0.1,0) -- (4.5,0) node[right] {\scriptsize $n$};
              \draw[->] (0,-0.2) -- (0,12.5) node[above] {\scriptsize $f(n)$};
              \foreach \x in {0,...,4} {
                  \node [anchor=north] at (\x,-0.5-0ex) {\tiny $\x$};
              }
              \foreach \y in {0,1,...,12} {
                  \node [anchor=east] at (-0.5-0.5em,\y) {\tiny $\y$};
              }
          \begin{scope}
            \clip (-.1,-.1) rectangle (4.4,12.4);
            \fill[blue!10] (0,0) --
              plot[id=plot6] function{x} -- (8.4,40.4) -- (0,40.4) -- cycle;
          \end{scope}
          \end{scope}
        \end{scope}
        \begin{scope}
          \clip (-.1,-.1) rectangle (4.4,12.4);
        \draw[line width=2pt, blue, dotted] plot[id=plot7] function{x};
        \draw[line width=2pt, black] plot[id=plot8] function{(x<=1)*x + (x>1)*(2*x-1)};
        \end{scope}
        \node[anchor=west, blue] at (4.4,4.4) {$(0,1,0)$};
      \end{tikzpicture}%
    }
  \end{subfigure}
  \hfill
  \begin{subfigure}{.48\linewidth}
    \resizebox{\linewidth}{!}{%
    \begin{tikzpicture}[x=2cm,y=4.7mm,domain=0:4.5]
        \begin{scope}[transparency group]
          \begin{scope}[blend mode=multiply]
              \draw[very thin,color=gray!50,ystep=1,xstep=1] (-0.1,-0.1) grid (4.4,12.4);
              \draw[->] (-0.1,0) -- (4.5,0) node[right] {\scriptsize $n$};
              \draw[->] (0,-0.2) -- (0,12.5) node[above] {\scriptsize $f(n)$};
              \foreach \x in {0,...,4} {
                  \node [anchor=north] at (\x,-0.5-0ex) {\tiny $\x$};
              }
              \foreach \y in {0,1,...,12} {
                  \node [anchor=east] at (-0.5-0.5em,\y) {\tiny $\y$};
              }
          \begin{scope}
            \clip (-.1,-.1) rectangle (4.4,12.4);
            \fill[blue!10] (0,0) --
              plot[id=plot9] function{(x*(x+1)*.5<2*x)?(x*(x+1)*.5):(2*x)} -- (8.4,40.4) -- (0,40.4) -- cycle;
          \end{scope}
          \end{scope}
        \end{scope}
        \begin{scope}
          \clip (-.1,-.1) rectangle (4.4,12.4);
        \draw[line width=2pt, blue, dotted] plot[id=plot10] function{x};
        \draw[line width=2pt, black] plot[id=plot11] function{(x<=1)*x + (x>1)*(2*x-1)};
        \draw[line width=2pt, blue] plot[id=plot12] function{2*x};
        \draw[line width=2pt, purple] plot[id=plot13] function{x*(x+1)*.5};
        \end{scope}
        \node[anchor=west, blue] at (4.4,4.4) {$(0,1,0)$};
        \node[anchor=west, blue] at (4.4,8.8) {$(0,2,0)$};
        \node[anchor=north west, purple] at (4.4,12) {$(1,1,0)$};
      \end{tikzpicture}%
    }
  \end{subfigure}
  %

  \begin{subfigure}{.48\linewidth}
    \resizebox{\linewidth}{!}{%
    \begin{tikzpicture}[x=2cm,y=4.7mm,domain=0:4.5]
        \begin{scope}[transparency group]
          \begin{scope}[blend mode=multiply]
              \draw[very thin,color=gray!50,ystep=1,xstep=1] (-0.1,-0.1) grid (4.4,12.4);
              \draw[->] (-0.1,0) -- (4.5,0) node[right] {\scriptsize $n$};
              \draw[->] (0,-0.2) -- (0,12.5) node[above] {\scriptsize $f(n)$};
              \foreach \x in {0,...,4} {
                  \node [anchor=north] at (\x,-0.5-0ex) {\tiny $\x$};
              }
              \foreach \y in {0,1,...,12} {
                  \node [anchor=east] at (-0.5-0.5em,\y) {\tiny $\y$};
              }
          \begin{scope}
            \clip (-.1,-.1) rectangle (4.4,12.4);
            \fill[blue!10] (0,0) --
              plot[id=plot13] function{(x*(x+1)*.5<2*x)?(x*(x+1)*.5):(2*x)} -- (8.4,40.4) -- (0,40.4) -- cycle;
          \end{scope}
          \end{scope}
        \end{scope}
        \begin{scope}
          \clip (-.1,-.1) rectangle (4.4,12.4);
        \draw[line width=2pt, blue] plot[id=plot14] function{2*x};
        \draw[line width=2pt, purple] plot[id=plot15] function{x*(x+1)*.5};
        \end{scope}
        \node[anchor=west, blue] at (4.4,8.8) {$(0,2,0)$};
        \node[anchor=north west, purple] at (4.4,12) {$(1,1,0)$};
      \end{tikzpicture}%
    }
  \end{subfigure}
  \hfill
  \begin{subfigure}{.48\linewidth}
    \resizebox{\linewidth}{!}{%
    \begin{tikzpicture}[x=2cm,y=4.7mm,domain=0:4.5]
        \begin{scope}[transparency group]
          \begin{scope}[blend mode=multiply]
              \draw[very thin,color=gray!50,ystep=1,xstep=1] (-0.1,-0.1) grid (4.4,12.4);
              \draw[->] (-0.1,0) -- (4.5,0) node[right] {\scriptsize $n$};
              \draw[->] (0,-0.2) -- (0,12.5) node[above] {\scriptsize $f(n)$};
              \foreach \x in {0,...,4} {
                  \node [anchor=north] at (\x,-0.5-0ex) {\tiny $\x$};
              }
              \foreach \y in {0,1,...,12} {
                  \node [anchor=east] at (-0.5-0.5em,\y) {\tiny $\y$};
              }
          \begin{scope}
            \clip (-.1,-.1) rectangle (4.4,12.4);
            \fill[blue!10] (0,0) --
              plot[id=plot16] function{(x*(x+1)*.5<3*x)?(x*(x+1)*.5):(3*x)} -- (8.4,40.4) -- (0,40.4) -- cycle;
          \end{scope}
          \end{scope}
        \end{scope}
        \begin{scope}
          \clip (-.1,-.1) rectangle (4.4,12.4);
        \draw[line width=2pt, blue!50, dashed] plot[id=plot17] function{2*x};
        \draw[line width=2pt, blue] plot[id=plot18] function{3*x};
        \draw[line width=1.5pt, blue!50] plot[id=plot19] function{2*(.5*x*(x-1))+1*(x)};
        \node at (2.5,3) {\textcolor{purple}{\large \textbf{Fixed!}}};
        \draw[->, line width=1pt, purple] (2.8,5) arc (80:-140:3mm) -- ++(120:2mm);
        \draw[line width=2pt, purple] plot[id=plot20] function{x*(x+1)*.5};
        \end{scope}
        \node[anchor=west, blue] at (4.4,8.8) {$(0,2,0)$};
        \node[anchor=north west, purple] at (4.4,12) {$(1,1,0)$};
        \node[anchor=west, blue] at (2.72,7.5) {\footnotesize $(2,1,0)$};
        \node[anchor=south east, blue] at (1.6,5) {\footnotesize $(0,3,0)$};
      \end{tikzpicture}%
    }
  \end{subfigure}
  %

  \begin{subfigure}{.48\linewidth}
    \resizebox{\linewidth}{!}{%
    \begin{tikzpicture}[x=2cm,y=4.7mm,domain=0:4.5]
        \begin{scope}[transparency group]
          \begin{scope}[blend mode=multiply]
              \draw[very thin,color=gray!50,ystep=1,xstep=1] (-0.1,-0.1) grid (4.4,12.4);
              \draw[->] (-0.1,0) -- (4.5,0) node[right] {\scriptsize $n$};
              \draw[->] (0,-0.2) -- (0,12.5) node[above] {\scriptsize $f(n)$};
              \foreach \x in {0,...,4} {
                  \node [anchor=north] at (\x,-0.5-0ex) {\tiny $\x$};
              }
              \foreach \y in {0,1,...,12} {
                  \node [anchor=east] at (-0.5-0.5em,\y) {\tiny $\y$};
              }
          \begin{scope}
            \clip (-.1,-.1) rectangle (4.4,12.4);
            \fill[blue!10] (0,0) --
              plot[id=plot21] function{x*(x+1)*.5} -- (8.4,40.4) -- (0,40.4) -- cycle;
          \end{scope}
          \end{scope}
        \end{scope}
        \begin{scope}
          \clip (-.1,-.1) rectangle (4.4,12.4);
        \draw[line width=1pt, blue!50] plot[id=plot22] function{4*x};
        \draw[line width=.5pt, blue!25] plot[id=plot23] function{3*(.5*x*(x-1))+1*(x)};
        \node at (2.5,3) {\textcolor{purple}{\large \textbf{Fixed!}}};
        \draw[->, line width=1pt, purple] (2.8,5) arc (80:-140:3mm) -- ++(120:2mm);
        \draw[line width=2pt, purple] plot[id=plot24] function{x*(x+1)*.5};
        \end{scope}
        \node[anchor=north west, purple] at (4.4,12) {$(1,1,0)$};
        \node[anchor=west, blue] at (2.4,7.5) {\scriptsize $(3,1,0)$};
        \node[anchor=south east, blue] at (1.6,6) {\scriptsize $(0,4,0)$};
      \end{tikzpicture}%
    }
  \end{subfigure}
  \hfill
  \begin{subfigure}{.48\linewidth}
    \resizebox{\linewidth}{!}{%
    \begin{tikzpicture}[x=2cm,y=4.7mm,domain=0:4.5]
        \begin{scope}[transparency group]
          \begin{scope}[blend mode=multiply]
              \draw[very thin,color=gray!50,ystep=1,xstep=1] (-0.1,-0.1) grid (4.4,12.4);
              \draw[->] (-0.1,0) -- (4.5,0) node[right] {\scriptsize $n$};
              \draw[->] (0,-0.2) -- (0,12.5) node[above] {\scriptsize $f(n)$};
              \foreach \x in {0,...,4} {
                  \node [anchor=north] at (\x,-0.5-0ex) {\tiny $\x$};
              }
              \foreach \y in {0,1,...,12} {
                  \node [anchor=east] at (-0.5-0.5em,\y) {\tiny $\y$};
              }
          \begin{scope}
            \clip (-.1,-.1) rectangle (4.4,12.4);
            \fill[blue!10] (0,0) --
              plot[id=plot25] function{x*(x+1)*.5} -- (8.4,40.4) -- (0,40.4) -- cycle;
          \end{scope}
          \end{scope}
        \end{scope}
        \begin{scope}
          \clip (-.1,-.1) rectangle (4.4,12.4);
        \node at (2.5,3) {\textcolor{purple}{\large \textbf{Fixed!}}};
        \draw[->, line width=1pt, purple] (2.8,5) arc (80:-140:3mm) -- ++(120:2mm);
        \draw[line width=2pt, purple] plot function{x*(x+1)*.5};
        \end{scope}
        \node[anchor=north west, purple] at (4.4,12) {$(1,1,0)$};
      \end{tikzpicture}%
    }
  \end{subfigure}
  \caption{Example of simple iteration with quadratic bounds. $\Phi(f)(n) = \ite(n=0,\,0,\,f(n-1)+n)$.}
  \label{fig:absiter-simple-quadratic-plots}
\end{figure}
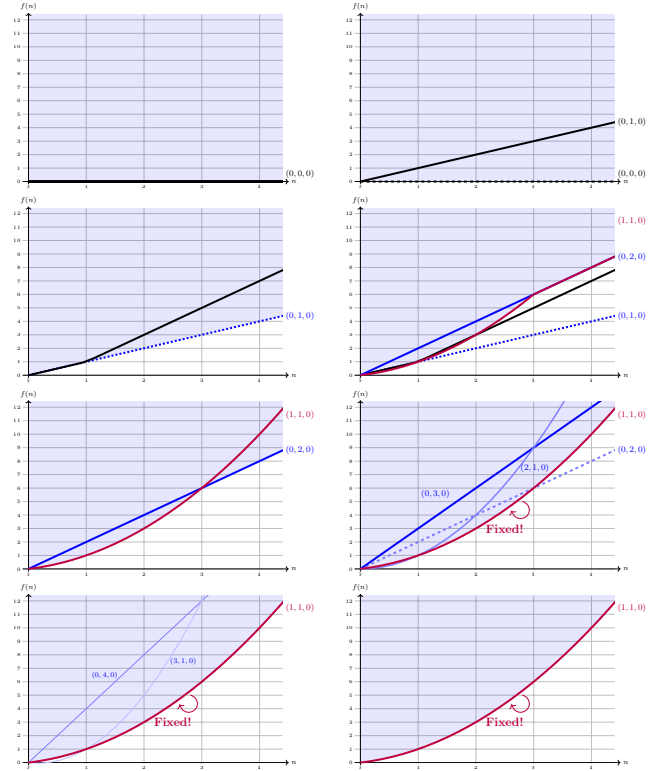

\subsection{Beyond Polynomials: Exponential Polynomials}
\label{subsec:exp-poly-bounds}

We may go even further and consider sets $\boundaries$ of boundaries
composed of all elements of the form $n\mapsto \sum_{b=1}^m a_b b^n$
(exponentials, spanned by bounded integer bases), or, even better,
$n\mapsto \sum_{b=1}^m\sum_{k=0}^d b^n n^k$ (exponential polynomials,
i.e. products of polynomials with exponentials).
In full generality, it would be desirable to allow arbitrary bases
rather than a finite number of them, or even work with more expressive
templates such as $n\mapsto \sum a_{b,k,e,f} b^n n^{k}\log(en+f)$.
However, for practical reasons,
it is convenient to let $\boundaries$ form a \emph{finite-dimensional
vector space}, which simplifies
transfer function design.

The principles that allow us to compute abstract semantics $\sema{.}$
of $\Seq$ for such exponentials and exponential polynomials are very
similar to those in previous sections: in several cases (e.g. $+$) we
immediately obtain it from the stability properties of $\boundaries$.
To go further, we can apply convexity arguments
and some transfer function synthesis.

In the case of exponentials and exponential polynomials, we can also
simplify (and make the optimisation problems more sparse) by using a
trick similar to the binomial basis used for polynomials. Once again,
this technique originates in the context of type-based automated
amortised resource analysis (AARA), in more recent
work~\cite{Kahn20-short}. The idea is to use the \emph{Stirling numbers of
the $2^\text{nd}$ kind}, which grow exponentially, and satisfy a
Pascal-like identity that makes them well-behaved under shifts.

\begin{definition}
  \textbf{Stirling numbers of the $2^\text{nd}$ kind}, defined for
  each $n,k\in\nn$ and denoted by $\stirling{n}{k}$, are
  given by the following recurrence relations
  \begingroup
    \abovedisplayskip=3pt
    \belowdisplayskip=3pt
  \begin{equation*}
    \begin{gathered}\textstyle
      \stirling{n}{n} = 1 \text{ for } n\geq 0,
      \qquad
      \stirling{n}{0} = \stirling{0}{n} = 0 \text{ for } n > 0,\\
      \textstyle
      \stirling{n+1}{k} = k\stirling{n}{k} + \stirling{n}{k-1}
        \text{ for } 0<k<n.
    \end{gathered}
  \end{equation*}
  \endgroup
They grow exponentially in $n$ for fixed $k$:
  $\stirling{n}{k}$ is asymptotically equivalent to $\frac{1}{k!}k^n$.

  Moreover, the sets of functions on $n$ spanned (as vector spaces) by $\{1, 2^n, 3^n, ... m^n\}$
  and by $\{\stirling{n+1}{1},...,\stirling{n+1}{m}\}$ are the same.
  We can change basis between these representations and provide
  explicit formulas for $\stirling{n}{k}$, see
  e.g.~\cite{ConcreteMathematics}.
\end{definition}

Using this, we can work in the vector space
$\boundaries=\Vect\{\stirling{n+1}{1},...,\stirling{n+1}{m}\}$,
extended with appropriate top elements, as before.
The \texttt{Pop} operation is simple. 
\begin{proposition}
  $\sem{\mathtt{Pop}}\big(\sum_{b\leq m} a_b \stirling{n+1}{b}\big)$ is the function
  $\sum_b p_b \stirling{n+1}{b}$, with coefficients given by 
  $p_b=b\cdot a_b+a_{b+1}$.
\end{proposition}

This approach extends naturally
to exponential polynomials in
stirling+binomial basis, using
$$\boundaries=\Vect\left\{\stirling{n+1}{b}\binom{n}{k}\;\middle|\;1\!\leq\! b\!\leq\!m, 0\!\leq\! k\!\leq\! d\right\}.$$
Here, similar computations can be performed by observing
 $\stirling{n+1}{b}\binom{n+1}{k} 
  = b\stirling{n}{b}\binom{n}{k} + b\stirling{n}{b}\binom{n}{k-1} +
    \stirling{n}{b-1}\binom{n}{k} + \stirling{n}{b-1}\binom{n}{k-1}
 $.

\subsection{Widenings}
\label{subsec:widenings}

We have not yet discussed widenings, which are crucial in practice for
ensuring the convergence of ascending Kleene sequences in the domain
$(\pp_\uparrow(\boundaries),\supseteq)$.
However, designing widenings is a highly heuristic task, best
approached in parallel with implementation and (automated)
experimentation.  We therefore present here only a few core ideas --
sufficient to guarantee finite-time convergence -- and defer their
evaluation and tuning to future work.

There are two primary sources of non-termination when working in
$\pp_\uparrow(\boundaries)$, especially when sets are represented
as upper closures of a finite set of generators.
\begin{enumerate}
\item The cardinality of the generator set may grow without bound
    (no ascending chain condition for sets of generators, no finite
  antichain condition in $(\boundaries,\pleq)$, no descending chain
  condition for sets of valid bounds).
\item The generators themselves may be repeatedly updated and grow
  unboundedly
    (no ascending chain condition in $(\boundaries,\pleq)$),
    as illustrated in Section~\ref{subsec:absiter-demo-1}.
\end{enumerate}

To address the first issue, we adopt a classical strategy used for
powerset
domains:\footnote{See~\cite{DBLP:journals/sttt/BagnaraHZ07} for references and advanced discussions
on widenings for powerset domains.}
fix an arbitrary integer bound \texttt{max-nb-constraints} on the
number of generators.  When this bound is exceeded, we either drop or
join new constraints, selecting them in arbitrary order.

To address the second issue, as mentioned in
Section~\ref{subsec:absiter-demo-1}, we only consider the case where
$\boundaries$ is parameterised by finite vectors of numbers,
and propose using parameter-wise \textbf{stability widening},
possibly combined with thresholds and delayed widening -- this is the
most common widening for intervals~\cite{CousotCousot76-short,MineTutorialBook2017-short}, where it is
called \emph{standard widening}.
%
Intuitively, the effect of $\nabla_{\text{stability}}$ is to track the
evolution of each parameter of each constraint: once a parameter is
updated repeatedly, its value is increased to the
next threshold, and ultimately
to $+\infty$.

We may also try to directly use precise polyhedra
widenings~\cite{BHRZ05-SCP} on the convex set of parameter vectors.

Finally, to accelerate convergence and potentially improve precision,
we may choose to do as mentioned in
Section~\ref{subsec:absiter-demo-1}, and exit early if we detect
\textbf{(post)fixed generators}, i.e. generators $b\in\boundaries$ of
an abstract value $F^\sharp\in\pp_\uparrow(\boundaries)$ such that
$b\in \sema{\texttt{Prog}}(b)$ for the formulation
$\sema{\cdot}:\boundaries\to\pp_{\uparrow}(\boundaries)$ of the
abstract semantics.






\section{Convexity and Transfer Function Synthesis}
\label{sec:convexity-and-synthesis}


As observed in Theorem~\ref{th:convexity-constraints}, when the class
of boundary functions $\boundaries$ is convex, the set
$\alpha_\boundaries(f)$ of all constraints in $\boundaries$ satisfied
by a concrete $f$ is itself a convex set.

Moreover, designing an (optimal) abstract transfer function
$\transfa:\boundaries\to\pp_\uparrow(\boundaries)$ corresponding
to a construct with concrete semantics
$\transfc:(\dd\to\rri)\to(\dd\to\rri)$ can be stated as the
following problem.
\begin{center}\fbox{
  \begin{minipage}[c]{.5\linewidth}\centering
    Given a function $b\in\boundaries$,\\
    find \emph{all} $f^\sharp\in\boundaries$ such that\\
    $f^\sharp\mathop{\dot{\geq}}\transfc(b).$
  \end{minipage}%
  }
\end{center}
This can be easily generalised to multi-input constructs,
and the set of solutions to this problem
is convex.
Note that it is acceptable to discover \emph{some} but not all
solutions: 
this yields
a non-optimal but sound transfer function.
In particular, it is acceptable to replace $\pleq$ in the above
problem with a sound approximation $\sqsubseteq_\boundaries$, and/or
to perform integer-to-real relaxations.
Now, assume we are given a parametrisation $p:\rr^k \to \boundaries$
which preserves the notion of convexity from constraint space
$\boundaries$ to parameter space $\rr^k$, i.e. such that
$F^\sharp\in\pp(\boundaries)$ is convex whenever
$p^{-1}(F^\sharp)\in\pp(\rr^k)$ is convex, e.g. if $p$ is linear,
as in all prior examples.
Transfer function design can then be reformulated
as the following problem, whose set of solutions is convex in
$\rr^k$.\\[-1.75em]
\begin{center}\fbox{
  \begin{minipage}[c]{.5\linewidth}\centering
    Given a vector $\vec{y}\in\rr^k$,\\
    find \emph{all} $\vec{x}\in\rr^k$ such that\\
    $p(\vec{x})\mathop{\dot{\geq}}\transfc(p(\vec{y})).$
  \end{minipage}%
  }
\end{center}
Therefore, \textbf{transfer function synthesis} reduces to the
question of (automatic) discovery of a set of generators for this
\emph{finite-dimensional convex set}, or more exactly this family of
convex sets parameterised by $\vec{y}$.


\vspace{.25em}

For many classes of expressions allowed
in $p$ and $\transfc$, the above problem is tractable for tools such
as CAS
augmented with SMT solvers and
advanced algebraic decision procedures, e.g. for quantifier
elimination.
For example, we have performed transfer function synthesis experiments
in two CAS, \texttt{Mathematica}~\cite{mathematica-v13-2-short} and
\texttt{Sage}~\cite{sagemath-short}, by explicitly encoding the semantics of
several concrete transfer functions $\transfc$ and explicitly giving
various families of boundaries $\boundaries$.
These CAS then possess primitives that allow to
$(1)$ query for instances of parameters that solve the above problem,
$(2)$ provide a finite representation of the set of solutions of the
above problem (via quantifier elimination and other reductions),
or even $(3)$ discover minimal solutions and express these sets via
generators.

Of course, the problem does not always admit a minimal
set of solutions that generate all others under
convexity, even when the solution set is finitely representable.
In such cases, we may (arbitrarily or automatically) select a
finite number of generators.\footnote{This is reminiscent of the classical example of
approximating a disk with a polyhedron.
This is not a coincidence (see
Rem.~\ref{rem:convexity-for-relational-constraint-domains}).}

As the reader might expect, in general, these synthesis operations can
become computationally expensive.
Nevertheless, it is important to note that synthesis is required
\emph{only once}, during the \emph{design} our abstract domain.  For
example, we may begin by specifying that we want to consider the case
where $\boundaries$ is the set of bivariate quartic polynomials
$\rr_4[X,Y]$.  We can then apply transfer function synthesis \emph{a single
time} to derive all relevant abstract transfer functions,
e.g. for some extension of the language $\Seq$.  Finally, the abstract
domain can be implemented, using the automatically generated transfer
functions, with no further calls to the CAS required at analysis
time.\footnote{In addition to alleviating the burden of manual abstract domain
\emph{design}, transfer function synthesis can be used to improve
\emph{precision}, by generating transfer function for more than just
the most basic syntactic constructs, see related
work~\cite{Reps2004-short,Monniaux_LMCS10-short,BrauerKingLMCS12-short,AmurthOOPSLA22-short}.}


\vspace{.25em}

Nevertheless, it is desirable to speed up the process of transfer
function synthesis. By relaxing full automation and allowing ourselves
to reformulate the above optimisation problems, we can drastically
simplify the queries sent to solvers.
Indeed, a key difficulty of the above problem lies in the use of the
order $\pleq$, which introduces an extra quantifier of the form
``$\forall n\in\nn,\,\texttt{<...>}$'' in our queries (or possibly a
relaxed variant like
``$\forall n\in \rr_+$'').
The problem becomes significantly easier
if we find a way to discard the variable $n$, and formulate the
problem purely in terms of inequalities over the parameters $\vec{x}$ and
$\vec{y}$.
This is precisely what we did in our ``warm-up'' example of
Remark~\ref{rem:transfer-synthesis-example-teaser}, which gave the
abstract semantics of $\sema{\mathtt{Push}~c}$ in
Proposition~\ref{prop:sema-affine-bounds} from
Section~\ref{subsec:affine-bounds}.
We now present two further examples, corresponding to transfer
function synthesis in
Sections~\ref{subsec:poly-bounds}~and~\ref{subsec:exp-poly-bounds}.

\begin{example}
  For univariate
  polynomials of degree $d$ written in the monomial basis, the convex
  problem for synthesising $\sema{\texttt{Push}~x}$ on the input $\sum
  a_k n^k$, with unknowns $(p_d,...,p_1,p_0)$, can be written
  \begingroup
  \abovedisplayskip=4pt
  \belowdisplayskip=2pt
  $$
    {\small
    \begin{cases}
      & \hfill p_0 \geq x\phantom{_k,} \\
      \wedge & \forall k \in [0,d],\; {\displaystyle \sum_{j\leq d}  p_j \binom{j}{k}} \geq a_k,
    \end{cases}}
  $$
\endgroup
  following the same Pop/Push trick used in
  Remark~\ref{rem:transfer-synthesis-example-teaser}, and using
  relaxation to only enforce coordinate-wise order.
  Note, however, that the corresponding matrix is very dense.
  By contrast, using the \emph{binomial} basis presentation, we obtain
\begingroup
  \abovedisplayskip=2pt
  \belowdisplayskip=2pt
  \abovedisplayshortskip=2pt
  \belowdisplayshortskip=2pt
  $$\small
    \begin{cases}
      & \hfill p_0 \geq x \phantom{_d,}\\
      \wedge & \forall k<d,\;p_k + p_{k+1} \geq a_k\phantom{,}\\\
      \wedge & \hfill p_d \geq a_d,
    \end{cases}
  $$
\endgroup
  which is much more sparse, and whose generators can in fact even
  reasonably be computed by hand.
\end{example}

\begin{example}
  For exponential polynomials in the basis of binomials and Stirling
  numbers, we can synthesise the abstract transfer function
  $\sema{\texttt{Push}~x}$ applied to the function
  $\big(\sum_{k,b}a_{b,k}\stirling{n+1}{b}\binom{n}{k}\big)$ by
  computing generators (in $\vec{p_{b,k}}$) for the family of convex
  sets given by the following inequalities (with the convention
  $p_{m+1,k}=p_{b,d+1}=0$ for all $b,k$ to simplify notation).
    \begingroup
  \abovedisplayskip=2pt
  \belowdisplayskip=4pt
  \abovedisplayshortskip=2pt
  \belowdisplayshortskip=4pt
  $$\small
    \begin{cases}
      & \hfill p_{1,0} \geq x \phantom{_{b,k},}\\
      \wedge & \forall b,\,\forall k,\;
       b p_{b,k} + b p_{b,k+1} + p_{b+1,k} + p_{b+1,k+1} \geq a_{b,k},
    \end{cases}
  $$
\endgroup
\end{example}

\begin{remark}[Convexity for Relational Domains]
  \label{rem:convexity-for-relational-constraint-domains}
  One may wonder what enabled us to design non-linear numerical
  abstractions for functional constraint domains, since this
  is a challenge in the classical case of numerical relations (c.f.
  Section~\ref{subsec:B-bound-intro}).
  The convexity property of Theorem~\ref{th:convexity-constraints} is
  clearly a key ingredient.

  Does this mean that Theorem~\ref{th:convexity-constraints} no longer
  holds for relational constraints?
  No.
  In fact, a variant of Theorem~\ref{th:convexity-constraints} can be
  formulated in the relational case. We conjecture that this approach
  may prove useful to enable transfer function synthesis in classical
  numerical abstract domains, and perhaps even enable novel non-linear
  abstractions of $\pp(\rr^d)$.

  We leave these directions to future work.
  %
  %
  For now, we simply offer the reader a ``constraint space
  perspective'' on the example of abstracting of a disk using
  polyhedra -- the classical example used to demonstrate that
  polyhedra do not come from a full Galois connection.

  From a purely theoretical perspective, we can in fact build a Galois
  connection from $\pp(\rr^d)$ to the set of \emph{all} affine
  constraints, just as we did for \boundariesbounds, allowing
  infinitely many constraints in principle.
  For such affine constraints, however, closed convex polyhedra in
  state space correspond dually to closed convex polyhedra
  in constraint space. 
\end{remark}
\begin{example}[Dual of a Disk]
  We may abstract the disk
  $D=\{(x,y)\in\rr^2\,|\,x^2+y^2 \leq 1\}$ 
  by the infinite set
  $\alpha(D)=\{(a,b,c)\in\rr^3\,|\,\forall(x,y)\in D,\, ax+by\leq c\}$
  of affine inequalities satisfied by $D$.

  We can compute the shape of $\alpha(D)$ in constraint space,
  and obtain $\alpha(D)=\{(a,b,c)\in\rr^3\,|\,c \geq \|(a,b)\|_2\}$,
  where $\|\cdot\|_2$ is the Euclidean norm.
  %
  In other words, the abstraction of the disk may be visualised in
  constraint space as an upward, closed convex cone.
  In fact, for any object $E\in\pp(\rr^2)$, the set $\alpha(E)$ would
  always be a convex cone in constraint space.

  There is, however, a redundancy in the parametrisation of affine
  inequalities by $\rr^3$.
  If we remove this
  redundancy, we are left with the \emph{projective space}
  $\mathbb{RP}^2$, in which $\alpha(D)$ looks like a disk (its points
  are rays of the 
  above cone).
  To recover a finite set of constraints (at the cost of precision),
  one must choose a finite number of points of the disk, and they
  should be \emph{extremal} -- i.e. belong to its boundary, which is a
  circle.
  Finally, if these points are chosen in an equipartition of the
  circle, and interpreted back from constraint space to state space, we
  obtain the classical approximation of a disk by a regular
  polyhedron.
\end{example}

More generally, this final example suggests the viewpoint of transfer
function synthesis via tessellations of boundaries in constraint
space.

\section{Abstract Domains of Functions}
\label{sec:abstr-dom-functions}

We now succinctly present how the ideas introduced in
Section~\ref{sec:abstr-dom-sequences} can be extended beyond the case
of \emph{sequences} $\nn\to\rri$, defined over a minimalistic operator
language based on \texttt{Push}/\texttt{Pop} constructs, to more
general \emph{functions} $\dd\to\rri$ on more general
domains $\dd$.
This extension is presented as a 
generalisation of the
sequence case, highlighting the additional features required.



\subsection{Multivariate Functions}

Handling multivariate functions of the form $\dd\to\rri$ with
$\dd=\nn^d$, is relatively straightforward, at least when starting
from a finite-dimensional vector space $\boundaries \cong \bigoplus_k
\rr\cdot f_k$ (extended with top elements), as in the case of the
examples in Section~\ref{sec:abstr-dom-sequences}.

To extend such \boundariesbound domain, initially defined for
univariate
functions $\nn\to\rri$ to an abstract domain for multivariate
functions $\nn^d\to\rri$, we can simply perform a \textbf{(tensor)
  product} (using the coordinate-wise order on parameters), i.e.
\begin{equation*}
    \boundaries^{\otimes d}
    := \bigotimes_{i<d} \bigoplus_k \rr\cdot f_k^{(i)}
    = \bigoplus_{k_1,...,k_d} \rr \cdot f_{k_1}^{(1)} \otimes ... \otimes f_{k_d}^{(d)}.
\end{equation*}
The main significant change 
is that more memory is required to
store abstract values: if the univariate
case was dealt with an $r$-dimensional parameterisation, we now need
$r\times d$ 
numbers 
to describe an element of
$\boundaries^{\otimes d}$.

\begin{example}
Univariate
  polynomials $\rr[n]$ in the monomial basis $P(n) = \sum_k a_k n^k$
  simply give rise to multivariate polynomials $\rr[x,y] =
  (\rr[n])^{\otimes 2}$ in the monomial basis $P(x,y) = \sum
  a_{k_x,k_y} x^{k_x}y^{k_y}$.
\end{example}

We can easily define new recursive constructs in each
direction $i$, written $\sem{\texttt{Pop}_i}$ and
$\sem{\texttt{Push}_i~c}$, e.g.\phantom{.}
$\texttt{Push}_y~c:\;\;f(x,y,z)\rightsquigarrow\ite\big(y=0,\,c,\,f(x,y-1,z)\big),$
which may be composed to construct more complex recursive calls.
Transfer functions (e.g. $\sema{\Diamond}$, $\sema{\texttt{Cst~c}}$,
$\sema{\texttt{Pop}_i}$, etc.) generalise naturally.
\begin{example}
  For polynomials, $1\cdot (x+1)^2yz = (x^2+2x+1)yz = 1\cdot x^2yz +2\cdot xyz + 1\cdot yz$.
  Similarly, in the binomial basis,
  $\binom{x+1}{k_x}\binom{y}{k_y}\binom{z}{k_z} = \binom{x}{k_x}\binom{y}{k_y}\binom{z}{k_z} + \binom{x}{k_x-1}\binom{y}{k_y}\binom{z}{k_z},$
  so more generally, 
  $\sema{\texttt{Pop}_x}\big(\sum a_{(k_x,k_y,k_z)}\binom{x}{k_x}\binom{y}{k_y}\binom{z}{k_z}\big)=$
  \begingroup
    \abovedisplayskip=2pt
    \belowdisplayskip=0pt
    $$\uparrow\!\!\Big\{\sum \big(a_{(\mathbf{k_x},k_y,k_z)}+a_{(\mathbf{k_x+1},k_y,k_z)}\big)\binom{x}{k_x}\binom{y}{k_y}\binom{z}{k_z}\Big\}.$$
  \endgroup
\end{example}

\subsection{Piecewise Functions, for Complex Control Flows}
\label{subsec:piecewise-functions}

Extending our domains to support control flows beyond the ones
associated with \emph{sequences} is more interesting, and requires
adding new features.

Note that for now, we can only represent recursive calls of the form
$f(\vec{n}-\vec{k})$, where $\vec{k}$ is a constant vector, in a way that
is identical in the whole domain $\dd$, except for base
cases which occur only when at least one input variable is $0$.
In the cost analysis of symbolic declarative programs, it is reasonable to
assume that base cases are defined for data structures of small size.
It is also reasonable to assume that recursive cases consists in taking
data, deconstructing it, and making recursive calls on data
structures of smaller size. Such assumptions are made, implicitly or
explicitly, in many cost analysers, such as~\cite{DBLP:journals/toplas/0002AH12-shortest,caslog-shortest,vh-03-shortest}.

However, this is not a sound assumption for more general static
analysis: variables inside loops, for example, may not evolve in a
global monotone way. This assumption is also a limitation in the
\emph{cost} analysis of typical \emph{imperative} programs, as
thoroughly
discussed in~\cite{montoya-phdthesis-short}.

\begin{lrbox}{\codebox}
\begin{minipage}{4cm}\small
\begin{verbatim}
while (0<i&&i<n){
  if(b)
    i++;
  else
    i--;
}
\end{verbatim}
\end{minipage}
\end{lrbox}

\begin{example}
  \label{ex:aflores-loop-1}
  Let us consider a simple multiphase loop
  presented in~\cite{montoya-phdthesis-short}, and the corresponding cost
  equation $\Phi$ below, which gives the number of loop iterations as
  a function of the initial values of the variables.
  Crucially, observe that the domain $\dd=\rr^3$ must be partitioned
  into several subdomains $\dd_i$ in order to express the equation
  faithfully.

  \fbox{\usebox{\codebox}}

  \begingroup
    \abovedisplayskip=2pt
    \belowdisplayskip=0pt
  \begin{equation*}\small
    \Phi(f)(i,n,b) = {\footnotesize\begin{cases}
          0 & \text{if $i=0$}\quad\hfill(\dd_1)\\
          0 & \text{if $i\geq n$}\quad\hfill(\dd_2)\\
          1+f(i+1,n,b) & \text{if $0<i<n \wedge b\geq 1$}\quad\hfill(\dd_3)\\
          1+f(i-1,n,b) & \text{if $0<i<n \wedge b=0$}\quad\hfill(\dd_4)
        \end{cases}}
  \end{equation*}
  \endgroup
\end{example}

To handle such cases, we propose replacing our domains of functions,
which give a single \emph{global} bound expressed as a conjunction of
$b\in\boundaries$, by a domain of \emph{piecewise functions}, where
$\dd$ is partitioned into several subdomains, and a conjunction of
bounds is associated with each
subdomain.
Note that the subdomains explicitly (syntactically) present in the
equation provide valuable information that can be used in the
analysis.

\begin{definition}[Domain of $\boundaries$-bounds by $\casesdomain$-cases]
  Let $\boundaries\subseteq(\dd\to\rri)$ as for usual \boundariesbound domains,
  and let $\casesdomain$ be a numerical abstract domain (e.g.
  polyhedra, intervals) with concretisation
  $\gamma_\casesdomain:\casesdomain\to\pp(\dd)$
  in our set of inputs.
  We define the abstract domain $\mathbb{P}(\boundaries,\casesdomain)$
  via
  $\gamma_{\mathbb{P}\boundaries}:\mathbb{P}(\boundaries,\casesdomain)\to(\dd\!\to\!\rri),$
  where elements $F^\sharp\in\mathbb{P}(\boundaries,\casesdomain)$
  are of the form $F^\sharp=\big\{(F_1,c_1),...,(F_k,c_k),F_{\text{else}}\big\}$
  with $F_i\in\pp_\uparrow(\boundaries)$ and $c_i\in\casesdomain$,
  and $\gamma_{\mathbb{P}\boundaries}(F^\sharp)(\vn)=$\\
  $\begin{cases}
        \gamma_\boundaries(F_i)(\vn)
        & \text{ if } \exists i,\vn\in\gamma_\casesdomain(c_i)
          \text{ and } \forall j\neq i,\, \vn\not\in\gamma_\casesdomain(c_j),\\
        \gamma_\boundaries(F_{\text{else}})(\vn)
        & \text{ otherwise.}
   \end{cases}$
\end{definition}

In other words, if the $c_i$ define a partition of $\dd$, $F^\sharp$
bounds concrete functions by a piecewise conjunction of bounds,
where each piece is defined by $\gamma_\casesdomain(c_i)$, and the
corresponding set of bounds is $F_i\in\pp_{\uparrow}(\boundaries)$:
\begingroup
  \abovedisplayskip=2pt
  \belowdisplayskip=0pt
\[
  f\pleq \gamma_{\mathbb{P}\boundaries}(F^\sharp)
  \Leftrightarrow
  \bigwedge_i\Big(
    \vn\in\gamma_\casesdomain(c_i)
    \Rightarrow\!\!
    \bigwedge_{b_{i,j}\in F_i} f(\vn) \leq b_{i,j}(\vn)
  \Big).
\]
\endgroup

\begin{figure}[H]
  \begingroup
  \abovedisplayskip=0pt
  \belowdisplayskip=0pt
  \abovedisplayshortskip=0pt
  \belowdisplayshortskip=0pt
  \begin{center}
    \vspace{-.65cm}
    \includegraphics[height=2cm,width=3.25cm]{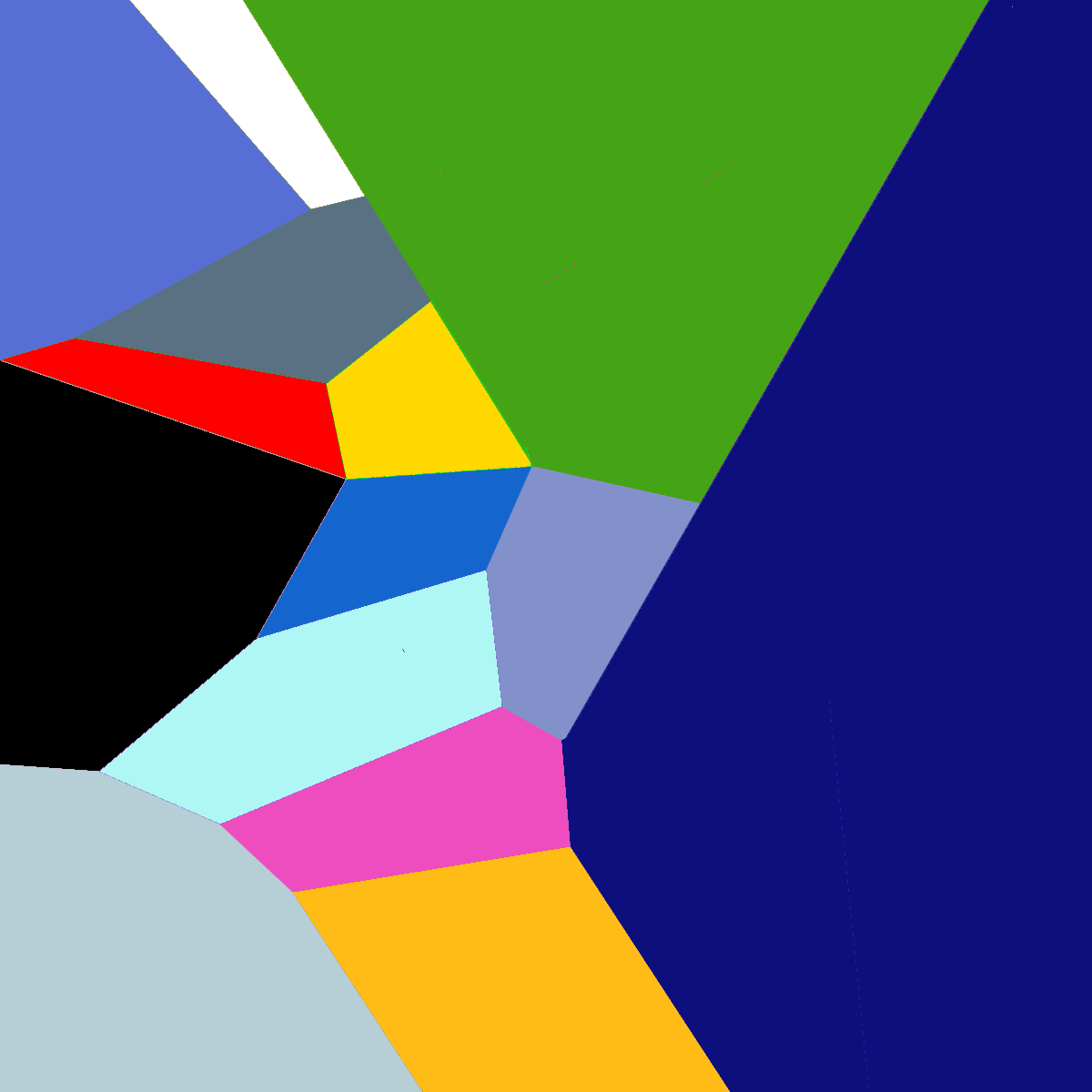}
    \vspace{-.5cm}
  \end{center}
  \caption{A Voronoi diagram as an example of partition by polyhedra.
    Our abstract domain
    $\mathbb{P}(\boundaries,\casesdomain)$ associates
    a conjunction of bounds $b\in\boundaries$ with each coloured
    piece.}
  \label{fig:voronoi}
  \endgroup
\end{figure}
\vspace{-.5cm}

\begin{remark}
  Functors for binary decision tree (BDT) abstract domains, as
  introduced in~\cite{urban2014-bdt-short,ChenCousotBDT15-short}, can be used to
  implement these ideas.
\end{remark}


The \emph{key} additional
primitives needed to implement such piecewise \boundariesbound domains
are
\begin{enumerate}[leftmargin=*,topsep=2pt]
  \item The ability to create new cases to handle recursive calls.
    Indeed, calls made in one subdomain may query the value of $f$ in
    other subdomains -- this is referred to as \emph{the interface between
    subdomains} in~\cite{order-recsolv-preprint-extended},
  \item The ability to \textbf{merge cases}: i.e. ``compute best
    $\boundaries$-approximation of $\boundaries$-cases using
    $\casesdomain$-parts''.
\end{enumerate}
In this paper, we do not aim
to provide a general solution to these challenges
for complex $\boundaries$ and $\casesdomain$.  Instead, we simply
state these requirements to help guide future work.


\begin{example}
  To illustrate how such constructs may be used, we demonstrate an
  abstract iteration using $\boundaries$ affine bounds and
  $\casesdomain$ polyhedra on the case of Ex.~\ref{ex:aflores-loop-1}.
  %
  We restrict attention to $\rri_+$ for simplicity, and initialise at
  $\bot^\sharp$ with the constant $0$ function, one per subdomain
  syntactically given by the equation.
  %
  The first step is straighforward in both concrete and abstract
  semantics,
  since $f$ resolves to $0$ everywhere, yielding
  $\Phi^\sharp(\bot^\sharp)$.

  However, the second
  iteration becomes more interesting. For example, consider $(i,n,b)\in\dd_3$: the
  recursive call $f(i+1,n,b)$ may either remain in $\dd_3$, or enter
  $\dd_2$, which resolves to different values. To capture this precisely,
  we should \textbf{split} $(\dd_3)$ into two new
  subdomains: the part $(\dd_{3,2})$ that calls $(\dd_2)$ (the
  ``interface with $\dd_2$''), and the part $(\dd_{3,3})$ that remains
  in $(\dd_3)$. A similar split applies to $(\dd_4)$.

\begin{minipage}{\linewidth}
  \vspace{-.5\baselineskip}
  {\small
  \begin{flalign*}
        \bot^\sharp &= \begin{cases}
          \{0\} & \quad(\dd_1)\\
          \{0\} & \quad(\dd_2)\\
          \{0\} & \quad(\dd_3)\\
          \{0\} & \quad(\dd_4)
        \end{cases}\\
        %
        (\Phi^\sharp)(\bot^\sharp) &= \begin{cases}
          \{0\} & \quad(\dd_1)\\
          \{0\} & \quad(\dd_2)\\
          \{1\} & \quad(\dd_3)\\
          \{1\} & \quad(\dd_4)
        \end{cases} \\
        %
        (\Phi\circ\Phi^\sharp)(\bot^\sharp) &= \begin{cases}
          0 & \hspace{-5cm}(\dd_1)\\
          0 & \hspace{-5cm}(\dd_2)\\
          \begin{cases}
            2 & \texttt{if $(i,n,b)\in\dd_3\wedge(i+1,n,b)\in\dd_3$}\\
            1 & \texttt{if $(i,n,b)\in\dd_3\wedge(i+1,n,b)\in\dd_2$}
          \end{cases}\\
          \begin{cases}
            1 & \texttt{if $(i,n,b)\in\dd_4\wedge(i-1,n,b)\in\dd_1$}\\
            2 & \texttt{if $(i,n,b)\in\dd_4\wedge(i-1,n,b)\in\dd_4$}
          \end{cases}
        \end{cases}\\
        %
        (\Phi^\sharp)^{(2)}(\bot^\sharp) &= \begin{cases}
          \{0\} & \quad(\dd_1)\\
          \{0\} & \quad(\dd_2)\\
          \{n-i\} & \quad(\dd_3)\\
          \{i\} & \quad(\dd_4)
        \end{cases}
        \qquad\longleftrightarrow \fsol(i,n,b)~!
  \end{flalign*}
  }
  \vspace{.5\baselineskip}
  \end{minipage}

  Despite our goal of
  precision, we must limit subdomain proliferation
  to prevent combinatorial explosion. To do so, we \textbf{merge}
  again the subdomains. For this, we must find \emph{the best
  approximation} of the function values over
  $\dd_{3,2}\cup\dd_{3,3}$ by a function \emph{representable over
  $\dd_3$ as a conjunction of affine bounds}.

  In our case, this \texttt{merge} operation can be implemented using
  the \texttt{join} of polyhedra, and directly leads to the function
  shown above, which in this case is the exact
  solution to the equation of Ex.~\ref{ex:aflores-loop-1}.
\end{example}

\subsection{Towards Continuous Systems (Future Work)}
\label{subsec:continous-systems}

As a final direction for extension, we consider
\emph{continuous systems}, 
where this work may serve as a stepping stone.
%
%
Here, functions $f:\rr_+^{(k)}\to\rri$ are defined over domains $\dd$
of real numbers rather than integers, so that $f$ may arise as the
solution of \textbf{differential equations} (ordinary or partial).
%
%
%
As noted in the comparison with~\cite{order-recsolv-sas24-nourl}, this
direction is especially
promising for applying
\boundariesbounds, since exact input-output sampling is unavailable
for differential equations.
%
%

As a first step in this direction, we present an example where
Theorem~\ref{th:proof-principle} yields \emph{infinite-horizon} bounds
on a differential equation’s solution.

In our example, the (interval) candidate bound is asserted \emph{a
priori} and checked. An extension of \boundariesbound domains to
$\boundaries\in\pp(\rr_+\to\rri)$ could enable automated
\emph{discovery} of such bounds.



\begin{example}
  Consider the following non-linear differential equation, where
  $\dot{v}$ is the first time derivative of $v$.
  \begingroup
    \abovedisplayskip=2pt
    \belowdisplayskip=2pt
    $$v : \rr_+ \to \rr,\quad v(0) = v_0,\quad \dot{v} =  -\alpha\cdot v^2 - \beta\cdot v + \gamma,$$
  \endgroup
  To mimic the discrete case,
  we can turn this equation into an operator
  $\Phi\in\End_{\dot{\sqsubseteq}_\ii}({^*}\rr_+\to\ii({^*}\rr))$,
  inspired by non-standard analysis (c.f.~\cite{Robinson1966-nsa,Keisler2012-em,Hasuo12-nsa-CAV,KidoHasuo15-absint-nsa-short}) on the
  hyperreals $^*\rr$, which are an extension of the real numbers with
  an infinitesimal $\epsilon$.
  Given the technical overhead of working with hyperreals, we may
  instead choose to view this as a family of operators parameterised
  by (finite) real numbers
  $\Phi_\epsilon\in\End_{\dot{\sqsubseteq}_\ii}(\rr_+\to\ii(\rr))$,
  corresponding to Euler schemata with $\epsilon>0$, and make proofs
  for $\epsilon\to 0$.
  We can define $\Phi$ by $\Phi(f)(t)=$\\
  {$\begin{cases}
      [v_0, v_0] 
      \qquad
      \text{if }t<\epsilon\\

      f(t\!-\!\epsilon)-\alpha\!\cdot\!\epsilon\!\cdot\! f(t\!-\!\epsilon) \cdot_\ii f(t\!-\!\epsilon) - \beta\!\cdot\!\epsilon\!\cdot\! f(t\!-\!\epsilon) + \epsilon\!\cdot\!\gamma
      \\\phantom{[v_0, v_0]}\qquad
      \text{if }t\geq\epsilon,
    \end{cases}$}\\[2pt]
  where $\cdot_\ii$ is the monotone multiplication of intervals.
  Then, the solution of the differential equation is $\lfp\Phi$.
  We can also rewrite this operator in a different form, that is
  less natural, but diminishes the loss of precision caused by
  interval arithmetic. Indeed, with some factorisations, we can write
  the operator $\widetilde{\Phi}$, defined 
  by $\widetilde{\Phi}(f)(t)=[v_0,v_0]$
  for $t<\epsilon$ and 
  $\widetilde{\Phi}(f)(t)=$\\
  $f(t\!-\!\epsilon) \cdot_\ii
    \big((1-\epsilon\!\cdot\!\beta)\cdot_\ii[1,1] + (-\epsilon\!\cdot\!\alpha)\cdot_\ii f(t\!-\!\epsilon)\big)
    + \epsilon\cdot\gamma$
  for $t\geq\epsilon$.
  The solution is unchanged: $\lfp\tilde{\Phi}=\lfp\Phi.$

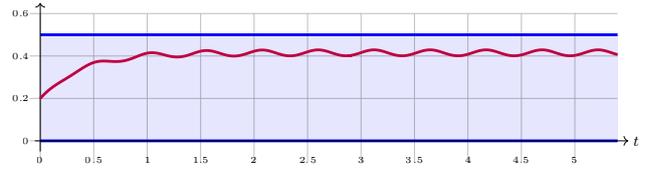
\begin{figure}[t]
  \resizebox{\linewidth}{!}{%
      \begin{tikzpicture}[x=2cm,y=4cm,domain=0:5.5]
        \begin{scope}[transparency group]
        \begin{scope}[blend mode=multiply]
        \draw[very thin,color=gray!50,ystep=.2,xstep=.5] (-.1,-.1) grid (5.4,0.6);
        \draw[->] (-.05,0) -- (5.5,0) node[right] {$t$};
        \draw[->] (0,-.05) -- (0,.65) node[above] {};
        \foreach \x in {0,0.5,...,5} {
              \node [anchor=north] at (\x,-.05) {\tiny $\x$};
          }
          \foreach \y in {0,0.2,0.4,0.6} {
              \node [anchor=east] at (-0.05,\y) {\tiny $\y$};
          }
        \begin{scope}
          \clip (-.1,-1.1) rectangle (5.4,0.6);

          \fill[fill=blue!10] (0,0)--(5.4,0)--(5.4,.5)--(0,.5);
          \draw[-, line width=1.5pt, blue] (0, 0) -- (5.4, 0);
          \draw[-, line width=1.5pt, blue] (0, .5) -- (5.4, .5);

        \end{scope}
        \end{scope}
        \end{scope}
        \begin{scope}
          \clip (-.1,-.1) rectangle (5.4,0.6);
          \draw[line width=1.5pt, purple] plot [smooth] file {ode-ex.table};
        \end{scope}
      \end{tikzpicture}%
     }

  \vspace{-3cm}

  \caption{Example with $\alpha(t)\!\in\![0,2]$, $\beta\!=\!2$, $\gamma\!=\!1$, $v_0=\frac{1}{5}$.}
  \label{fig:ode-ex-plot}
\end{figure}

  Now, consider a candidate interval bound of the form
  $\hat{f}:t\mapsto [0,M]$, with $M>0$ a fixed real.
  We can then check that $\widetilde{\Phi}\hat{f}\,\dot{\sqsubseteq}_\ii\,\hat{f}$
  whenever $v_0\in[0,M]$, $\alpha,\beta,\gamma\geq 0$ and $M \geq \gamma/\beta$,
  which gives us an infinite-horizon bound on the solution.

  The differential equation does not admit a closed-form solution if
  we generalise it with time-dependent $\alpha(t)$, $\beta(t)$, $\gamma(t)$,
  but the reasoning above still applies, and we can obtain the bounds
  illustrated in Fig.~\ref{fig:ode-ex-plot}.
  With further computations, it is also possible to discover and prove
  non-constant bounds on the solution.
\end{example}

\begin{remark}
  Note that the operator used in our proof is \emph{globally
  monotone}, directly allowing the use of order-theoretical fixed
  point theorems.
  This contrasts with the Picard-Lindelh{\"o}f operator
  often used in static analysis of hybrid
  systems, which uses instead a \emph{topological} fixed point
  theorem, and requires a restriction of the time horizon to get back
  to order theory.
\end{remark}




\section{Related Work}
\label{sec:related-work}
\textbf{Higher-order abstract interpretation}~\cite{cousot1994higher-short},
can be defined as the study of abstract domains whose
values are interpreted as functions and operators.
Its relevance is clear for the abstract interpretation of higher-order
functional languages, a topic first studied in the early 90s (e.g.
\cite{BurnSCP86,BurnFP91-short,HuntPhd91-short,cousot1994higher-short}),
and which has recently received renewed attention for practical
analysis of languages such as
\texttt{OCaml}~\cite{montagu-icfp2020-short,montagu-salto-prelim-ml23-short,Bautista-FMSD2024-short,valnet-ECOOP25-short},
where more refined value properties are explored.

Abstractions of higher-order objects are valuable
in other settings,
including static analysis of dynamical systems, such as simple loops
(e.g. for summarisation, acceleration, and in numerical filters) and
differential equations in hybrid systems. In dynamical systems,
non-linear properties quickly arise by iteration of simple (linear)
dynamics.
%
Sequence domains are also relevant to static analysis of data-flow
languages in synchronous programming.

For the analysis of filters written in \texttt{C},
\cite{Feret05-short} introduced an abstract domain of
\emph{arithmetico-geometric sequences}, which gave us the initial
inspiration for \boundariesbounds.
This domain abstracts sets of functions, but is intended to be used as
a classical ``non-relational'' value domain: variables in a loop are
(independently) bounded by a (non-linear) function of a local loop
counter. As in
\boundariesbounds, the finite representation of
non-linear functions is chosen to admit precise abstraction of simple
arithmetic and recursive structures, e.g. the operation \texttt{next},
which is similar to our \texttt{Push}.

For differential equations, the \emph{flowpipes}
of~\cite{GoubaultHSCC17-short,GoubaultCAV18} provide inner and outer
interval-valued approximations of functions. The interval version of
\boundariesbounds is similar to outer flowpipes in both spirit and
concretisations, though their
inner machinery and order theory differs.

\textbf{Non-linear abstract domains.} Few domains go beyond the
expressiveness of polyhedra. Notable exceptions include those
for polynomial
equalities~\cite{RodriguezKapurSas04-short,MullerOlm04-short},
polynomial \emph{inequalities} by reduction to
polyhedra~\cite{BRZ-polyineq-SAS05-short}, and the \emph{wedge}
domain~\cite{kincaid2018-shortest}.
Outside pure abstract interpretation, semialgebraic reasoning
techniques (e.g. cylindrical algebraic
decomposition~\cite{Collins75-CAD-short} or more tractable
optimisation-based sums-of-squares methods~\cite{RouxVS18-short})
have been used to reason on polynomial inequality invariants,
especially for continuous and hybrid systems.
Efficient domains with lower expressivity than general (semi)algebraic
sets have also been designed, e.g. the ellipsoid
domain~\cite{Feret04-short} and the conic extrapolation of higher
ellipsoids of~\cite{MV-conic-ellipsoids-CAV15-short}.
Finally, incomparable domains have been considered, such as tropical
polyhedra~\cite{allamigeonTropicalSAS08-short}, or the
arithmetico-geometric sequences of~\cite{Feret05-short}, which are a
particular case of exponential polynomials.

\textbf{Transfer function synthesis} is a growing family of techniques
enabled by the emergence of powerful decision procedures, quantifier
elimination techiques, as well as approaches for optimisation, search
and learning, among others. These techniques help design abstract
domains and improve their precision.
%
Important early work includes~\cite{GrafSaidi97-short} for
predicate-abstraction domains and~\cite{Reps2004-short} for domains with no
infinite ascending chains.
For numerical domains, the power of quantifier elimination and
optimisation techniques was demonstrated in a line of
work by Monniaux~\cite{Monniaux_POPL09-short,Monniaux_CAV09-shortest,Monniaux_CAV10-short,Monniaux_LMCS10-short}.
Strategies to avoid expensive quantifier elimination were discussed
in~\cite{BrauerKingLMCS12-short}, which also contains a recommended
introduction to the literature on the topic.
%
More recent
works~\cite{AmurthOOPSLA22-short,PasadoOOPSLA23-short,PruningOOPSLA24-short,USSTAD-preprint-25,ZixinHuangPhD25-chap5}
have explored many additional
strategies, e.g. based on search or
numerical optimisation, and inspired by approaches from program
synthesis or machine learning.

\textbf{Exact recurrence solvers.}
Centuries of work have created a large body of
knowledge on recurrence solving, 
with classical results including closed forms of \emph{C-recursive
sequences}~\cite{Petkovsek2013sketch-short},
or more recently,
hypergeometric solutions of \emph{P-recursive
sequences}~\cite{Petkovsek92}.
Several algorithms have been implemented
in CAS, based on mathematical frameworks such as
\emph{Difference Algebra}~\cite{Karr81-short,Levin08}, \emph{Finite
Calculus}~\cite{ConcreteMathematics}, or \emph{Generating
Functions}~\cite{Flajolet09-shortest}, mixed with template-based methods.
However, as explained in the introduction, these techniques have a
different focus than ours.
Beyond classical CAS, specialised tools such as
\texttt{PURRS}~\cite{BagnaraPZZ05} provide
approximate resolution of some non-linear, multivariate, and
divide-and-conquer equations.

\textbf{Applications of recurrence solving} to static analysis include
\emph{cost analysis}, \emph{loop summarisation/acceleration}, and more
generally, \emph{invariant synthesis and verification}. The reader is
invited to consult our references and previous contributions for more
information on these topics.

\section{Conclusions and Future Work}
\label{sec:conclusions}


Motivated by the inference of closed-form bounds for
recursively defined functions, we have presented both theoretical and
practical results on the order theory of lattices of functions, within
an equation-as-operator viewpoint.
In particular, we proposed a \emph{domain abstraction} mechanism,
which formalises the abstraction of a symbolic semantic equation by a
numerical functional equation, also also paves the way towards
applications to dimensionality reduction.
Most importantly, we introduced \emph{\boundariesbound domains} --
constraint-based abstract domains using predefined boundary functions
(e.g. polynomials, exponentials) -- enabling the inference of
\emph{non-linear numerical invariants}. A key convexity property
simplifies transfer function design and, in some cases, enables
automation.
We demonstrated
their application to
concrete domains of \emph{sequences}, and
sketched how they may be used for more general \emph{functions}.
%

This work complements our previous contributions, where we showed
that our results have multiple applications in program analysis, such
as static cost analysis, loop acceleration, and the analysis of
higher-order declarative languages.
%
%
Still, this paper is a stepping stone towards broader applications and
generalisation. We have provided theoretical constructions and
illustrative examples, laying the groundwork for future
implementations. Experimental work will be essential to refine these
ideas -- especially regarding 
widening operators and efficient representations of abstract
values. In particular, it would be valuable to develop tools based on
\boundariesbound domains for cost or value analysis of numerical
higher-order programs, and compare them with 
existing approaches.



Our investigations also raise several open questions.
%
%
To move beyond minimal operator languages with hardcoded recursive
structure, Section~\ref{subsec:piecewise-functions} highlights a
compelling problem: designing
practical techniques to compute ``best $\boundaries$-approximations of
$\boundaries$-cases by $\casesdomain$-parts''.
%
%
Another key challenge lies in the operator abstraction induced by
domain abstraction (Section~\ref{subsec:domain-abstr-galois}), which
is a priori non-constructive. Developing automated techniques to
select appropriate mappings $m$ and compute their abstractions --
especially across broader classes of $m$ -- is of central importance
in static cost analysis.
%
%
So far, we have focused on sets $\boundaries$ that form finitely
generated vector spaces. Extending the framework to richer families of
bounds is another natural direction.
%
%
Finally, as outlined in Section~\ref{subsec:continous-systems}, we
plan to explore applications to continuous systems. While this paper
focuses
on discrete domains to develop the theoretical foundations of
\boundariesbounds, their generality makes them a promising framework
for continuous domains, where the need for such abstractions
may
be even greater than in cost analysis.

\addcontentsline{toc}{section}{References}
\bibliographystyle{spmpsci}
\bibliography{\bibpath/clip,\bibpath/general}


\end{document}